\documentclass[10pt]{amsart}
\usepackage[margin=3.5cm]{geometry}
\usepackage{amssymb}
\usepackage{amsmath}
\usepackage{centernot}
\usepackage{graphicx} 
\usepackage[mathscr]{euscript}
\usepackage{enumerate}
\usepackage{xspace}
\usepackage{color}
\usepackage{enumitem}
\usepackage{amsthm}
\usepackage{dsfont}
\usepackage{mathrsfs}
\usepackage[scr=boondoxo]{mathalfa}
\usepackage{bbm}
\usepackage[colorlinks=true, linkcolor = blue, citecolor = blue]{hyperref}
\usepackage[numbers]{natbib}
\usepackage[backgroundcolor=white,bordercolor=red]{todonotes}
\DeclareMathAlphabet{\mathpzc}{OT1}{pzc}{m}{it}
\usepackage[nameinlink]{cleveref}

\begin{document}

\theoremstyle{plain}
\newtheorem{theorem}{Theorem}[section]
\newtheorem{lemma}[theorem]{Lemma}
\newtheorem{proposition}[theorem]{Proposition}
\newtheorem{corollary}[theorem]{Corollary}
\newtheorem{Ass}[theorem]{Assumption}

\theoremstyle{definition}
\newtheorem{discussion}[theorem]{Discussion}
\newtheorem{definition}[theorem]{Definition}
\newtheorem{remark}[theorem]{Remark}
\newtheorem{example}[theorem]{Example}
\newtheorem{condition}[theorem]{Condition}
\newtheorem{agreement}[theorem]{Agreement}
\newtheorem{SA}[theorem]{Standing Assumption}
\newtheorem*{observation}{Observations}
\newtheorem*{RL}{Comments on Related Literature}

\renewcommand{\chapterautorefname}{Chapter} 
\renewcommand{\sectionautorefname}{Section} 

\crefname{lemma}{lemma}{lemmas}
\Crefname{lemma}{Lemma}{Lemmata}
\crefname{corollary}{corollary}{corollaries}
\Crefname{corollary}{Corollary}{Corollaries}

\newcommand{\m}{\mathfrak{m}}
\newcommand{\mo}{\mathfrak{m}^\circ}
\newcommand{\mon}{\mathfrak{m}^{\circ, n}}
\newcommand{\tm}{\tilde{\mathfrak{m}}}
\newcommand{\M}{\tilde{\mathfrak{m}}}
\newcommand{\Mo}{\tilde{\mathfrak{m}}^\circ}
\newcommand{\Mon}{\tilde{\mathfrak{m}}^{\circ, n}}
\newcommand{\s}{\mathfrak{s}}
\newcommand{\ts}{\tilde{\mathfrak{s}}}
\renewcommand{\S}{S} 
\renewcommand{\v}{\mathfrak{v}}
\renewcommand{\u}{\mathfrak{u}}
\newcommand{\q}{\mathfrak{q}}
\def\stackrelboth#1#2#3{\mathrel{\mathop{#2}\limits^{#1}_{#3}}}
\renewcommand{\l}{\mathscr{l}}
\newcommand{\E}{\mathsf{E}} 
\renewcommand{\P}{\mathsf{P}}
\newcommand{\Po}{\mathds{P}^\circ}
\newcommand{\Pon}{\mathds{P}^{\circ, n}}
\newcommand{\Q}{\mathsf{Q}} 
\newcommand{\tP}{\tilde{\mathsf{P}}}
\newcommand{\tQ}{\tilde{\mathsf{Q}}}
\newcommand{\Qo}{\tilde{\mathds{P}}^\circ}
\newcommand{\Qon}{\tilde{\mathds{P}}^{\circ, n}}
\newcommand{\W}{\mathds{W}}
\newcommand{\on}{\operatorname}
\newcommand{\oU}{U}
\newcommand{\of}{[\hspace{-0.06cm}[}
\newcommand{\gs}{]\hspace{-0.06cm}]}
\newcommand{\ofr}{(\hspace{-0.09cm}(}
\newcommand{\gsr}{)\hspace{-0.09cm})}
\renewcommand{\emptyset}{\varnothing}

\renewcommand{\theequation}{\thesection.\arabic{equation}}
\numberwithin{equation}{section}

\newcommand{\1}{\mathds{1}}
\newcommand{\f}{\mathfrak{f}}
\newcommand{\g}{\mathfrak{g}}
\newcommand{\e}{\mathfrak{e}}
\renewcommand{\t}{T}
\newcommand{\bR}{\mathbb{R}}
\newcommand{\cF}{\mathcal{F}}
\newcommand{\cG}{\mathcal{G}}
\newcommand{\cA}{\mathcal{A}}
\newcommand{\x}{\mathfrak{z}}

\newcommand*{\Y}{Y}
\newcommand*{\U}{U}
\newcommand*{\Z}{Z}
\newcommand*{\B}{B}
\newcommand*{\V}{V}

\newcommand*{\ol}{\overline}
\newcommand*{\wt}{\widetilde}
\newcommand*{\canOmega}{\ol\Omega}
\newcommand*{\canF}{\overline{\mathcal{F}}\vphantom{()}}
\newcommand*{\canbfF}{\ol{\mathbf F} \vphantom{()}}
\newcommand*{\canY}{\ol\Y}
\newcommand*{\canU}{\ol\U}
\newcommand*{\canV}{\ol\V}
\newcommand*{\canZ}{\ol\Z}
\newcommand*{\cantheta}{\bar\theta}
\newcommand*{\wtY}{\wt\Y}
\newcommand*{\wtZ}{\wt\Z}

\newcommand*{\siQ}{{}^\sigma\hspace{-1pt}\Q}
\newcommand*{\siB}{{}^\sigma\mathbb B}
\newcommand*{\siS}{{}^\sigma\!\S}
\newcommand*{\sibfF}{{}^\sigma\mathbf F}
\newcommand*{\sicF}{{}^\sigma\!\cF}

\title[Criteria for the absence of arbitrage]{Criteria for the absence of arbitrage \\ in general diffusion markets}

\author[D. Criens]{David Criens}
\address{D. Criens - University of Freiburg, Ernst-Zermelo-Str. 1, 79104 Freiburg, Germany.}
\email{david.criens@stochastik.uni-freiburg.de}

\author[M. Urusov]{Mikhail Urusov}
\address{M. Urusov - University of Duisburg-Essen, Thea-Leymann-Str. 9, 45127 Essen, Germany.}
\email{mikhail.urusov@uni-due.de}


\makeatletter
\@namedef{subjclassname@2020}{\textup{2020} Mathematics Subject Classification}
\makeatother


\thanks{We thank Ioannis Karatzas for asking us about NUPBR for general diffusion markets, which initiated this research, and Yuri Kabanov for motivating us to extend our work into the direction of NA.
We are also very grateful to both of them for their inspiring feedback on the first version of this article.
Further, we thank the participants of the 16th Bachelier Colloquium on Financial Mathematics and Stochastic Calculus (M\'etabief, France, January 15--19, 2024), where this work has been presented, for fruitful discussions.
Last but definitely not least, we thank Martin Schweizer, the associate editor and two anonymous referees for numerous constructive comments and suggestions that helped us to improve the manuscript.}

\date{\today}

\allowdisplaybreaks

\begin{abstract}
We establish deterministic necessary and sufficient conditions for the no-arbitrage notions
NA (``no arbitrage''), NUPBR (``no unbounded profit with bounded risk'') and NFLVR (``no free lunch with vanishing risk'')
in general diffusion market models with finite and infinite time horizons.
These are single asset models whose (discounted) asset price process $\Y$ is a regular continuous strong Markov process that is also a semimartingale.
We further characterize the existence of an equivalent martingale measure in such models.
All deterministic criteria are provided in terms of the scale function and the speed measure of~$\Y$.
\end{abstract}

\maketitle

\frenchspacing
\pagestyle{myheadings}

\section{Introduction}
The question whether certain arbitrage opportunities exist for a given financial market model is of fundamental importance to develop a theory of finance that answers questions related to pricing and hedging of contingent claims and portfolio optimization (see, e.g., the monographs
\cite{DS2006,CJY,karatzaskardaras,KaraShre_MF,shir}).
On the other hand, the existence of certain arbitrage opportunities motivates concepts of asset price bubbles that also attracted much attention in recent years
(cf., e.g., \cite{protter_bubbles} for an overview).

\smallskip
A very flexible class of continuous time financial market models are general diffusion models, which are single asset models whose (discounted) asset price process is a one-dimensional path-continuous regular strong Markov process that is also a semimartingale. This class includes all It\^o diffusion models of the type 
\begin{align} \label{eq: intro SDE}
d Y_t = \mu (Y_t) dt + \sigma (Y_t) dW_t
\end{align}
under the famous Engelbert--Schmidt conditions.
 Further, it covers models with local time effects such as skewness and stickiness. Such effects can, for instance, be observed when companies got takeover offers (see the Introduction of \cite{criensurusov22} for more details). An important feature of these models is that they are characterized by two deterministic objects, the {\em scale function} and the {\em speed measure}, which we call the {\em diffusion characteristics}.
The purpose of this paper is to investigate several important notions of no-arbitrage for general diffusion models and to describe them in a deterministic manner through the diffusion characteristics.

\smallskip
Before we comment more specifically on our contributions, let us shortly discuss the no-arbitrage concepts under consideration. Mainly thanks to the groundbreaking work of Delbaen and Schachermayer \cite{DS}, the no-arbitrage condition NFLVR (``no free lunch with vanishing risk'') can nowadays be seen as the most classical concept for continuous time. Their FTAP (``fundamental theorem of asset pricing'')
establishes the equivalence between
NFLVR and the existence of an ELMM (``equivalent local martingale measure'').
We emphasize that the work of Delbaen and Schachermayer is by no means restricted to path-continuous frameworks that are under consideration in this paper, see \cite{DS2006} for an account of their work.

The notion NA (``no arbitrage'') is the classical concept in discrete time and it appeals through its meaningful economic interpretation.
We notice that, in continuous time, NA is understood as ``no arbitrage with admissible strategies'' (admissibility is needed to exclude for example doubling strategies).
One of the main reasons for introducing NFLVR is that the FTAP for discrete time NA does not transfer to continuous time.
While the existence of an ELMM is certainly sufficient for NA, it is not necessary.
For continuous semimartingale models,
Delbaen and Schachermayer~\cite{DS1995} proved the important fact that the existence of an ACLMM (``absolutely continuous local martingale measure'') is necessary but not sufficient for NA.
Building on this work of Delbaen and Schachermayer, Kabanov and Stricker~\cite{KabStr2005} and Strasser \cite{Strasser2005} established a complete FTAP for continuous time NA, which provides a necessary and sufficient condition that is related to the existence of ACLMMs for shifted market models.

The last notion we comment on is NUPBR (``no unbounded profit with bounded risk'').
Its importance was already noticed by Delbaen and Schachermayer~\cite{DS} who treated NUPBR in Section~3 of their paper~\cite{DS} without giving it a specific name. 
In a related context, NUPBR appeared in Kabanov's paper \cite{Kabanov1997} under the name ``BK property''.
Even earlier, Kabanov and Kramkov \cite{KabanovKramkov1994} investigated the closely related notion
NAA$_1$ (``no asymptotic arbitrage of the first kind''), although in a different context of large financial markets.
In general, NUPBR is weaker than NFLVR and it neither implies nor is implied by NA. More specifically, Delbaen and Schachermayer~\cite{DS} proved that 
\begin{equation}\label{eq:270724a1}
\text{NFLVR} \quad \Longleftrightarrow \quad \text{NA and NUPBR}.
\end{equation}
Later, Karatzas and Kardaras \cite{karatzaskardaras07} identified NUPBR as the minimal concept of no arbitrage required to develop a theory of finance that includes hedging of contingent claims and portfolio optimization
(also the name ``NUPBR'' first appeared in \cite{karatzaskardaras07}).
We further highlight the works of Hulley and Schweizer \cite{HS10} and Takaoka and Schweizer \cite{takaoka14} on probabilistic characterizations of NUPBR.
In this realm, we also mention the probabilistic characterization of the notion NA$_1$ (``no arbitrage of the first kind'') by Kardaras \cite{Kardaras2012}, as NA$_1$ turns out to be equivalent to NUPBR (see Kardaras \cite{Kardaras2010}).
Moreover, in the present context of a single financial market, all three aforementioned notions NAA$_1$, NA$_1$ and NUPBR are equivalent (see Lemma~A.1 in Kabanov, Kardaras and Song \cite{KabanovKardarasSong2016}) and we, therefore, only speak about NUPBR in what follows.
For the behavior of NUPBR under filtration shrinking we refer to Kardaras and Ruf \cite{kardarasruf}.
NUPBR under filtration enlargements was studied by Acciaio, Fontana and Kardaras \cite{AFK_16}, Aksamit, Choulli, Deng and Jeanblanc \cite{ACDJ} and Aksamit, Choulli and Jeanblanc \cite{ACJ}.
A profound discussion of many facets of NUPBR can also be found in Chapter~2 of the recent monograph \cite{karatzaskardaras} by Karatzas and Kardaras
(there NUPBR is called \emph{market viability}).
For an overview of various notions of no arbitrage and their interrelations, the interested reader is also referred to the paper of Fontana~\cite{Fontana2015} and the Ph.D. thesis of Hulley \cite[Chapter~1]{Hulley2009}.

\smallskip
The main results in this paper are deterministic characterizations of NA, NUPBR and NFLVR for general diffusion market models (with finite and infinite time horizons) in terms of their diffusion characteristics. Furthermore, we investigate the existence of EMMs (``equivalent martingale measures''), which can be viewed as a very strong no-arbitrage condition (see \cite{cherny}).
It is worth mentioning that our results for NA, NUPBR and NFLVR also propagate to ``integrated general diffusion markets'' where the asset price process \(\S\) has the form
$$
\S=\S_0+\int_0^\cdot \sigma_u\,d\Y_u, \quad \sigma \not = 0, 
$$
for a general diffusion semimartingale \(\Y\).
In particular, this means that our results cover stochastic exponentials of diffusion models, which is a widespread modeling approach.

In the following, we highlight some facets of our main contributions in more detail.
As a general structural condition, we prove that each of the notions NA, NUPBR and NFLVR forces the scale function to be continuously differentiable with a strictly positive
absolutely continuous derivative. Broadly speaking, this means that in the absence of arbitrage the scale function has to be of the same type as for an It\^o diffusion. In particular, this shows that skewness effects always imply arbitrage. This does not apply to stickiness effects that might or might not lead to arbitrage.
The observation that NFLVR entails the scale function to be continuously differentiable with strictly positive absolutely continuous derivative is not entirely new, since it was already proved in our previous paper \cite{criensurusov22} for certain canonical diffusion settings. We think it is quite interesting to observe that even the weaker notions NA and NUPBR are sufficient for this structural property.

Like in the most general semimartingale market models as investigated in~\cite{DS}, also in our diffusion framework, NA and NUPBR are in general position, i.e., neither of them implies the other, and 
we cannot say more than~\eqref{eq:270724a1}.
However, our results identify three general subsettings of our framework with finite time horizon
(only depending on the underlying state space and the boundary behavior)
where
\begin{align*}
\text{either}
\quad
&
\text{NFLVR}
\;\;\Longleftrightarrow\;\;
\text{NA}
\;\;\Longleftrightarrow\;\;
\text{NUPBR},
\\
\text{or}
\quad
&
\text{NFLVR}
\;\;\Longleftrightarrow\;\;
\text{NA}
\;\;\;\Longrightarrow\;\;
\textup{NUPBR},
\\
\text{or}
\quad
&
\text{NFLVR}
\;\;\Longleftrightarrow\;\;
\text{NUPBR}
\;\;\;\Longrightarrow\;\;\;
\textup{NA}.
\end{align*}
For instance, the first subsetting includes all general diffusion models that are regular on \(\bR\), i.e., reach all real numbers in finite time with positive probability.

Further, it is worth mentioning that we do not require the underlying filtration to be generated by the asset price process $\Y$. In fact, we allow for larger filtrations. 
The only property we need is that $\Y$ has to be strongly Markov relative to the filtration.
It is natural to ask whether counterexamples like the one discussed in \cite{DS1998counter} are possible in our framework, i.e., whether, under NA (or under NFLVR), it is possible that the ``minimal martingale density''
(cf. the process $\widehat Z$ on p.~42 in \cite{HS10})
is a strict local martingale.
Phrased differently, this means that the minimal ACLMM (resp. ELMM) fails to exists although some ACLMM (resp. ELMM) exists.
We refer to \cite{DS1998counter} for more detail on this question and on what can happen in general.
It turns out that such counterexamples are impossible in our framework, i.e., whenever an ACLMM (resp. ELMM) exists, the same is true for the minimal version.

\smallskip
Lastly, we comment on literature closely related to our main results and discuss some aspects of our proofs.
To the best of our knowledge, even in the one-dimensioinal It\^o diffusion framework that is given by \eqref{eq: intro SDE},
the literature contains no deterministic characterization of the notions NA and NUPBR.
By contrast, there are many papers on the characterization of NFLVR and the existence of EMMs. In fact, for general diffusion models in a canonical framework, a deterministic characterization of NFLVR has been established in our previous paper \cite{criensurusov22}.
The articles of Criens~\cite{Criens2018,criens20}, Delbaen and Shirakawa~\cite{DelbaenShirakawa2002} and Mijatovi\'c and Urusov \cite{MU12b}
appear to be closest to the present paper and \cite{criensurusov22} in the sense that they also aim for deterministic conditions. In all of these papers, the asset price process is some sort of It\^o process with non-vanishing volatility, which excludes, for instance, sticky points in the interior of the state space. 
On a technical level, the proofs in these papers rely on the idea of studying the true martingale property and positivity of certain stochastic exponentials that are natural candidate density processes. 
In case the scale function is continuously differentiable with strictly positive
absolutely continuous
derivative, we can adapt this strategy, construct a tractable candidate density process as in the It\^o diffusion setting and investigate its properties to understand NUPBR. For NA, we use a different approach that is based on the probabilistic characterizations by Kabanov and Stricker~\cite{KabStr2005} and Strasser \cite{Strasser2005}. Namely, we investigate the existence of ACLMMs for time-shifted market models. Again, it is important to understand first the structure of the scale function. In both cases, to prove the desired structure, we use arguments based on the strong Markov property, a local change of measure up to a positive predictable time, martingale problem techniques and results on absolute continuity and equivalence of general diffusions that we recently established in our previous paper \cite{criensurusov22}.
We refer to Section~\ref{sec: outline proof} for more detailed comments on this strategy.
Our characterization of NFLVR is a direct consequence of our results for NA and NUPBR. As we think that NFLVR is of particular interest, we also provide short direct proofs of our characterizations of NFLVR that transfer results from \cite{criensurusov22} to the more general setting of this paper.
Related to this, the question when a non-negative local martingale is a true martingale was recently studied by Desmettre, Leobacher and Rogers~\cite{desmettre} in the same general diffusion framework that is considered in this paper. 
Although the results from~\cite{desmettre} are quite general, there are diffusion models whose candidate densities cannot be brought to the form studied in~\cite{desmettre}. In Section~\ref{sec: outline proof}, we comment on this point in more detail.
Finally, our characterizations of the existence of EMMs are based on Kotani's \cite{kotani} fine result on the martingale property of general diffusions on natural scale.

\smallskip
The paper is organized as follows. In Section~\ref{sec:no-arbitrage} we introduce and discuss the notions NUPBR, NA and NFLVR and the corresponding FTAPs. Our setting and the main results are presented in Section~\ref{sec:setting}. An outline and comments on our proofs can be found in Section~\ref{sec: outline proof}.
Finally, our detailed proofs are presented in Section~\ref{sec:proofs}. To provide the possibility for a linear reading, we recall each main theorem before its proof.

\section{No-arbitrage notions}\label{sec:no-arbitrage}

Throughout this paper, we consider a finite or infinite deterministic time horizon \(T \in (0, \infty]\).

\begin{agreement}\label{agr:050224a1}
In case $T=\infty$ we understand the interval $[0,T]$ as $\bR_+ = [0, \infty)$ and read expressions like ``$t \in [0, T]$'' as ``$t\in\bR_+$''.
\end{agreement}

	A pair \((\mathbb{B}, \S)\) is said to be a \emph{financial market} if \(\mathbb{B} = (\Omega, \mathcal{F}, \mathbf{F}= (\mathcal{F}_t)_{t \in [0, T]}, \P)\) is a filtered probability space with a right-continuous filtration \(\mathbf{F}\) that supports a continuous real-valued \(\mathbf{F}\)-\(\P\)-semimartingale \(\S= (\S_t)_{t \in [0, T]}\).
For a financial market \((\mathbb{B}, \S)\), let \(L (\mathbb{B}, \S)\) be the set of all \(\mathbf{F}\)-predictable real-valued processes \(H = (H_t)_{t \in [0, T]}\) which are integrable
w.r.t. 
\(\S\), i.e., that satisfy \(\P\)-a.s., for all \(t \in [0, T]\),
\[
\int_0^t |H_s|\, d [\on{Var}(A)]_s < \infty
\quad\text{and}\quad
\int_0^t H^2_s\,d \langle M, M\rangle_s < \infty,
\]
where $\S=\S_0+M+A$ is the canonical decomposition of $\S$, i.e., $M$ is a continuous \(\mathbf{F}\)-\(\P\)-local martingale with $M_0=0$ and $A$ is a continuous \(\mathbf{F}\)-adapted process of finite variation with $A_0=0$.
In our financial context, the elements of \(L (\mathbb{B}, \S)\) are called \emph{trading strategies}. To ease our presentation, we write 
\[
V^H \triangleq \int_0^\cdot H_s d \S_s
\]
for the {\em value process} associated to the trading strategy \(H \in L (\mathbb{B}, \S)\).
\begin{definition}
Let \((\mathbb{B}, \S)\) be a financial market.
For \(c \in \bR_+\), a trading strategy \(H \in L (\mathbb{B}, \S)\) is called {\em \(c\)-admissible} if \(\P\)-a.s. \(V^H \geq - c\). Further, we call a trading strategy {\em admissible} if it is \(c\)-admissible for some \(c \in \bR_+\).
\end{definition}

We define $K_1$ to be the set of all contingent claims attainable from zero initial capital via a $1$-admissible strategy, i.e.,
\begin{align*}
K_1
\triangleq \Big\{ V^H_T
 \colon H &\text{ is \(1\)-admissible and, if \(T = \infty\), then }
V^H_\infty\triangleq\lim_{t\to\infty}V^H_t
\text{ exists \(\P\)-a.s.} \Big\}.
\end{align*}
Further, let $K$ be the set of all contingent claims attainable from zero initial capital via some admissible strategy, i.e.,
$$
K=\bigcup_{n\in\mathbb N}(n K_1),
$$
and let \(C\) be the set of all essentially bounded random variables that are dominated by claims in \(K\), i.e.,
\[
C \triangleq \Big\{ g \in L^\infty \colon \exists f \in K \text{ such that } g \leq f\ \text{\(\P\)-a.s.} \Big\}.
\]

Next, we recall the definitions of the basic no-arbitrage notion
\emph{NA}
and the notions
\emph{no unbounded profit with bounded risk (NUPBR)}
and
\emph{no free lunch with vanishing risk (NFLVR)}.

\begin{definition}[NA]\label{def:100224a1}
	Let $(\mathbb B,\S)$ be a financial market.
	We say that a strategy $H\in L(\mathbb B,\S)$ realizes \emph{arbitrage} if
	\begin{enumerate}
		\item[\textup{(i)}]
		$H$ is admissible,
		
		\item[\textup{(ii)}]
		in case $T=\infty$,
		$V^H_\infty = \lim_{t \to \infty} V^H_t$ exists \(\P\)-a.s.,
		
		\item[\textup{(iii)}]
		$\P(V^H_T\ge0)=1$ and $\P(V^H_T>0)>0$.
	\end{enumerate}
	We say that the financial market $(\mathbb B,\S)$ satisfies \emph{NA} if there is no strategy realizing arbitrage.
\end{definition}

We remark that, in continuous-time models, it is necessary to consider admissible strategies, as non-admissible arbitrages (i.e., strategies satisfying only (ii)--(iii)) exist practically in any interesting model (e.g., in the classical Black-Scholes model).

\begin{definition}[NUPBR]\label{def:100224a2}
The financial market \((\mathbb{B}, \S)\) satisfies \emph{NUPBR} if the set $K_1$ is bounded in \(\P\)-probability, i.e., 
	\[
	\lim_{m \to \infty} \sup_{V \in K_1} \P (V > m) = 0.
	\]
\end{definition}

\begin{definition}[NFLVR]\label{def:100224a3}
We say that \emph{NFLVR} holds in the market \((\mathbb{B}, \S)\) if
		\[
		\overline{C} \cap L^\infty_+ = \{0\}, 
		\]
		where \(\overline{C}\) denotes the closure of \(C\) in \(L^\infty\) w.r.t. the norm topology and \(L^\infty_+\) denotes the cone of nonnegative random variables in \(L^\infty\).
\end{definition}

\begin{remark}
For the sake of comparison between NFLVR and NA, it is worth noting that NA is, clearly, equivalent to $C \cap L^\infty_+ = \{0\}$.
\end{remark}

In the following we recall the fundamental stochastic characterizations of the previous no-arbitrage concepts. 

\begin{definition}
We call a probability measure \(\Q\) on \((\Omega, \mathcal{F})\)
an \emph{absolutely continuous local martingale measure}
(resp., an \emph{equivalent local martingale measure})
for the market \((\mathbb{B}, \S)\) if \(\Q\ll \P\) (resp., \(\Q\sim \P\))
and \(\S\) is an \(\mathbf{F}\)-\(\Q\)-local martingale.
We use the abbreviation \emph{ACLMM} (resp., \emph{ELMM}) in the following.
\end{definition}

\begin{agreement}\label{agr:030724a1}
To clarify our terminology, for an adapted c\`adl\`ag process $M=(M_t)_{t\in[0,T]}$ (where $M_0$ is allowed to be non-integrable), we say that $M$ is an \emph{$\mathbf F$-$\P$-local martingale} when $M-M_0$ is an $\mathbf F$-$\P$-local martingale starting from zero.
Furthermore, under a \emph{localizing sequence} for $M$ we understand any sequence $(\tau_n)_{n=1}^\infty$ of $\mathbf F$-stopping times such that $\tau_n\nearrow\infty$ $\P$-a.s., $n\to\infty$, and all stopped processes $M^{\tau_n}-M_0$ are $\mathbf F$-$\P$-martingales.
\end{agreement}

\begin{definition}\label{def:050224a1}
We say that a process \(Z = (Z_t)_{t \in [0, T]}\) is a
\emph{strict martingale density (SMD)}
for the market \((\mathbb{B}, \S)\) if it is a strictly positive c\`adl\`ag
\(\mathbf{F}\)-adapted process with \(\Z_0 = 1\) such that \(\Z\) and \(\Z\S\) are \(\mathbf{F}\)-\(\P\)-local martingales.
\end{definition}

In accordance with Agreement~\ref{agr:050224a1}, in case of an infinite time horizon $T=\infty$, Definition~\ref{def:050224a1} does not rely on the terminal value $\Z_\infty$. In particular, $\Z_\infty$ is not asked to be strictly positive, see, however, Theorem~\ref{theo: FTAP} and Remark~\ref{rem:080224a1} below.

\begin{remark}
	By virtue of \cite[Proposition~III.3.8]{JS}, the density process of an ELMM is an~SMD. 
\end{remark}

The following result is the seminal fundamental theorem of asset pricing for NFLVR by Delbaen and Schachermayer~\cite{DS}.
We emphasize that it holds in this form both for $T<\infty$ and~$T=\infty$.

\begin{theorem}[FTAP for NFLVR]\label{theo: FTAP NFLVR}
For a financial market \((\mathbb{B}, \S)\), 
$$
\text{NFLVR}
\quad\Longleftrightarrow\quad
\text{there exists an ELMM}.
$$
\end{theorem}

Next, we discuss the NA condition.
A complete stochastic characterization of NA for the case of a finite time horizon within a Brownian setting was established by Levental and Skorohod~\cite{LevSko1995}.
A necessary condition for NA for general continuous prices processes was given by Delbaen and Schachermayer \cite{DS1995}. We recall it in the following theorem.

\begin{theorem}\label{theo: FTAP NA}
For a financial market \((\mathbb{B}, \S)\),
$$
\text{NA}
\quad\Longrightarrow\quad
\text{there exists an ACLMM}.
$$
\end{theorem}

We, again, emphasize that Theorem~\ref{theo: FTAP NA} holds in this form both for $T<\infty$ and $T=\infty$.
For a finite time horizon, building upon \cite{DS1995}, Kabanov and Stricker \cite{KabStr2005} and Strasser \cite{Strasser2005} provided even necessary and sufficient conditions for NA. We recall them in the next theorem.

\begin{theorem}[FTAP for NA, $T<\infty$]\label{th:210324a1}
Consider a financial market $(\mathbb B,\S)$ with $T<\infty$.
The following are equivalent:
\begin{enumerate}
\item[\textup{(i)}]
The market satisfies NA.

\item[\textup{(ii)}]
For every $\mathbf F$-stopping time $\sigma\le T$ there exists an ACLMM $\siQ$ for the market $(\siB,\siS)$ with
$$
\siQ\sim\P\text{ on }\cF_\sigma,
$$
where $\siB\triangleq(\Omega,\cF,(\sicF_t)_{t\in[0,T]},\P)$,
$\sicF_t\triangleq\cF_{(\sigma+t)\wedge T}$,
$\siS\triangleq(\siS_t)_{t\in[0,T]}$,
$\siS_t\triangleq\S_{(\sigma+t)\wedge T}$.
\end{enumerate}
\end{theorem}

We proceed with a characterization of NUPBR.
For a finite time horizon \(T < \infty\),
the following theorem was proved by Choulli and Stricker \cite{ChoulliStricker}.
In case of the infinite time horizon \(T = \infty\),
the result can be deduced from a more general theorem by Karatzas and Kardaras \cite{karatzaskardaras07} combined with a statement from Hulley and Schweizer \cite{HS10} (the latter being applied for each finite time horizon).
For completeness, we outline a short argument, using terminology from \cite{HS10,karatzaskardaras07}.

\begin{theorem}[FTAP for NUPBR]\label{theo: FTAP}
	For a financial market \((\mathbb{B}, \S)\), the following are equivalent:
	\begin{enumerate}
		\item[\textup{(i)}] The market satisfies the NUPBR condition.
		\item[\textup{(ii)}] There exists a SMD \(Z\) for the market and, if \(T = \infty\), then \(\P\)-a.s. \(Z_\infty \triangleq \lim_{t \to \infty} Z_t > 0\).\footnote{The limit \(\lim_{t \to \infty} Z_t\) exists a.s. by the supermartingale convergence theorem.}
	\end{enumerate}
\end{theorem}

\begin{proof}
For \(T < \infty\), the theorem is a direct consequence of \cite[Theorem 2.9]{ChoulliStricker}. We now discuss the case \(T = \infty\). The implication (ii) \(\Rightarrow\) (i) follows directly from the implication (2) $\Rightarrow$ (3) in \cite[Theorem 4.12]{karatzaskardaras07}. Suppose that (i) holds.
Then, by the implication (3) $\Rightarrow$ (1) in \cite[Theorem~4.12]{karatzaskardaras07}, the num{\'e}raire portfolio \(\rho\) exists and its value process \(V^\rho\) satisfies \(\P\)-a.s. \(V^\rho_\infty < \infty\). As \(S\) has continuous paths, it follows from \cite[Theorem~7]{HS10} applied for each finite time horizon that \(V^\rho = 1/\widehat{Z}\), where \(\widehat{Z}\) denotes the \emph{minimal martingale density} as given on p.~42 in \cite{HS10}. The process \(\widehat{Z}\) is a SMD and \(\P\)-a.s. \(\widehat{Z}_\infty > 0\) follows from the fact that \(\P\)-a.s. \(V^\rho_\infty < \infty\). Consequently, (ii) holds.
\end{proof}

The literature contains also other important characterizations of NUPBR.
One of them is the famous {\em structure condition} that we recall in the following theorem
for the case $T<\infty$,
see \cite[Theorem~1, Proposition~2]{Schweizer1995}, \cite[Theorem 2.9]{ChoulliStricker}
or \cite[Theorem~7]{HS10}
(cf. also \cite[Remark~1, p.~42]{HS10} for comments on the terminology).

\begin{theorem}[Structure condition, $T<\infty$]\label{theo: SC}
Take a financial market \((\mathbb{B}, S)\) with \(T < \infty\) and let \(S = S_0 + M + A\),
$M_0=A_0=0$,
be the canonical decomposition of the
continuous
semimartingale~\(S\)
($M$ is the local martingale part, $A$ is the finite variation part of~$S$).
Then, NUPBR holds if and only if there exists a real-valued predictable process \(\lambda\) such that a.s.
	\[
	A = \int_0^\cdot \lambda_s\, d \langle M, M\rangle_s
	\quad\text{and}\quad
	\int_0^T \lambda^2_s\, d \langle M, M\rangle_s < \infty.
	\] 
\end{theorem}

\begin{remark}
The structure condition shows that NUPBR for finite time horizons is fully determined by the local semimartingale characteristics of the asset price process. As pointed out in \cite[Example 4.7]{karatzaskardaras07}, this is not the case for NA and NFLVR, being one of the fundamental differences between these notions. 
\end{remark}

We, finally, recall the following result, which follows,
both for $T<\infty$ and for $T=\infty$,
from Corollary~3.8 in \cite{DS}, see also \cite[Theorem 1.3]{DS1995}.

\begin{theorem}\label{th:080224a1}
For a financial market \((\mathbb{B}, \S)\),
$$
\text{NFLVR}
\quad\Longleftrightarrow\quad
\text{NA and NUPBR}.
$$
\end{theorem}

\begin{remark}[Infinite time horizon vs. all finite time horizons]\label{rem:080224a1}
As a referee kindly pointed out, next to classical notions for the infinite time horizon \(T = \infty\)
(we call them ``global'' versions in this remark),
the literature also contains ``local'' versions, meaning that the no-arbitrage condition holds for \emph{all finite} time horizons.
For instance, in the context of NUPBR such a concept is used by
Kardaras \cite{Kardaras2014} and in the recent monograph \cite{karatzaskardaras} by Karatzas and Kardaras.
On the contrary, our terminology in the context of NUPBR for $T=\infty$ is consistent with Karatzas and Kardaras \cite{karatzaskardaras07}.
For a more detailed discussion of the difference between the global and local versions of NUPBR we refer to B{\'a}lint and Schweizer \cite{balint_schweizer_2020,balint_schweizer}.

Evidently, for NA, NUPBR and NFLVR, their global versions entail their localizations. The converse is not true as the counterexample from \cite[Remark~5.3]{balint_schweizer_2020} illustrates.
The difference is also visible from the general perspective of the above theorems. For example, the difference between

(a)
NUPBR for $T=\infty$ and

(b)
NUPBR for all finite time horizons

\noindent
can be understood through condition~(ii) of Theorem~\ref{theo: FTAP} because the requirement $Z_\infty>0$ is needed for~(a) but not for~(b), 
see also the comment after Definition~\ref{def:050224a1}. 
For further discussions in this direction, we also refer to Example~\ref{ex:100224a1}, the subsequent discussion and Examples~\ref{ex:140324a2} and~\ref{ex:140324a1} below.
\end{remark}

\section{Setting and main results}\label{sec:setting}

In the following, we investigate a financial market driven by a regular continuous strong Markov process $\Y$.
For brevity, we use the term \emph{general diffusion} (and sometimes simply \emph{diffusion}) as a synonym of ``regular continuous strong Markov process''.
A quite complete overview of the theory for general diffusions can be found in the monograph \cite{itokean74} by It\^{o} and McKean.
Shorter introductions are given in \cite{breiman1968probability,kallenberg,RY,RW2}.
For a condensed overview we refer either to Chapter~2 of the book \cite{borodin_salminen} by Borodin and Salminen or to Section~2.2 of our previous paper \cite{criensurusov22}. 

As the concepts of scale and speed are crucial for our main results, we recall some facts about these concept without going too much into detail.
We take as state space \(J \subset \bR\) a bounded or unbounded, closed, open or half-open interval. A scale function is a strictly increasing continuous function $\s\colon J\to\bR$
and a speed measure is a measure $\m$ on $(J,\mathcal B(J))$ that satisfies
$\m([a,b])\in(0,\infty)$ for all $a<b$ in $J^\circ$, where \(J^\circ\) denotes the interior of~\(J\). Set
$$
l\triangleq\inf J\in[-\infty,\infty)
\quad\text{and}\quad
r\triangleq\sup J\in(-\infty,\infty].
$$
The values $\s(l)$ and $\s(r)$ are defined by continuity (they can be infinite).
We also remark that the speed measure can be infinite near $l$ and $r$, and that the values $\m(\{l\})$ and \(\m (\{r\})\) can be anything in $[0,\infty]$ provided $l\in J$ and $r\in J$, respectively.

We are in a position to explain our financial framework.
Recall that we always consider a
(finite or infinite)
deterministic time horizon \(T \in (0, \infty]\) and that we use Agreement~\ref{agr:050224a1}.

Let \(\mathbb{B} = (\Omega, \cF, \mathbf{F} = (\cF_t)_{t \in [0, T]}, \P)\) be a filtered probability space with a right-continuous filtration that supports a regular continuous strong Markov process \(\Y = (\Y_t)_{t \in [0, T]}\) with state space \(J\), scale function \(\s\), speed measure \(\m\) and deterministic starting value \(x_0 \in J^\circ\). 
In the above context, the strong Markov property refers to the filtration~\(\mathbf{F}\).

We also assume that \(\Y\) is a semimartingale on \(\mathbb{B}\).\footnote{This is not automatically true in the general diffusion setting, as, for instance, if \(B\) is a Brownian motion, then $\sqrt{|B|}$ is \emph{not} a semimartingale. From the viewpoint of mathematical finance, this assumption is very natural. Lastly, we stress that the semimartingale property of \(\Y\) is solely a property of the scale function \(\s\). For a detailed discussion we refer to \cite[Section~5]{CinJPrSha}.}
Clearly, \((\mathbb{B}, \Y)\) is a financial market in the sense defined in Section~\ref{sec:no-arbitrage}, which allows us to investigate NA, NUPBR and NFLVR for this market.

\medskip
In order to formulate our results we first introduce several conditions.
We start by recalling Feller's test for explosions (\cite[Proposition 16.43]{breiman1968probability}).
It states that a
(finite or infinite)
boundary point $b\in\{l,r\}$ is accessible for the diffusion $\Y$ (that is, \(b\in J\)) if and only if
\begin{equation}\label{eq:170323a1}
|\s (b)| < \infty \quad \text{ and } \quad \int_B | \s (b) - \s (x) | \m (dx) < \infty
\end{equation}
for some (equivalently, for every) open interval \(B \subsetneq J^\circ\) with \(b\) as endpoint.
Consequently, $b\in\{l,r\}$ is inaccessible for the diffusion $Y$ (that is, \(b\notin J\)) if and only if either
\begin{equation}\label{eq:170323a2}
|\s (b)| = \infty
\end{equation}
or
\begin{equation}\label{eq:170323a3}
|\s (b)| < \infty \quad \text{ and } \quad \int_B | \s (b) - \s (x) | \m (dx) = \infty
\end{equation}
for some (equivalently, for every) open interval \(B \subsetneq J^\circ\) with \(b\) as endpoint.

\begin{remark}\label{rem:170323a1}
What is hidden in our requirement $J\subset\bR$ is that
we need to specify the scale function $\s$ and the speed measure $\m$ in such a way that either \eqref{eq:170323a2} or~\eqref{eq:170323a3} hold for infinite boundary points of~$J$.\footnote{\label{ft:040324a1}A more general setting would be to consider intervals $J\subset[-\infty,\infty]$ as state spaces, i.e., to allow also infinite boundaries to be accessible.
While this is natural when studying general diffusions (e.g., think about the SDE $d\Y_t=\Y_t^2\,dt+dW_t$, driven by a Brownian motion $W$, in which case $\infty$ is an accessible boundary), we need the setting to be compatible with financial modeling, which means that $\Y$ should be a semimartingale. This already forces infinite boundary points to be inaccessible and we come to the necessity to require $J\subset\bR$.}
\end{remark}

Another condition we need is the following.

\begin{condition}\label{cond:050324a1}
There exists a Borel function \(\beta \colon J^\circ \to \bR\) such that
	\begin{align} 
	\qquad\beta^2 \in L^1_\textup{loc}(J^\circ), \label{eq: nflvr1}
	\end{align}
	and, up to increasing affine transformations,
	\begin{align} \label{eq: nflvrS}
	\s (x) = \int^x \exp \Big\{ \int^y \beta (z) \,dz \Big\} \, dy, \quad x \in J^\circ.
	\end{align}
\end{condition}

\begin{condition}\label{cond:170623a1}
Condition~\ref{cond:050324a1} holds and every finite boundary point $b\in\{l,r\}\cap\bR$ is either inaccessible or absorbing for~$\Y$.
\end{condition}

\begin{remark}
Phrased differently, the last requirement of Condition~\ref{cond:170623a1} means that every finite boundary point $b\in\{l,r\}\cap\bR$ is neither instantaneously reflecting nor slowly reflecting for~$Y$.
We will see below that each of the no-arbitrage notions NA, NUPBR and NFLVR excludes any kind of reflecting boundaries. We refer to \cite{BucknerDowdHulley2024} for a recent discussion of reflecting boundaries in the context of weak notions of arbitrage.
Notice that infinite boundary points are inaccessible in our setting (recall Footnote~\ref{ft:040324a1}). Hence, they need not to be treated explicitly in Condition~\ref{cond:170623a1}.
\end{remark}

In case Condition~\ref{cond:170623a1} holds, for a finite boundary point $b\in\{l,r\}\cap\bR$, the following two integrability conditions will be of fundamental importance in our main results:
\begin{align}
&\int_B |x - b| [\beta (x)]^2\,dx<\infty,
\label{eq: b good}\\[1mm]
&\int_B |x - b| \s' (x)\,\m(dx)=\infty,
\label{eq: very good add cond}
\end{align}
for some (equivalently, for every) open interval \(B \subsetneq J^\circ\) with \(b\) as endpoint
(this refers both to~\eqref{eq: b good} and to~\eqref{eq: very good add cond}).

To give an idea of the above conditions, let us discuss the above conditions for the important example of It\^o diffusion market models.

\begin{example}[It\^{o} diffusion market]\label{ex:090224a1}
Recall that $J^\circ=(l,r)\subset\bR$.
Take two Borel functions $\mu \colon (l, r) \to \bR$ and $\sigma \colon (l, r) \to \bR$ satisfying the so-called Engelbert--Schmidt conditions
\begin{equation} \label{eq: ESC}
	\begin{split}
\sigma(x)\ne0\;\;\forall \, x\in J^\circ,
\\
\frac1{\sigma^2},\frac{\mu}{\sigma^2}\in L^1_\textup{loc}(J^\circ).
\end{split}
\end{equation}
It is well-known that the (Brownian) SDE
\begin{equation}\label{eq:090224a1}
d\Y_t=\mu(\Y_t)\,dt+\sigma(\Y_t)\,dW_t,\quad\Y_0=x_0 \in J^\circ,
\end{equation}
has a unique in law weak solution that possibly reaches the boundary points \(\{l, r\}\) in finite time (see \cite{ES1991} or \cite[Theorem~5.5.15 and Section~5.5.C]{KaraShre}).
We stipulate that the solution process gets absorbed in the boundaries that are reached in finite time.
Whether a boundary point is accessible or inaccessible is determined
via $\mu$ and $\sigma$
by Feller's test for explosion as given in \cite[Theorem~5.5.29]{KaraShre}. 
Notice that we need to specify $\mu$ and $\sigma$ in such a way that infinite boundaries are inaccessible (recall Remark~\ref{rem:170323a1}).

A solution process $\Y$ is a regular continuous strong Markov process with scale function
\begin{equation}\label{eq:090224a2}
\s (x) = \int^x \exp \Big\{ -\int^y \frac{2\mu (z)}{\sigma^2 (z)}\, dz \Big\} \, dy, \quad x \in J^\circ, 
\end{equation}
and speed measure 
\begin{equation}\label{eq:090224a3}
\m(dx)=\frac{dx}{\s' (x)\sigma^2(x)} \text{ on } \mathcal{B} (J^\circ), \quad \m (\{b\}) = \infty \text{ for an accessible boundary point \(b\)}.
\end{equation}
With these objects at hand, we remark that Feller's test from \cite[Theorem~5.5.29]{KaraShre} coincides with
\eqref{eq:170323a1}--\eqref{eq:170323a3} above.

Let us now comment on the above conditions.
It is evident that the scale function \(\s\) is of the form~\eqref{eq: nflvrS} with \(\beta = - 2\mu / \sigma^2\). Consequently, Condition~\ref{cond:050324a1} holds if and only if \(\mu^2 / \sigma^4 \in L^1_\textup{loc} (J^\circ)\). Furthermore, Condition~\ref{cond:170623a1} coincides with Condition~\ref{cond:050324a1}
because accessible boundaries are stipulated to be absorbing.

In the presence of accessible boundaries, solution processes to \eqref{eq:090224a1} are not necessarily semimartingales.
Indeed, the SDE~\eqref{eq:090224a1} drives the process $\Y$ only till an accessible boundary is hit (cf. \cite[Definition 5.5.20]{KaraShre}),
and at this time the semimartingale property can get lost,
see \cite[Section~4]{MU2015} for counterexamples.
In this regard, \cite[Corollary 3.6]{MU2015} provides necessary and sufficient condition for the semimartingale property of \(\Y\) that only depends on the coefficients \(\mu\) and \(\sigma\). For the reader's convenience, we recall this result:
$\Y$ is a semimartingale if and only if
	\begin{enumerate}
		\item[\textup{(i)}]
		the infinite boundary points of $J^\circ$ are inaccessible and,
		
		\item[\textup{(ii)}]
		for every \emph{accessible} boundary point $b\in\{l,r\}\cap\bR$, it holds
		\begin{equation}\label{eq:140324a1}
			\int_B \frac{\big|(\s(x)-\s(b))\mu(x)\big|}{\s'(x)\sigma^2(x)}\,dx<\infty
		\end{equation}
		for some (equivalently, for every) open interval $B\subsetneq J^\circ$ with $b$ as endpoint.
	\end{enumerate}
We further recall that a sufficient condition for~\eqref{eq:140324a1} is that, for every \emph{accessible} boundary point $b\in\{l,r\}\cap\bR$, it holds
\begin{equation}\label{eq:140324a2}
	\text{either }\mu\ge0\;\mu_L\text{-a.e. on }B
	\text{ or }\mu\le0\;\mu_L\text{-a.e. on }B
\end{equation}
for a sufficiently small open interval $B\subsetneq J^\circ$ with $b$ as endpoint, where $\mu_L$ denotes the Lebesgue measure (\cite[Corollary 3.11]{MU2015}).
\end{example}

At this point, we stress that the results we are going to present in this paper apply to all general diffusions which is a much richer class than the class of It\^o diffusions.
In the following examples, we only mention some striking effects that are also included in our general diffusion framework.

\begin{example}[General diffusion market with a sticky point] \label{ex: sticky}
	Another interesting class of general diffusions are the ones with stickiness. The most prominent example is sticky Brownian motion\footnote{More precisely,
		Brownian motion with state space $\bR$ and sticky at zero.},
	which is the (unique in law) solution \(\Y\) to the system 
	\[
	d \Y_t = \1_{\{ \Y_t \not = 0\}} d W_t, \quad \1_{\{\Y_t = 0\}} dt = \rho\, d L^0_t (\Y),
	\]
	where \(\rho > 0\) is a so-called stickiness parameter and \(L^0 (\Y)\) is the (right-continuous) semimartingale local time of the solution \(\Y\) in zero. For a discussion of this representation we refer to the paper \cite{EngPes}.

	The sticky Brownian motion is a general diffusion on natural scale
	with state space $\bR$
	and speed measure 
	\[
	\m (dx) = dx  + \rho \, \delta_0 (dx).	
	\]
	At this point, we notice that the sticky Brownian motion cannot be realized as a solution to an SDE as in Example~\ref{ex:090224a1} because the speed measure in Example~\ref{ex:090224a1} is always absolutely continuous w.r.t. the Lebesgue measure.
It is worth noting that a.a. paths of $Y$ spend positive Lebesgue time in zero without having intervals of zeros.
\end{example}

\begin{example}[General diffusion market with a countable dense set of sticky points] \label{ex:260624a1}
A rather ``extreme'' version of Example~\ref{ex: sticky} is due to Feller and McKean (see \cite[Section 2.12]{freedman} for more comments).
Let \(D = \{d_1, d_2, \dots \}\) be a countable dense subset of \(\bR\) and let \(\m\) be a measure on \((\bR, \mathcal{B}(\bR))\) that is concentrated on \(D\) and such that \(\m (\{d_k\}) > 0\) and \(\sum_{k = 1}^\infty \m (\{d_k\}) < \infty\). Then, \(\m\) is a valid speed measure.
Let $x_0\in\bR$ and let $\Y$ be a general diffusion on natural scale with state space $\bR$, speed measure \(\m\) and starting value \(x_0\).
Like in the previous example, $Y$ cannot be realized as a solution to an SDE as in Example~\ref{ex:090224a1}.
While this process is continuous, it has quite peculiar paths in the sense that it spends positive Lebesgue time in every point $d_k$ (without intervals of constancy) and zero Lebesgue time in every point in $\bR\setminus D$. Moreover,
one can prove that \(\P(\Y_t \in D) = 1\) for all \(t>0\),
cf. \cite[Lemma 2.144]{freedman} or \cite[Section 4.11]{itokean74}.
\end{example} 

\begin{example}[General diffusion market with skewness] \label{ex: skew}
	Another interesting class are diffusions with skewness. The most basic example is the skew Brownian motion\footnote{More precisely, Brownian motion with state space \(\bR\) and skew at zero.},
which is a solution process \(\Y\) of the equation
	\[
	d \Y_t = d W_t + (2 \alpha - 1) d \ell^0_t (\Y), 
	\] 
	where \(\alpha \in (0,1) \setminus \{\frac{1}{2}\}\) is the so-called skewness parameter and \(\ell^0 (\Y)\) is the symmetric semimartingale local time of \(\Y\) in zero. 
	
	It is well-known
	(see \cite{HS}, \cite[Appendix 1.12]{borodin_salminen} or \cite[Exercise X.2.24]{RY})
	that \(\Y\) is a general diffusion with state space \(\bR\), scale function 
	\[
	\s (x) = \begin{cases} (1 - \alpha) x, & x \geq 0, \\ \alpha x, & x < 0,\end{cases} 
	\] 
	and speed measure
	\[
	\m (dx) = \frac{dx}{v_\alpha (x)}
	\qquad\text{with}\qquad
	v_\alpha (x) = \begin{cases} 1 - \alpha, & x \geq 0, \\ \alpha, & x < 0.\end{cases}
	\]
As the scale function is not continuously differentiable, the skew Brownian motion cannot be realized as a solution to an SDE as in Example~\ref{ex:090224a1}.
Finally, we remark
that Condition~\ref{cond:050324a1} is violated for the skew Brownian motion model. 
\end{example}

\subsection{Main results: finite time horizon}
In this subsection, we provide deterministic characterizations for NA, NUPBR and NFLVR in the case $T<\infty$. Evidently, by virtue of Theorem~\ref{th:080224a1}, the main work lies in characterizations of NA and NUPBR, which are given in the following two theorems.

\begin{theorem}\label{th:090224a1}
Assume that $T<\infty$.
The financial market \((\mathbb{B}, \Y)\) satisfies NA if and only if Condition~\ref{cond:170623a1} holds and, for every $b\in\{l,r\}\cap\bR$, at least one of conditions \eqref{eq: b good}--\eqref{eq: very good add cond} is satisfied.
\end{theorem}

\begin{theorem}\label{th:090224a2}
Assume that $T<\infty$.
The financial market \((\mathbb{B}, \Y)\) satisfies NUPBR if and only if Condition~\ref{cond:170623a1} holds and, for every $b\in\{l,r\}\cap\bR$, at least one of conditions \textup{(a)--(b)} below is satisfied:
\begin{enumerate}
\item[\textup{(a)}]
condition~\eqref{eq: b good} holds;
\item[\textup{(b)}]
the boundary point $b$ is inaccessible for~$\Y$.
\end{enumerate}
\end{theorem}

Combining these results with Theorem~\ref{th:080224a1} gives us the following characterization of NFLVR.
A version of this result for a canonical diffusion setting is given by \cite[Theorem~3.5, Remark~3.10]{criensurusov22}.
For the It\^o diffusion market from Example~\ref{ex:090224a1} with $J^\circ=(0,\infty)$,
a deterministic characterization of NFLVR in terms of the drift and volatility coefficients was established in \cite[Theorem 3.1]{MU12b}.

\begin{corollary}\label{cor:090224a1}
Assume that $T<\infty$.
The financial market \((\mathbb{B}, \Y)\) satisfies NFLVR if and only if Condition~\ref{cond:170623a1} holds and, for every $b\in\{l,r\}\cap\bR$, at least one of conditions \textup{(A)--(B)} below is satisfied:
\begin{enumerate}
\item[\textup{(A)}]
condition~\eqref{eq: b good} holds;
\item[\textup{(B)}]
the boundary point $b$ is inaccessible for $\Y$ and \eqref{eq: very good add cond} holds.
\end{enumerate}
\end{corollary}

Evidently, the above characterizations of NA, NUPBR and NFLVR do not depend on the time horizon $T\in(0,\infty)$. Hence, we also have the following:

\begin{corollary}
If the financial market $(\mathbb B,\Y)$
satisfies NA (resp., NUPBR; resp., NFLVR) for \emph{some} $T\in(0,\infty)$,
then it
satisfies NA (resp., NUPBR; resp., NFLVR) for \emph{all} $T\in(0,\infty)$.
\end{corollary}

To apply the above characterizations in practice, one often has to understand the finiteness or infiniteness of several (deterministic) integrals.
The following lemma shows some interdependencies between the involved integrals, which are useful for verifying certain conditions in specific situations
(see Examples~\ref{ex:090224a2} and~\ref{ex:090224a3} below).

\begin{lemma}\label{lem:170623a1}
Assume that Condition~\ref{cond:050324a1} holds and let \(b \in \{l, r\} \cap \bR\).
\begin{enumerate}
\item[\textup{(i)}]
If $|\s(b)|=\infty$, then \eqref{eq: b good} for $b$ is violated.

\item[\textup{(ii)}]
If \eqref{eq: b good} holds for $b$, then $|\s(b)|<\infty$.

\item[\textup{(iii)}]
Suppose that one of the conditions \textup{(iii.a)--(iii.b)} below is satisfied:
\begin{enumerate}
\item[\textup{(iii.a)}]
the boundary point $b$ is accessible for $\Y$ and \eqref{eq: very good add cond} holds;
\item[\textup{(iii.b)}]
the boundary point $b$ is inaccessible for $\Y$ and \eqref{eq: very good add cond} is violated.
\end{enumerate}
Then, \eqref{eq: b good} for $b$ is violated.

\item[\textup{(iv)}]
If \eqref{eq: b good} holds for $b$, then one of the conditions \textup{(iv.a)--(iv.b)} below is satisfied:
\begin{enumerate}
\item[\textup{(iv.a)}]
the boundary point $b$ is inaccessible for $\Y$ and \eqref{eq: very good add cond} holds;
\item[\textup{(iv.b)}]
the boundary point $b$ is accessible for $\Y$ and \eqref{eq: very good add cond} is violated.
\end{enumerate}
\end{enumerate}
\end{lemma}

\begin{example}\label{ex:090224a2}
In the setting of Example~\ref{ex:090224a1}, we take $x_0\in J^\circ=(0,\infty)$, $\mu(x)=1/x$ and $\sigma\equiv1$.
This means that $\Y$ is a Bessel process of dimension $3$ started at $x_0$.
This is a famous example for that
Delbaen and Schachermayer \cite{DS1995a} showed the existence of arbitrage (with admissible strategies).
Moreover, Karatzas and Kardaras \cite[Example 4.6]{karatzaskardaras07} constructed an arbitrage in closed form in this example.
To provide another perspective, we discuss how our theorems apply.

As $\mu^2 / \sigma^4 \in L^1_\textup{loc} (J^\circ)$,
Condition~\ref{cond:170623a1} is satisfied in this example (recall the discussion in Example~\ref{ex:090224a1}).
To apply our results, we need to verify \eqref{eq: b good} and~\eqref{eq: very good add cond} for finite boundary points (i.e., only for $b=0$ here) and to check whether $0$ is accessible for $\Y$.
Straightforward calculations reveal that $\s(0)=-\infty$ (hence $0$ is inaccessible for $\Y$) and that \eqref{eq: very good add cond} with $b=0$ is violated.
Further calculations are not needed: the fact that \eqref{eq: b good} with $b=0$ is violated now follows from Lemma~\ref{lem:170623a1}.

Summing up, for every finite time horizon, in this example NUPBR holds, while NA and NFLVR are violated.
\end{example}

\begin{example}\label{ex:090224a3}
In the setting of Example~\ref{ex:090224a1}, we take $x_0\in J^\circ=(0,\infty)$, $\mu\equiv-1$ and $\sigma(x)=x$.
As \eqref{eq:140324a2} is satisfied, $\Y$ is a semimartingale (cf. the discussion around~\eqref{eq:140324a1}).

As $\mu^2 / \sigma^4 \in L^1_\textup{loc} (J^\circ)$,
Condition~\ref{cond:170623a1} is satisfied.
Straightforward calculations reveal that the origin is accessible for $\Y$ and that \eqref{eq: very good add cond} with $b=0$ holds.
Again, Lemma~\ref{lem:170623a1} implies that \eqref{eq: b good} with $b=0$ is violated.

Summing up, for every finite time horizon, in this example NA holds, while NUPBR and NFLVR are violated.
\end{example}

\begin{discussion}[Relations between the no-arbitrage notions]\label{disc:090224a1}
Examples \ref{ex:090224a2} and~\ref{ex:090224a3} show that the notions NA and NUPBR are \emph{in a general position} (that is, neither implies the other), while their relation to NFLVR is given in Theorem~\ref{th:080224a1}.
In general, this is well understood.
On the contrary, the following observation seems to be new.
There are very natural classes of general diffusion markets, where NA and NUPBR are \emph{not} in a general position (that is, one of them implies the other).
Fix some $T<\infty$ and consider the following subsettings of our setting.

\smallskip
\emph{Subsetting~1:} $J=\bR$.
Here, we have
$$
\text{NFLVR}
\;\;\Longleftrightarrow\;\;
\text{NA}
\;\;\Longleftrightarrow\;\;
\text{NUPBR}
\;\;\Longleftrightarrow\;\;
\text{Condition~\ref{cond:050324a1}}.
$$

\smallskip
\emph{Subsetting~2:} $J\subsetneq\bR$ and all finite boundary points of $J$ are inaccessible for $Y$.
Here, we have
$$
\text{NFLVR}
\;\;\Longleftrightarrow\;\;
\text{NA}
\;\;\Longrightarrow\;\;
\text{NUPBR}
\;\;\Longleftrightarrow\;\;
\text{Condition~\ref{cond:050324a1}},
$$
but NUPBR does not imply NA, as illustrated by Example~\ref{ex:090224a2}.

\smallskip 
\emph{Subsetting~3:} $J\subsetneq\bR$ and all finite boundary points of $J$ are accessible for $Y$.
Here, we have
$$
\text{NFLVR}
\;\;\Longleftrightarrow\;\;
\text{NUPBR}
\;\;\Longrightarrow\;\;
\text{NA}
\;\;\Longrightarrow\;\;
\text{Condition~\ref{cond:170623a1}},
$$
but NA does not imply NUPBR as illustrated by Example~\ref{ex:090224a3}.
Moreover, in Subsetting~3, Condition~\ref{cond:170623a1} does not imply NA (see Example~\ref{ex:090224a4} below).
\end{discussion}

\begin{example}\label{ex:090224a4}
Motivated by the previous discussion, let us also construct an example within Subsetting~3 of Discussion~\ref{disc:090224a1}, where Condition~\ref{cond:170623a1} is satisfied but NA fails.

To this end, in the setting of Example~\ref{ex:090224a1}, we take $\delta\in(0,2)$, $x_0\in J^\circ=(0,\infty)$, $\mu\equiv\delta$ and $\sigma(x)=2\sqrt x$.
It is well-known that the origin is accessible for~$Y$ (one could simply verify~\eqref{eq:170323a1} with $b=0$),
so we are in Subsetting~3 of Discussion~\ref{disc:090224a1}.\footnote{This process $Y$ is a squared Bessel process of dimension $\delta$ started at $x_0$ and absorbed in the origin (the absorbing boundary condition comes from the setting of Example~\ref{ex:090224a1}).}
Condition~\ref{cond:170623a1} is clearly satisfied,
while straightforward calculations show that both \eqref{eq: b good} and \eqref{eq: very good add cond} with $b=0$ are violated.
By Theorem~\ref{th:090224a1}, NA fails on every finite time horizon.
\end{example}

\begin{example}
For the skew market model from Example~\ref{ex: skew}, NA, NUPBR and NFLVR all fail on every finite time horizon,
because Condition~\ref{cond:050324a1} is violated.
\end{example}

\subsection{Main results: infinite time horizon}\label{subsec:mr_infinite}
Next, we consider the case $T=\infty$.
First, suppose that $Y$ is on natural scale (i.e., $\s=\on{id}$, up to an increasing affine transformation).
Then, there is nothing to study:
\begin{itemize}
\item
If every finite boundary point $b\in\{l,r\}\cap\bR$ is either inaccessible or absorbing for $Y$, then $Y$ is an \(\mathbf{F}\)-\(\P\)-local martingale.
This means that $\P$ is an ELMM.
Hence, NFLVR, NA and NUPBR hold (see Theorems \ref{theo: FTAP NFLVR} and~\ref{th:080224a1}).

\item
If, on the contrary, $Y$ has an (instantaneously or slowly) reflecting finite boundary point $b\in\{l,r\}\cap\bR$, then NA and NUPBR, hence also NFLVR, are violated (see Lemmata \ref{lem:060324a2} and~\ref{lem:060324a4} below).
\end{itemize}
Therefore, below we only consider the case that $Y$ is \emph{not} on natural scale.

\begin{theorem}\label{th:100224a1}
Suppose that $T=\infty$ and that $Y$ is \emph{not} on natural scale.
The financial market \((\mathbb{B}, \Y)\) satisfies NA if and only if Condition~\ref{cond:170623a1} holds and, for every $b\in\{l,r\}$,
one of the following conditions \textup{(a)--(b)} is satisfied:
\begin{enumerate}
\item[\textup{(a)}]
$|b|<\infty$ and \eqref{eq: b good} holds;
\item[\textup{(b)}]
$|b|=\infty$ and the other boundary point $b^*$ satisfies~\textup{(a)}.
\end{enumerate}
\end{theorem}

The reason for excluding the case where $Y$ is on natural scale from this characterization (and from the following ones) is that,
in this case with $J=\bR$,
neither (a) nor (b) is satisfied, while, as discussed above, NFLVR, NA and NUPBR hold.

\begin{theorem}\label{th:100224a2}
Suppose that $T=\infty$ and that $Y$ is \emph{not} on natural scale.
The financial market \((\mathbb{B}, \Y)\) satisfies NUPBR if and only if Condition~\ref{cond:170623a1} holds and, for every $b\in\{l,r\}$,
one of the following conditions \textup{(A)--(B)} is satisfied:
\begin{enumerate}
\item[\textup{(A)}]
$|b|<\infty$ and \eqref{eq: b good} holds;
\item[\textup{(B)}]
$|\s(b)|=\infty$ and the other boundary point $b^*$ satisfies~\textup{(A)}.
\end{enumerate}
\end{theorem}

In connection with the formulation of Theorem~\ref{th:100224a2}, we remark that, under Condition~\ref{cond:170623a1}, a boundary point $b\in\{l,r\}$ cannot satisfy (A) and (B) simultaneously.
While this is not easy to see from (A) and (B) directly, this fact follows from Lemma~\ref{lem:170623a1}:
if $|b|<\infty$ and $|\s(b)|=\infty$, then \eqref{eq: b good} is violated.

The above theorems together with Theorem~\ref{th:080224a1} give us the following characterization of NFLVR on the infinite time horizon.
A version of this result for a canonical diffusion setting is given by \cite[Theorem~3.8, Remark~3.10]{criensurusov22}.
For the It\^o diffusion market from Example~\ref{ex:090224a1}
with $J^\circ=(0,\infty)$,
a deterministic characterization in terms of the drift and volatility coefficients was established in \cite[Theorem 3.5]{MU12b}.

\begin{corollary}\label{cor:100224a1}
Suppose that $T=\infty$ and that $Y$ is \emph{not} on natural scale.
The financial market \((\mathbb{B}, \Y)\) satisfies NFLVR if and only if Condition~\ref{cond:170623a1} holds and, for every $b\in\{l,r\}$,
one of the following conditions \textup{(I)--(II)} is satisfied:
\begin{enumerate}
\item[\textup{(I)}]
$|b|<\infty$ and \eqref{eq: b good} holds;
\item[\textup{(II)}]
$|b|=\infty$, $|\s(b)|=\infty$ and the other boundary point $b^*$ satisfies~\textup{(I)}.
\end{enumerate}
\end{corollary}

We highlight the following technical difference between the characterizations on a finite time horizon and those on the infinite time horizon.
While in the former the conditions additional to Condition~\ref{cond:170623a1} are imposed only on finite boundary points $b\in\{l,r\}\cap\bR$,
it is important that in the latter characterizations the conditions additional to Condition~\ref{cond:170623a1} are imposed on all boundary points $b\in\{l,r\}$.

An immediate consequence of the above characterizations is that NA or NUPBR (or NFLVR) can only hold on the infinite time horizon for a setting \emph{not} on natural scale when at least one boundary point is finite. 

\begin{corollary}\label{cor:100224a2}
Suppose that $T=\infty$,
that $\Y$ is \emph{not} on natural scale
and that $J=\bR$.
Then NA and NUPBR (hence also NFLVR) are violated.
\end{corollary}

\begin{example}\label{ex:100224a1}
Let $J=\bR$ and $\Y$ be a Brownian motion with linear drift, i.e., 
$\Y_t=W_t+t$.
By Corollary~\ref{cor:100224a2}, NA and NUPBR (hence also NFLVR) are violated on the infinite time horizon.
On the other hand, NFLVR (hence also NA and NUPBR) hold on every finite time horizon (one may construct an ELMM via Girsanov's theorem or, alternatively, 
observe that we are in Subsetting~1 of Discussion~\ref{disc:090224a1}, while Condition~\ref{cond:050324a1} is trivially satisfied).
\end{example}

A continuous semimartingale $(S_t)_{t\in\bR_+}$ is said to be \emph{closable} if the limit $S_\infty\triangleq\lim_{t\to\infty}S_t$ exists finitely a.s. and $S=(S_t)_{t\in[0,\infty]}$ is a \emph{semimartingale on $[0,\infty]$} in the sense that there is a decomposition $S=S_0+M+A$, $M_0=A_0=0$, with $M=(M_t)_{t\in[0,\infty]}$ being a \emph{continuous local martingale on $[0,\infty]$}
(i.e., there exists a localizing sequence $(\tau_n)_{n = 1}^\infty$ of stopping times with $\{\tau_n=\infty\}\nearrow\Omega$ a.s.)
and $A=(A_t)_{t\in[0,\infty]}$ being a continuous adapted process that has finite variation on the closed interval $[0,\infty]$.

In Remark~\ref{rem:080224a1}, we already discussed the difference between local and global versions of no arbitrage. We saw that in general for each of the notions NA, NUPBR and NFLVR, their local version is strictly weaker than their global version. Inspecting the counterexamples from \cite[Remark~5.3]{balint_schweizer_2020} and Example~\ref{ex:100224a1}, it turns out that the asset price processes are not closable. 
Indeed, the process \(Y\) from Example~\ref{ex:100224a1} is not closable, because
$$
W_t + t = t\, \Big (\frac{W_t}{t} + 1\Big) \to \infty\quad\text{a.s.}
$$
by the strong law of large numbers for Brownian motion. It is natural to ask whether all feasible counterexamples fail to be closable. In other words,
for each no-arbitrage notion $\text{NN}\in\{\text{NA},\text{NUPBR},\text{NFLVR}\}$,
we ask whether
\begin{center}
    NN holds for all finite time horizons and \(Y\) is closable
    \quad \(\overset{?\,}{\Longrightarrow}\) \quad
    NN holds for \(T =\infty\).
\end{center}
It turns out that this implication is {\em not true} and we provide a counterexample in Example~\ref{ex:140324a1} below.
While the process \(Y\) from Example~\ref{ex:140324a1} will have state space $(0,\infty)$, it is also natural to ask whether we could give a counterexample with state space $[0,\infty)$.
Our characterization results imply that this is impossible for NUPBR and NFLVR but possible for NA, as shown in the following example.

\begin{example}\label{ex:140324a2}
	In the setting of Example~\ref{ex:090224a1} we take $x_0\in J^\circ=(0,\infty)$, $\mu\equiv-1$ and $\sigma(x)=x$.
	In other words, we revisit Example~\ref{ex:090224a3}, where we observed that NA holds on every finite time horizon.
	Recalling that \eqref{eq: b good} with $b=0$ is violated, we infer from Theorem~\ref{th:100224a1} that NA fails on the infinite time horizon.

It remains to explain that $Y$ is closable.
As the origin is accessible for $Y$ and \(\s (\infty) = \infty\),
a.s. \(Y\) hits the origin in finite time and gets absorbed there.
Consequently, \(Y_\infty = 0\) a.s.
and the semimartingale property of \((Y_t)_{t \in \bR_+}\) transfers directly to \((Y_t)_{t \in [0, \infty]}\). We conclude that \(Y\) is closable.
\end{example}

\subsection{Generalization: integrated diffusion markets}
In this subsection, we discuss a natural extension of our previous results to so-called {\em integrated diffusion markets}. 
We again consider a finite or infinite deterministic time horizon $T\in(0,\infty]$ and use Agreement~\ref{agr:050224a1}.
Let \(\sigma \in L (\mathbb{B}, \Y)\) be a \emph{non-vanishing} process in the sense that $\P$-a.s. \(\sigma_t \ne 0\) for all $t\in[0,T]$.
We set
\begin{equation}\label{eq:100224a1}
\S_t \triangleq \S_0 + \int_0^t \sigma_s \,d \Y_s, \quad t \in [0, T],
\end{equation}
where \(\S_0\colon\Omega\to\bR\) is an $\cF_0$-measurable initial value.
It follows directly from the Definitions \ref{def:100224a1}--\ref{def:100224a3} that,
for each no-arbitrage notion $\text{NN}\in\{\text{NA},\text{NUPBR},\text{NFLVR}\}$,
$$
\text{NN in the market }(\mathbb B,\S)
\;\;\Longleftrightarrow\;\;
\text{NN in the market }(\mathbb B,\Y).
$$
In this sense, our results also provide deterministic characterizations of NA, NUPBR and NFLVR in the market $(\mathbb B,\S)$ both for finite and for infinite time horizon.

Introducing a non-vanishing volatility process \(\sigma\) covers classical exponential models very naturally.
Namely, if \(\S = \S_0\, \mathcal{E} (\Y)\) is the stochastic exponential of a general diffusion semimartingale \(\Y\) with a strictly positive $\cF_0$-measurable initial value \(\S_0\colon\Omega\to(0,\infty)\), then
\[
\S_t = \S_0 + \int_0^t \S_s \,d \Y_s, \quad t \in [0, T],
\]		
is a particular case of~\eqref{eq:100224a1} (with $\sigma=\S$, which is strictly positive, hence non-vanishing, as needed).

\begin{example}\label{ex:140324a3}
Let $\Y$ have state space $J=\bR$.
Let $\S=\S_0\,\mathcal E(\Y)$ (a widespread approach) or, more generally, let $\S$ be given by~\eqref{eq:100224a1}.
For the financial market $(\mathbb B,\S)$, we can immediately state the following.
\begin{itemize}
\item
If $T<\infty$, then
$$
\text{NFLVR}
\;\;\Longleftrightarrow\;\;
\text{NA}
\;\;\Longleftrightarrow\;\;
\text{NUPBR}
\;\;\Longleftrightarrow\;\;
\text{Condition~\ref{cond:050324a1}}
$$
(see Subsetting~1 of Discussion~\ref{disc:090224a1}).

\item
If $T=\infty$, then \(\on{NN} \in \{ \on{NA}, \on{NUPBR}, \on{NFLVR}\}\) holds if and only if \(\Y\) is on natural scale, see Corollary~\ref{cor:100224a2} and the beginning of Section~\ref{subsec:mr_infinite}.
\end{itemize}
\end{example}

\begin{example}\label{ex:140324a1}
Let $J=(0,\infty)$ and $\Y$ be the stochastic exponential of a Brownian motion with drift, i.e., 
\begin{equation}\label{eq:140324a3}
d\Y_t= \Y_t (a\,dt+ dW_t),\quad\Y_0=x_0\in(0,\infty),
\end{equation}
where $a\in\bR\setminus\{0\}$.
Example~\ref{ex:140324a3} applied to $(W_t+at)_{t \in [0, T]}$ implies that the notions NA, NUPBR and NFLVR hold on every finite time horizon and are violated on the infinite time horizon.

In relation with the discussion after Example~\ref{ex:100224a1}, we show that $\Y$ is a closable semimartingale whenever $a<1/2$.
Take such an $a\in(-\infty,1/2)\setminus\{0\}$.
Notice that the canonical decomposition $\Y=x_0+A+M$, $A_0=M_0=0$, is provided in~\eqref{eq:140324a3} ($A$ is the finite variation part, $M$ is the local martingale part of $\Y$).
From the formula
$\Y_t=x_0\exp\{(a-1/2)t+W_t\}$ it is immediate that $\Y_\infty\triangleq\lim_{t\to\infty}\Y_t=0$ a.s. and
$\int_0^\infty |a|\Y_s\,ds<\infty$ a.s.
The latter means that $A$ has finite variation on the whole time interval $[0,\infty]$.
In particular, we have finite limits $A_\infty\triangleq\lim_{t\to\infty}A_t$ and $M_\infty\triangleq\lim_{t\to\infty}M_t$.
Finally, the fact that $M$ is a continuous local martingale on $[0,\infty]$ follows from the fact that $M$ is a continuous local martingale on $\bR_+$ with a finite limit $M_\infty$.
Indeed, the sequence of stopping times $\tau_n\triangleq\inf\{t\geq 0\colon |M_t|\ge n\}$, $n\in\mathbb N$, is localizing for $M$ with $\{\tau_n=\infty\}\nearrow\Omega$ a.s.
\end{example}

\subsection{Existence of an EMM}
We complete our presentation with a characterization of a no-arbitrage notion that is even stronger than NFLVR.
Under different names, Sin \cite{sin}, Yan \cite{yan} and Cherny \cite{cherny} introduced some strengthenings of NFLVR. For each of their notions, they proved equivalence to the existence of an {\em equivalent martingale measure (EMM)}, i.e., an equivalent measure that turns the asset price process into a
{\em uniformly integrable}\footnote{Every martingale on a time interval closed from the right is uniformly integrable, so, in the case $T<\infty$, uniform integrability is a void requirement. On the contrary, when $T=\infty$, uniform integrability is essential in the no-arbitrage notion of \cite{cherny}. We also notice that \cite{sin} and \cite{yan} consider only the case $T<\infty$.}
martingale.
Deterministic characterizations of the existence of EMMs in the It\^o diffusion framework from Example~\ref{ex:090224a1}
with $J^\circ=(0,\infty)$
can be found in Section~3.2 of \cite{MU12b}.

We consider a general financial market \((\mathbb{B}, Y)\) with finite or infinite deterministic time horizon $T\in(0,\infty]$ and use Agreement~\ref{agr:050224a1}.

\begin{theorem} \label{theo: NGA}
	Let $T<\infty$.
	There exists an EMM for the market \((\mathbb{B}, Y)\), i.e., a probability measure \(\Q \sim \P\) such that \(Y\) is an \(\mathbf{F}\)-\(\Q\)-martingale, if and only if \((\mathbb{B}, Y)\) satisfies NFLVR and,
	for every infinite boundary point \(b \in \{l,r\}\setminus\bR\),
	\[
	\int_B | x | \s' (x) \m (dx) = \infty
	\]
	holds for some (equivalently, for every) open interval \(B \subsetneq J^\circ\) with \(b\) as endpoint.
\end{theorem}

As for the notions NA, NUPBR and NFLVR, the above characterization of the existence of an EMM does not depend on the time horizon $T\in(0,\infty)$. Hence, we also have the following:

\begin{corollary}
If there exists an EMM for the market $(\mathbb B,\Y)$
for \emph{some} $T\in(0,\infty)$,
then there exists an EMM for this market
for \emph{all} $T\in(0,\infty)$.
\end{corollary}

\begin{theorem} \label{theo: NGA2}
	Let $T=\infty$.
	There exists an EMM for the market \((\mathbb{B}, Y)\), i.e., a probability measure \(\Q \sim \P\) such that \(Y\) is a uniformly integrable \(\mathbf{F}\)-\(\Q\)-martingale, if and only if NFLVR holds and both boundaries \(l\) and \(r\) are finite, i.e., \(l, r \in \bR\).
\end{theorem}

It is interesting to note that in our setting an EMM for the infinite time horizon \(T = \infty\) can only exist if the state space \(J\) is bounded.

\section{Outline and comments on the proofs} \label{sec: outline proof}

Before we prove our main theorems in the upcoming sections, we comment on proof tactics. Let us start with the setting of Example~\ref{ex:090224a1}, i.e., we assume that
\begin{equation} \label{eq: SDE}
	\begin{split}
	d Y_t &= \mu (Y_t) dt + \sigma (Y_t) d W_t, \quad t < \zeta (Y) \triangleq \inf \{t \in [0, T] \colon Y_t \not \in J^\circ \}, \\
	Y_0 &= x_0, 
\end{split}
\end{equation}
where \(\mu, \sigma \colon J^\circ \to \bR\) satisfy the Engelbert--Schmidt conditions from \eqref{eq: ESC}. 
To understand NUPBR in such a market we need to understand when the process 
\begin{align*}
	Z_t &\triangleq \begin{cases}\exp \big\{ \int_0^t \theta_s d W_s - \frac{1}{2} \int_0^t \theta^2_s ds \big\}, & t < \zeta (Y), \\ 
	 \liminf_{s \nearrow \zeta (Y)} Z_s,& t \geq \zeta (Y), \end{cases}
\end{align*}
with
\[
\theta_t \triangleq - \Big(\frac{\mu}{\sigma}\Big)(Y_t),
\]
defines a positive local martingale. We stress that this question consists of two parts, its well-definedness as a local martingale and the strict positivity.
Related to the structure condition from Theorem~\ref{theo: SC}, it is essentially known that these properties hold if and only if a.s. 
\[
\int_0^{T \wedge \zeta (Y)} \theta^2_s \, ds < \infty. 
\]
For the It\^o diffusion setting, the finiteness of such an additive functional can be characterized via the coefficients \(\mu\) and \(\sigma\), see \cite{MU12ECP} for details.
The question of NA has a similar flavor.
Essentially, one needs to understand when the above process \(Z\) is a true martingale but it does not need to be strictly positive anymore.
A delicate point here is that, in the case where $Z$ is a strict local martingale,
it could a priori be the case that there is an ACLMM with a different density process, cf.
\cite{DS1998counter} or \cite[Chapter 10]{DS2006} for a counterexample illustrating this point.
As a by-product of our arguments, we establish that such counterexamples are impossible in our setting, even when the underlying filtration is larger than the one generated by $\Y$.
The key ingredient for this fact is the strong Markov property of \(Y\) w.r.t. the underlying filtration.
Additional technicalities arise from the fact that the existence of an ACLMM is only a necessary condition for NA (recall Theorem~\ref{theo: FTAP NA})
and that we actually have to investigate the existence of ACLMMs for a family of shifted market models.

\smallskip
A crucial difference between the It\^o diffusion and our general diffusion setting is that the semimartingale decomposition for the It\^o diffusion setting is explicitly given from the outset (which allows a direct application of Theorem~\ref{theo: SC}, for instance).
In particular, this gives access to a tractable representation of a candidate for the minimal SMD.
To see how we can still benefit from the above consideration, let us reformulate the formula for \(Z\) in terms of the natural scale transformation of the It\^o diffusion \(Y\). 
Recall that the scale function of the above It\^o diffusion is given by 
\begin{align*}
		\s (x) = \int^x \exp \Big\{ -\int^y \frac{2\mu (z)}{\sigma^2 (z)}\, dz \Big\} \, dy, \quad x \in J^\circ.
\end{align*}
It follows easily from the generalized It\^o formula (\cite[Theorem~IV.45.1]{RW2} or \cite[Lemma B.23]{criensurusov22}) that \(U \triangleq \s (Y)\) satisfies the SDE
\[
d U_t = \tilde{\sigma} (U_t) d W_t, \quad t < \zeta (Y), 
\]
where
\begin{align*}
	\tilde{\sigma} (x) \triangleq \begin{cases} \s' ( \s^{-1} (x)) \sigma (\s^{-1} (x)), & x \in \s (J^\circ), \\ 0, & x \not \in \s (J^\circ), \end{cases} 
\end{align*}
cf., for example, \cite[Proposition~5.5.13]{KaraShre}. Now, for \(t < \zeta (Y)\), we may write 
\begin{align*}
	\int_0^t \theta_s d W_s &= \int_0^t \frac{\theta_s}{\tilde{\sigma} (U_s)} d U_s 
	\\&= - \int_0^t \frac{\mu (\s^{-1} (U_s)) }{\sigma^2 (\s^{-1} (U_s)) \s' (\s^{-1} (U_s))} d U_s 
	\\&= \int_0^t  \frac{\s'' (\s^{-1} (U_s))}{2\,  (\s' (\s^{-1} (U_s)))^2} d U_s.
\end{align*}
This formula does not rely on the fact that \(Y\) has SDE dynamics. Instead, the important point is that the scale function \(\s\) is continuously differentiable with a strictly positive absolutely continuous derivative. In fact, this is precisely the structure that is given by~\eqref{eq: nflvrS}.

In summary, if we prove that both NUPBR and NA force the scale function to have a representation of
the type~\eqref{eq: nflvrS}, we could investigate the candidate density \(Z\) as in the It\^o diffusion case. To pitch the key ideas of the proofs for this representation, in case NUPBR holds, we can use a local change of measure up to a positive predictable time to deduce the
structure~\eqref{eq: nflvrS} from some deep results on the separating time for general diffusions that we recently proved in our paper \cite{criensurusov22}. When NA holds, the existence of an ACLMM, which is guaranteed by Theorem~\ref{theo: FTAP NA}, allows us again to reduce our question to one solved in \cite{criensurusov22}.
Of course, after \eqref{eq: nflvrS} is established,
the program outlined above  still requires a careful analysis of certain properties of general diffusions. 

\smallskip
We end this section with comments on the relation of our work to the recent paper \cite{desmettre}.
Let \(Y\) be a general diffusion with scale function \(\s\) and speed measure \(\m\), and take a function \(\varphi\) of the form 
\begin{align} \label{eq: varphi in domain}
\varphi (x) = c_1 + \int_\xi^x \Big( c_2 + \int_\xi^w 2 g (u) \m (du) \Big) \s (dw), 
\end{align} 
where \(c_1, c_2 \in \bR\) and \(g \colon J \to \bR\) is continuous with limits at the endpoints of \(J\). We presume that \(\varphi\) is strictly positive on \(J^\circ\). 
The main result from \cite{desmettre} provides necessary and sufficient conditions for the local martingale 
\begin{equation}\label{eq:270724a2}
 Z^*=\varphi (Y) \exp \Big\{ - \int_0^\cdot \frac{\mathcal{G} \varphi (Y_s)}{\varphi (Y_s)} ds \Big\}
\end{equation}
($\mathcal G = \frac{1}{2} \frac{d}{d\m} \frac{d}{d\s}$ denotes the infinitesimal generator of $Y$),
extended continuously in case the integral diverges, to be a true martingale.
In turn, by virtue of \cite[Theorem 2.7]{desmettre}, broadly speaking, these are necessary and sufficient conditions that allow us to make a change of measure from \(Y\) to a diffusion with scale \(\varphi^{-2} d \s\) and speed \(\varphi^2 d \m\). 
It is tempting to ask whether one can apply this result to deduce necessary and sufficient conditions for NFLVR.\footnote{That is, the idea is to search for a candidate density process of the form~\eqref{eq:270724a2}.
In this realm, it is worth mentioning that local martingales of the form~\eqref{eq:270724a2} are also studied in  \cite{Cetin2018,Cetin2024}.}
The answer to this question is in general negative. To understand this, we anticipate Lemma~\ref{lem:050324a1} below, which shows that we would need to apply the result with \(\varphi \equiv \sqrt{\s'}\). However, \(\sqrt{\s'}\) is not necessarily of the form~\eqref{eq: varphi in domain}. In fact, it often has much less regularity. For example, if \(Y\) has the SDE dynamics~\eqref{eq: SDE}, then 
\[
\sqrt{\s'} = \exp \Big\{ - \int^\cdot \frac{\mu (z)}{\sigma^2(z)} dz \Big\}
\]
is, in general, only absolutely continuous,
while functions of the form~\eqref{eq: varphi in domain} are continuously differentiable with absolutely continuous derivatives (use \eqref{eq:090224a2} and~\eqref{eq:090224a3}).
In other words, in this way we cannot gain full generality even in the case when $Y$ has SDE dynamics~\eqref{eq: SDE}.

\smallskip
Last, we provide an overview on the content of Section~\ref{sec:proofs} that presents the proofs of our main results.
Section~\ref{sec: preparation} provides important preliminary lemmata.
In particular, there we show that NUPBR and NA can be transferred to the canonical diffusion setting and that
these notions
are, up to a slightly smaller time horizon, uniform in the initial value of the diffusion. This last observation is very important because it allows us to establish the necessity of our deterministic characterizations through local considerations.
We further investigate the existence of a candidate SMD and show that any SMD or ACLMM can be locally related to a diffusion on natural scale.
These preparations suffice for proving our characterizations of NUPBR (but not yet for NA):
in Section~\ref{sec: pf lem:170623a1}, we prove Lemma~\ref{lem:170623a1};
the proof of Theorem~\ref{th:090224a2} is given in Section~\ref{sec: pf th:090224a2};
Theorem~\ref{th:100224a2} is proved in Section~\ref{sec: pf th:100224a2}.
Section~\ref{sec: more lemma} contains further preparatory results.
More specifically, there we lift the martingale property to a bigger filtration and investigate the existence of a candidate ACLMM.
Building upon the latter together with the lemmata from Section~\ref{sec: preparation}, we are able to characterize NA:
Theorems~\ref{th:090224a1} and \ref{th:100224a1} are proved in Section~\ref{sec: pf th:090224a1, th:100224a1}.
Although the characterizations of NFLVR immediately follow from the characterizations of NUPBR and NA, we present a direct alternative approach in Section~\ref{sec: alternative NFKVR}.
Finally, Section~\ref{sec: theo: NGA, theo: NGA2} provides the proofs of Theorems \ref{theo: NGA} and~\ref{theo: NGA2}.

\section{Proofs of main results}\label{sec:proofs}

In the proofs we need to use many fine results on diffusions and continuous semimartingales.
For the reader's convenience, in addition to original or textbook references, we also
provide references to the appendices of \cite{criensurusov22}, where most of the used results are collected.

Before we start our program, we introduce the so-called {\em canonical setting}.
Let \(\canOmega\) be the space of continuous functions \([0, T] \to J\) (recall Agreement~\ref{agr:050224a1}).
The coordinate process on \(\canOmega\) is denoted by \(\canY\), i.e., \(\canY_t (\omega) = \omega (t)\). Finally, we define
\begin{itemize}
\item
\(\canF \triangleq \sigma (\canY_t, t \in [0, T])\),

\item
$\canF_t^o \triangleq \sigma (\canY_s, s \in [0, t])$, $t \in [0, T]$,

\item
\(\canF_t \triangleq \canF^o_{t +}\), \(t \in [0, T)\), and

\item
\(\canF_{T} \triangleq \canF\).
\end{itemize}
For each point \(x \in J\), let \(\P_x\) be the law of the general diffusion~\(\Y\) with scale function $\s$, speed measure~$\m$ and initial value \(x\).
We set \(\mathbb{C}_x \triangleq (\canOmega, \canF, \canbfF = (\canF_t)_{t \in [0, T]}, \P_{x})\).
In addition to that, sometimes we also use the notation $\canbfF^o=(\canF^o_t)_{t\in[0,T]}$.
We also set
$$
T_a\triangleq\inf\{t\in[0,T] \colon \canY_t=a\}
$$
with the usual convention $\inf\emptyset\triangleq\infty$ ($a$ can be anything in $[-\infty,\infty]$) and
$$
\zeta\triangleq T_l\wedge T_r.
$$
Notice that $T_a$ and $\zeta$ are $\canbfF^o$-stopping times.

\medskip
In the following, we work partly with the space $\mathbb{B} = (\Omega, \cF, \mathbf{F} = (\cF_t)_{t \in [0, T]}, \P)$ from Section~\ref{sec:setting} and partly with the canonical setting.
In this sense, the ``canonical analog'' of $\mathbb B$ is the space $\mathbb C_{x_0}$, while the coordinate process $\canY$ seen under $\P_{x_0}$
(which is a measure on $(\canOmega,\canF)$)
is the ``canonical analog'' of the diffusion $\Y$ seen under $\P$
(which is a measure on $(\Omega,\cF)$).
On the other hand, there is no analog of the filtration $\mathbf F$ in the canonical setting.
The filtration $\canbfF$ is the ``canonical analog'' of the right-continuous filtration $\mathbf F^\Y=(\cF^\Y_t)_{t\in[0,T]}$ generated by $Y$.
We, finally, remark that we do not need new notation for the analogs of $T_a$ and $\zeta$ on $\mathbb B$, as we can always write $T_a(\Y)$ and $\zeta(\Y)$ for this purpose.

\subsection{Preliminary lemmata} \label{sec: preparation}
This section builds the foundation for the strategy of proof that was outlined in Section~\ref{sec: outline proof} above. First, we show that Condition~\ref{cond:170623a1} allows us to define a candidate for an SMD and the density of an ACLMM. Afterwards, we collect some preparations that are needed to establish Condition~\ref{cond:170623a1} under NA or NUPBR.

\smallskip 
The following lemma provides, under Condition~\ref{cond:170623a1}, an important step towards the construction of an SMD.
It is only a ``step'' because, even on a finite time horizon, we need more than just Condition~\ref{cond:170623a1} for the existence of an SMD.

\begin{lemma}[Step towards SMD]\label{lem:050324a3}
	Assume Condition~\ref{cond:170623a1} and consider the market $(\mathbb B,\Y)$ with the infinite time horizon $T=\infty$.
	Then there exists a continuous nonnegative process $\Z=(\Z_t)_{t\in\bR_+}$ with $\Z_0=1$, strictly positive on the stochastic interval $[0,\zeta(\Y))$ and stopped at the time $\zeta(\Y)$, such that $\Z$ and $\Z\Y$ are $\mathbf F$-$\P$-local martingales.
\end{lemma}

\begin{proof}
	The process \(\U \triangleq \s (\Y)\) is a regular diffusion on natural scale (\cite[Theorem~V.46.12]{RW2} or \cite[Lemma~B.3]{criensurusov22}).
	As all accessible boundaries of \(\Y\) are absorbing (this a part of Condition~\ref{cond:170623a1}), the process \(\U\) is a continuous local martingale by \cite[Corollary~V.46.15]{RW2} (or \cite[Lemma~B.2]{criensurusov22}).
	To keep our notation simple, we write \(\q \triangleq \s^{-1}\). 
	Recall that the scale function \(\s\) is given by formula~\eqref{eq: nflvrS}.
	Hence, \(\q\) is continuously differentiable with a strictly positive absolutely continuous derivative \(\q' = 1/ \s' (\q)\) on $s(J^\circ)$.
	Further, as \(\beta^2 \in L^1_\textup{loc} (J^\circ)\) by~\eqref{eq: nflvr1}, we get that
	\begin{align} \label{eq:110224a1}
		\big[ \beta (\q) \q'  \big]^2 \in L^1_\textup{loc} (\s (J^\circ) ).
	\end{align}
	For a while we will work with processes defined on the stochastic interval $[0,\zeta(\Y))$.
	Define the process $\theta$ by the formula
	\[
	\theta_t \triangleq \tfrac12\beta(\q(\U_t))\q'(\U_t), \quad t\in[0,\zeta(\Y)).
	\]
	Using the occupation time formula for continuous semimartingales (\cite[Theorem~IV.45.1]{RW2} or \cite[Lemma~B.23]{criensurusov22}) together with~\eqref{eq:110224a1}, we obtain \(\P\)-a.s.
	\begin{align} \label{eq: identity occ time formula}
		\int_0^t \theta^2_s\,d \langle \U, \U\rangle_s
		=
		\int \frac{1}{4} \big(\beta(\q(x))\q'(x)\big)^2 L^x_t (\U)\,dx < \infty, \quad t\in[0, \zeta(\Y)).
	\end{align}
	At this point, we also use that \(x \mapsto L^x_t (U)\) is a.s. c\`adl\`ag and compactly supported (see \cite[Corollary~29.18]{kallenberg}). Thanks to \eqref{eq: identity occ time formula}, 
	we can define
	\begin{align} \label{eq: main Z}
		\Z_t \triangleq \begin{cases}
			\exp \big \{\int_0^t \theta_s\,d\U_s - \frac{1}{2} \int_0^t \theta_s^2\,d \langle \U, \U\rangle_s \big \}, & t\in[0, \zeta(\Y)),
			\\
			\liminf_{s \nearrow \zeta (Y)} \Z_s, & t \in [\zeta(\Y), \infty)
		\end{cases}
	\end{align}
	(notice that $Z$ is strictly positive on $[0,\zeta(\Y))$ and stopped at $\zeta(\Y)$).
	It is well-known
	(cf., for instance, \cite[Lemma~12.43]{Jacod} or \cite[Proposition~A.4]{CTR})
	that \(\Z\) is a continuous local martingale on the whole time interval $\bR_+$.
	
	It remains to show that $\Z\Y$ ($\equiv\Z\q (\U)$) is a local martingale on $\bR_+$.
	Recall that \(\mu_L\) denotes the Lebesgue measure and observe that the second derivative \(\q''\) of \(\q\) exists \(\mu_L\)-a.e. by the absolute continuity of \(\q'\).
	Using~\eqref{eq: nflvrS}, we obtain by a straightforward calculation that
	\begin{equation}\label{eq:110224a2}
		\q'' + \beta(\q) (\q')^2 = 0\quad\mu_L\text{-a.e.}
	\end{equation}
	Using integration by parts and the generalized It\^{o} formula (\cite[Theorem~IV.45.1]{RW2} or \cite[Lemma~B.23]{criensurusov22}), we obtain \(\P\)-a.s., for all \(t \in [0, \zeta(\Y))\),
	\begin{align*}
		d \Z_t \q (\U_t) &= \q (\U_t) d \Z_t + \Z_t d \q (\U_t) +  d \langle \Z, \q (\U) \rangle_t
		\\&= \q (\U_t) d \Z_t + \Z_t d \q (\U_t) +  \Z_t \q' (\U_t) \theta_t d \langle \U, \U\rangle_t
		\\&= \q (\U_t) d \Z_t + \Z_t \q' (\U_t) d\U_t
		+ \tfrac12 \Z_t \Big( \q''(\U_t) + \beta(\q(U_t)) \big(\q'(U_t)\big)^2 \Big) d \langle \U, \U\rangle_t
		\\&= \q (\U_t) d \Z_t + \Z_t \q' (\U_t) d\U_t,
	\end{align*}
	where the last equality follows from~\eqref{eq:110224a2} via the occupation time formula (\cite[Theorem~IV.45.1]{RW2} or \cite[Lemma~B.23]{criensurusov22}).
	Using that \(U\) and \(\Z\) are local martingales stopped at $\zeta(\Y)$, the above formula
	(and recalling \cite[Lemma~12.43]{Jacod} or \cite[Proposition~A.4]{CTR})
	shows that 
	\[
	\Z_t \q (U_t) = \begin{cases} \Z_t \q (U_t), & t \in [0, \zeta(\Y)), \\ 
		\lim_{s \nearrow \zeta (Y)} \Z_s \q (U_s), & t \in [\zeta(\Y), \infty), \end{cases}
	\]
	is a local martingale on $\bR_+$.
	This concludes the proof.
\end{proof}

\begin{remark}\label{rem:050324a1}
	For later reference, with $Z$ as in the above proof, we observe that
	\begin{itemize}
		\item
		$\P$-a.s. on $\{\int_0^{\zeta(\Y)} \theta_s^2\,d\langle U,U\rangle_s=\infty\}$
		we have
		$Z_{\zeta(\Y)}=0$,
		
		\item
		$\P$-a.s. on $\{\int_0^{\zeta(\Y)} \theta_s^2\,d\langle U,U\rangle_s<\infty\}$
		we have
		$Z_{\zeta(\Y)}>0$
	\end{itemize}
	(e.g., see \cite[Exercise~V.1.18]{RY}).
\end{remark}

In the remainder of this section, we provide some preparatory lemmata for the proof that NA or NUPBR imply Condition~\ref{cond:170623a1}.
	
\smallskip 
We need the following result by Bruggeman and Ruf, see \cite[Corollary 1.2]{bruggeman}.
As a referee kindly pointed out, a related result has first appeared in Watanabe's Appendix II from \cite{kotaniwatanabe}. We refer to~\cite{bruggeman} for further discussions.

\begin{lemma}[Diffusion hitting times have full support]\label{lem:BR}
Consider the canonical setting $\mathbb C_x$ with the infinite time horizon $T=\infty$, where
$x\in J$ is \emph{not} an absorbing boundary point.
Let $y\in J\setminus\{x\}$ and $U$ be any nonempty open subset of $(0,\infty)$.
Then, $\P_x(T_y\in U)>0$.
\end{lemma}

Next,
we recall a result by Kardaras and Ruf \cite{kardarasruf}, which shows that NUPBR is stable under filtration shrinkage.

\begin{lemma}[Stability of NUPBR under filtration shrinkage]\label{lem:060324a1}
If NUPBR holds in the market $(\mathbb B,\Y)$ for a finite time horizon $T<\infty$, then NUPBR holds also in the market $(\mathbb B^\Y,\Y)$ for the same time horizon $T$, where $\mathbb{B}^\Y \triangleq (\Omega, \cF, \mathbf{F}^\Y, \P)$.
\end{lemma}

We remark that this does not directly follow from the definition of NUPBR.
A very delicate point is that, in general, the set of trading strategies for the smaller filtration is \emph{not} contained in the set of trading strategies for the larger filtration.
But it turns out that, if NUPBR is satisfied on the larger filtration, then the inclusion mentioned in the previous sentence holds (for a continuous price process, which is the case in our setting).
We refer to \cite[Remark 2.4]{kardarasruf} for more detail.

Now, we discuss the analogous stability result for NA.

\begin{lemma}[Stability of NA under filtration shrinkage]\label{lem:210324a1}
If NA holds in the market $(\mathbb B,\Y)$ for a finite time horizon $T<\infty$, then NA holds also in the market $(\mathbb B^\Y,\Y)$ for the same time horizon~$T$.
\end{lemma}

It is worth noting that, because of the delicate point described after Lemma~\ref{lem:060324a1} or, in more detail, in \cite[Remark 2.4]{kardarasruf},
the claim of Lemma~\ref{lem:210324a1} looks a priori rather unclear.
In our setting, the claim follows from the Kabanov--Stricker and Strasser characterization of NA that was recalled in Theorem~\ref{th:210324a1}.

\begin{proof}[Proof of Lemma~\ref{lem:210324a1}]
The claim in Lemma~\ref{lem:210324a1} follows from Theorem~\ref{th:210324a1} together with the fact that a \emph{continuous} local martingale (say, $M$) w.r.t. some filtration (say, $\mathbf G$) remains a local martingale w.r.t. any filtration $\mathbf H$ that lies between the natural filtration of $M$ and the filtration~$\mathbf G$. 
\end{proof}

The following two lemmata relate the no-arbitrage notions NUPBR and NA to the last requirement in Condition~\ref{cond:170623a1}.

\begin{lemma}[Under NUPBR accessible boundaries are absorbing]\label{lem:060324a2}
If NUPBR holds in the market $(\mathbb B,\Y)$ for a finite time horizon $T<\infty$, then every $b\in\{l,r\}\cap\bR$ is either inaccessible or absorbing for~$\Y$.
\end{lemma}

\begin{proof}
Thanks to Theorem~\ref{theo: FTAP}, there exists an SMD \(\Z\).
For any \(b \in \{l, r\} \cap \bR\), the process \(\Z | \Y - b |\) is a nonnegative \(\mathbf{F}\)-\(\P\)-local martingale and hence, a nonnegative \(\mathbf{F}\)-\(\P\)-supermartingale.
Assume that $b$ is accessible.
By Lemma~\ref{lem:BR},
with positive probability, $\Y$ hits $b$ prior to time $T$.
As nonnegative supermartingales cannot resurrect from zero (\cite[Lemma~III.3.6]{JS}), \(b\) has to be absorbing.
\end{proof}

\begin{lemma}[Under NA accessible boundaries are absorbing]\label{lem:060324a4}
If NA holds in the market $(\mathbb B,\Y)$ for a finite time horizon $T<\infty$, then every $b\in\{l,r\}\cap\bR$ is either inaccessible or absorbing for~$\Y$.
\end{lemma}

\begin{proof}
Assume for contradiction that the right boundary point $r$ is finite and (instantaneously or slowly) reflecting.
Let \(\{ \P_x \colon x \in J\}\) be the diffusion with scale function \(\s\) and speed measure \(\m\) on the canonical space (cf. \cite[Definition~V.45.1]{RW2}).
Fix \(l<y_0 < z_0 < r\) and set 
\[
\tau \triangleq \inf \{t > T_r \colon \canY_t = y_0 \}, \quad \rho \triangleq \inf \{ t > \tau \colon \canY_t = z_0\}. 
\]
Lemma~\ref{lem:BR} yields that 
$\P_{x_0}(T_r<T/2)>0$,
$\P_r(T_{y_0}<T/2)>0$
and
$\P_{y_0}(T_{z_0}>T)>0$.
By the strong Markov property, using the notation $\theta_{T_r}$ for the shift operator (cf. \cite[Section III.3]{RY}),
we first compute
\begin{align*} 
\P_{x_0} ( \tau < T)
&= \P_{x_0} (T_r < T , T_{y_0} (\theta_{T_r}) + T_r < T) 
\\&\ge \P_{x_0} (T_r < T / 2, T_{y_0} (\theta_{T_r}) < T/2) 
\\&=
\E^{\P_{x_0}} \big[ \1_{\{T_r < T/2\}} \P_r (T_{y_0} < T/2) \big]
\\&=
\P_{x_0}(T_r<T/2)
\P_r(T_{y_0}<T/2)
>0
\end{align*}
and then, similarly,
$$
\P_{x_0} ( \tau < T, \rho > T )
\ge
\E^{\P_{x_0}} \big[ \1_{\{\tau < T\}} \P_{y_0} (T_{z_0} > T) \big]
=
\P_{x_0} ( \tau < T)
\P_{y_0}(T_{z_0}>T)
>0.
$$
Hence,
\[
\P (\tau (Y) < T, \rho (Y) > T) > 0.
\]
Now, consider the strategy \(H \triangleq - \1_{(T_r (\Y)\wedge T , T]}\). For all \(t \in [0, T]\), we have 
\begin{align*}
	V^H_t = \Y_{t \wedge T_{r} (Y)} - \Y_t \geq 0.
\end{align*}
Moreover, on \(\{ \tau (Y) < T, \rho (Y) > T\}\), 
\[
V^H_T = \Y_{T \wedge T_r (\Y)} - \Y_T = r - \Y_T \geq r - z_0 > 0. 
\]
This shows that we can realize arbitrage via the
admissible
strategy \(H\). Consequently, \(r\) must be accessible or absorbing. 
The left boundary point $l$ can be treated in a similar way.
\end{proof}

The next two lemmata discuss the role of initial values for the no-arbitrage notions.
This is naturally done in the canonical setting because there we have $\P_x$ for all $x\in J$.
It is worth emphasizing that passing from some $\P_x$ to another $\P_y$ changes the initial value, while the scale function and the speed measure remain unchanged.
We emphasize that \({\P_x \circ (\canY + y - x)^{-1} \ne \P_y}\)
unless the underlying diffusion is a process with stationary independent increments
(see \cite[Theorem 11.10]{kallenberg}).
In our general diffusion setting the latter can be only a scaled Brownian motion with drift
(see \cite[Theorem 14.4]{kallenberg}).
This means that changing the initial values while keeping scale and speed is indeed a complicated transformation that requires a thorough discussion.

\begin{lemma}[Changing initial values under NUPBR]\label{prop: equivalence initial value}
	Take an initial value \(x_0 \in J^\circ\) and assume that the market \((\mathbb{C}_{x_0}, \canY)\)  satisfies NUPBR for a finite time horizon \(T<\infty\).
	Then, for every \(y_0 \in J^\circ\) and every \(\varepsilon \in (0, T)\), the market \((\mathbb{C}_{y_0}, \canY)\) satisfies NUPBR for the time horizon \(T - \varepsilon\). 
\end{lemma}

\begin{proof}
In this proof we work on the canonical space $\mathbb C_{x_0}$ with time horizon $T$.
In particular, the notation $\canF_t$ means the respective $\sigma$-field on the space with time horizon $T$.
The transition to the canonical space with time horizon $T-\varepsilon$ happens in the very end
(directly after~\eqref{eq: shift}).
This does not create notational ambiguities.

Take some \(y_0 \in J^\circ \setminus \{x_0\}\) and \(\varepsilon \in (0, T)\).
Thanks to Theorem~\ref{theo: FTAP}, there exists a strictly positive \(\canbfF\)-\(\P_{x_0}\)-local martingale \(\canZ\) with \(\canZ_0 = 1\) such that \(\canZ\,\canY\) is an \(\canbfF\)-\(\P_{x_0}\)-local martingale.
 By \cite[Theorem 10.16]{Jacod}, the time-changed processes \(\canZ_{\cdot + T_{y_0} \wedge \varepsilon}\) and \(\canZ_{\cdot + T_{y_0} \wedge \varepsilon} \canY_{\cdot + T_{y_0} \wedge \varepsilon}\) are
 \(\mathbf{G}\)-\(\P_{x_0}\)-local martingales on the time interval \([0, T - \varepsilon]\), where \(\mathbf{G} = (\mathcal{G}_t)_{t \in [0, T - \varepsilon]}\) with \(\mathcal{G}_t \triangleq \canF_{t + T_{y_0} \wedge \varepsilon}\).
Here and below we understand expressions like
$t+T_{y_0}\wedge\varepsilon$
as
$t+(T_{y_0}\wedge\varepsilon)$.
By Lemma~\ref{lem:BR},
\(\P_{x_0} (T_{y_0} \leq \varepsilon) > 0\). Thus, we may define a probability measure \(\Q\) by 
 \[
 \Q(G) \triangleq \P_{x_0} ( G \mid  T_{y_0} \leq \varepsilon) = \frac{ \P_{x_0} (G \cap \{T_{y_0} \leq \varepsilon\})}{\P_{x_0} (T_{y_0} \leq \varepsilon) }, \quad G \in \canF.
 \]
 As \(\{T_{y_0} \leq \varepsilon\} \in \mathcal{G}_0 = \canF_{T_{y_0} \wedge \varepsilon}\), the \(\mathbf{G}\)-\(\P_{x_0}\)-local martingale property is not affected by a change from \(\P_{x_0}\) to the conditional probability \(\Q\). 
Consequently, \(\canZ_{\cdot + T_{y_0} \wedge \varepsilon} / \canZ_{T_{y_0} \wedge \varepsilon}\) is an SMD for the market \(((\canOmega, \canF, \mathbf{G},\Q), \canY_{\cdot + T_{y_0} \wedge \varepsilon})\)
with time horizon $T-\varepsilon$
and, thanks again to Theorem~\ref{theo: FTAP}, it satisfies the NUPBR condition.
By Lemma~\ref{lem:060324a1}, also the market 
\begin{equation}\label{eq:150724a1}
\big((\canOmega, \canF, \mathbf{F}^{\canY_{\cdot + T_{y_0} \wedge \varepsilon}},\Q), \canY_{\cdot + T_{y_0} \wedge \varepsilon}\big)
\end{equation}
with time horizon $T-\varepsilon$
satisfies NUPBR.
Now, consider a bounded $\canF ^o_{T-\varepsilon}$-measurable path functional
$F\colon C([0,T-\varepsilon],J)\to\bR$.
By the strong Markov property of the family $(\P_x)_{x\in J}$,
$$
\E^{\P_{x_0}}\big[F(\canY_{\cdot+T_{y_0}\wedge\varepsilon})\mid\canF_{T_{y_0}\wedge\varepsilon}\big]
=
\E^{\P_{y_0}}\big[F(\canY)\big]
\quad\P_{x_0}\text{-a.s. on }\{T_{y_0}\le\varepsilon\}
$$
(see \cite[Theorem~III.3.1]{RY}).
Multiplying both sides with $\1_{\{T_{y_0}\le\varepsilon\}}$ and computing the expectation under $\P_{x_0}$ yields
$$
\E^\Q\big[F(\canY_{\cdot+T_{y_0}\wedge\varepsilon})\big]
=
\E^{\P_{y_0}}\big[F(\canY)\big],
$$
which implies
\begin{align} \label{eq: shift}
\Q \circ \canY_{\cdot + T_{y_0} \wedge \varepsilon}^{-1} = \P_{y_0}
\text{ on }\canF^o_{T - \varepsilon}.
\end{align}
As market~\eqref{eq:150724a1} satisfies NUPBR
for the time horizon \(T - \varepsilon\),
we deduce from \cite[Proposition~10.38~(b)]{Jacod} that
the market \((\mathbb{C}_{y_0}, \canY)\) satisfies NUPBR
for the time horizon \(T - \varepsilon\).
The proof is complete.
\end{proof}

\begin{lemma}[Changing initial values under NA]\label{lem:060324a5}
	Take an initial value \(x_0 \in J^\circ\) and assume that the market \((\mathbb{C}_{x_0}, \canY)\)  satisfies NA for a finite time horizon \(T<\infty\).
	Then, for every \(y_0 \in J^\circ\) and every \(\varepsilon \in (0, T)\), the market \((\mathbb{C}_{y_0}, \canY)\) satisfies NA for the time horizon \(T - \varepsilon\). 
\end{lemma}

\begin{proof}
In this proof we partly need to work on the canonical space with time horizon~$T$ and partly on the canonical space with time horizon $T-\varepsilon$.
To ease our presentation, in this proof we understand
the path space \(\canOmega\) as the canonical space \(C(\bR_+; J)\) for the {\em infinite} time horizon that covers all time horizons simultaneously. Further, \(\canY\) denotes the coordinate process and \(\canbfF = (\canF^o_{t +})_{t \geq 0}\) denotes the right-continuous canonical filtration on the path space \(\canOmega = C (\bR_+; J)\). Finally, for any \(T > 0\), we set \(\canF^T_s \triangleq \canF_s \cap \canF^o_T\) and \(\canbfF^T \triangleq (\canF^T_t)_{t \geq 0}\).

	Take \(x_0, y_0 \in J^\circ, T < \infty\) and \(\varepsilon \in (0, T)\), and suppose that NA holds for the market \(((\canOmega, \canF, \canbfF^T, \P_{x_0}), (\canY)_{t \in [0, T]})\) and the time horizon \(T\).
With \(\P \equiv \P_{x_0}\), notice that $\P(T_{y_0}\le\varepsilon)>0$ by Lemma~\ref{lem:BR},
and define 
	\[
	\P^* \triangleq \P ( \theta_{T_{y_0}}^{-1} (\, \cdot \, ) \mid T_{y_0} \leq \varepsilon).
	\]
	Thanks to~\eqref{eq: shift},
	we have $\P^*=\P_{y_0}$ on $\canF^o_{T-\varepsilon}$. Consequently, the claim of the lemma follows once we show that NA holds for the market \(((\canOmega, \canF, \canbfF^{T- \varepsilon}, \P^*), (\canY)_{t \in [0, T - \varepsilon]})\). 
	Our tactic is to use Theorem~\ref{th:210324a1}. 
	Let \(\sigma\) be an \(\canbfF^{T - \varepsilon}\)-stopping time such that \(\sigma \leq T - \varepsilon\). We set 
	\[
	\rho \triangleq \sigma (\theta_{T_{y_0} \wedge \varepsilon}) + T_{y_0} \wedge \varepsilon,
	\]
	where $\theta_{T_{y_0} \wedge \varepsilon} \colon \canOmega \to \canOmega$ denotes the shift operator (see \cite[Section III.3]{RY}).
	Notice that \(\rho\) is an \(\canbfF^T\)-stopping time by \cite[Proposition 11.8]{kallenberg}. Thanks to Theorem~\ref{th:210324a1}, there exists a probability measure \(\Q \equiv {^\rho}\Q\) such that \(\Q \sim \P\) on \(\canF_\rho^T\), \(\Q \ll \P\) on
\(\canF^o_T\)
and such that \((\canY_{t \vee \rho} - \canY_\rho)_{t \in [0, T]}\) is an \(\canbfF^T\)-\(\Q\)-local martingale.
	Clearly, we have 
	\[
	\{ T_{y_0} \leq \varepsilon \} \in \canF_{T_{y_0} \wedge \varepsilon}^o \subset \canF_\rho^T.
	\]
As \(\P (T_{y_0} \leq \varepsilon ) > 0\) and \(\Q \sim \P\) on \(\canF_\rho^T\), it also holds that \(\Q (T_{y_0} \leq \varepsilon) > 0\). We can, therefore,
define the probability measure
	\begin{align*}
		\Q^* &\triangleq \Q ( \theta_{T_{y_0}}^{-1} (\, \cdot \, ) \mid T_{y_0} \leq \varepsilon)
	\end{align*}
on $(\canOmega,\canF)$.
	Noting that \(\theta^{-1}_{T_{y_0} \wedge \varepsilon} (\canF_\sigma^{T - \varepsilon}) \subset \canF_\rho^T\), \(\Q^* \sim \P^*\) on \(\canF_\sigma^{T - \varepsilon}\) follows from \(\P \sim \Q\) on \(\canF^T_\rho\). Furthermore, it is clear that \(\Q^* \ll \P^*\) on \(\canF^o_{T-\varepsilon}\).
We now show that \((\canY_{t \vee \sigma} - \canY_\sigma)_{t \in [0, T - \varepsilon]}\) is an $\canbfF^{T - \varepsilon}$-\(\Q^*\)-local martingale.
For \(N \in\mathbb N\), set 
	\begin{align*}
	\tau \equiv \tau_N &\triangleq \inf \{ t \geq 0 \colon |\canY_{t \vee \sigma} - \canY_\sigma| \geq N\}, \\
	\tau^* \equiv \tau^*_N &\triangleq \tau (\theta_{T_{y_0} \wedge \varepsilon}) + T_{y_0} \wedge \varepsilon = \inf \{ t \geq 0 \colon |\canY_{(t + T_{y_0} \wedge \varepsilon) \vee \rho} - \canY_{\rho}| \geq N \} + T_{y_0} \wedge \varepsilon.
	\end{align*} 
	For \(0 \leq s < t \leq T - \varepsilon\) and \(A \in \canF_s^{T - \varepsilon}\), we have 
	\begin{align*}
		\E^{\Q^*} \Big[ \1_A (\canY_{t \wedge \tau \vee \sigma} - \canY_\sigma) \Big] &= \E^{\Q} \Big[ \1_{\theta^{-1}_{T_{y_0}} (A)} (\canY_{(t + T_{y_0}) \wedge \tau^* \vee \rho} - \canY_{\rho}) \mid T_{y_0} \leq \varepsilon \Big]
			\\&= \E^{\Q} \Big[ \1_{\theta^{-1}_{T_{y_0}} (A)} (\canY_{(s + T_{y_0}) \wedge \tau^* \vee \rho} - \canY_{\rho}) \mid T_{y_0} \leq \varepsilon \Big]
			\\&= \E^{\Q^*} \Big[ \1_A (\canY_{s \wedge \tau \vee \sigma} - \canY_\sigma) \Big]
	\end{align*}
by the optional stopping theorem. Here, we used that \((t + T_{y_0} \wedge \varepsilon) \wedge \tau^*\) and \((s + T_{y_0} \wedge \varepsilon) \wedge \tau^*\) are \(\canbfF^T\)-stopping times\footnote{To see this, notice that 
	\[
	\gamma \triangleq (t + T_{y_0} \wedge \varepsilon) \wedge \tau^* = (\tau \wedge t\hspace{0.02cm}) (\theta_{T_{y_0} \wedge \varepsilon}) + T_{y_0} \wedge \varepsilon. 
	\] 
	For \(r \in [0, T)\), \(\{\gamma \leq r\} \in \canF_r\) follows from \cite[Proposition~11.8]{kallenberg}, and \(\{\gamma \leq T\} = \canOmega \in \canF^o_T\).},
\(\theta^{-1}_{T_{y_0}} (A) \cap \{T_{y_0} \leq \varepsilon\} \in \canF_{s + T_{y_0} \wedge \varepsilon}^T\) and that \(\canY_{(\, \cdot\, + T_{y_0}) \wedge \tau^* \vee \rho} - \canY_\rho\) is bounded by the definition of \(\tau^*\). 
Because \(\tau_N \nearrow \infty\) as \(N \nearrow \infty\), we proved that \((\canY_{t \vee \sigma} - \canY_\sigma)_{t \in [0, T - \varepsilon]}\) is an $\canbfF^{T - \varepsilon}$-\(\Q^*\)-local martingale. 
Finally, we deduce from Theorem~\ref{th:210324a1} that NA holds for the market \(((\canOmega, \canF, \canbfF^{T- \varepsilon}, \P^*), (\canY)_{t \in [0, T - \varepsilon]})\) and consequently,
the proof is complete.
\end{proof}

Next, we derive another useful technical result (Lemma~\ref{lem:050324a1} below).
To this end, we first consider an open interval $I\subset\bR$ and recall (\cite[Proposition 5.1]{CinJPrSha}) that the following are equivalent:
\begin{enumerate}
	\item[(a)]
	$f\colon I\to\bR$ is the difference of two convex functions $I\to\bR$;
	
	\item[(b)]
	$f\colon I\to\bR$ is continuous and has a right-continuous right-hand derivative $f'_+$ on $I$, and $f'_+$ has locally finite variation on~$I$;
	
	\item[(c)]
	$f\colon I\to\bR$ is continuous and has a left-continuous left-hand derivative $f'_-$ on $I$, and $f'_-$ has locally finite variation on~$I$.
\end{enumerate}

\begin{lemma}\label{lem:050324a2}
Let $f\colon J^\circ\to\bR$ be strictly increasing and the difference of two convex functions on $J^\circ$.
Define the measure $\tm$ on $J^\circ$ by the formula $d\tm\triangleq f'_+\,d\m$.
Then, $\tm$ is a valid speed measure on $J^\circ$.
\end{lemma}

\begin{proof}
We need to check that, for all $a<b$ in $J^\circ$, $\tm([a,b])\in(0,\infty)$.
As $f'_+$ is bounded on $[a,b]$ and $\m$ is a speed measure on $J^\circ$, we immediately get that $\tm([a,b])<\infty$.
For the strict positivity, notice that there is some $c\in(a,b)$ such that $f'_+(c)>0$ (otherwise $f$ would be constant on $[a,b]$).
By the right-continuity of $f'_+$, there is some $d\in(c,b)$ such that $f'_+>0$ on $[c,d]$, hence $\tm([a,b])\ge\tm([c,d])>0$.
Thus, $\tm$ is a valid speed measure.
\end{proof}

\begin{lemma}[Technical lemma]\label{lem:050324a1}
Assume that the scale function $\s$ is the difference of two convex functions on $J^\circ$.
We work on the filtered probability space $\mathbb C_{x_0}$ with the infinite time horizon.
Let $\xi$ be an $\canbfF$-predictable time that is
strictly positive (i.e., $\xi(\omega)>0$ for all $\omega\in\canOmega$)
and let $\Q$ be a probability measure on $(\canOmega,\canF)$ such that $\Q\ll\P_{x_0}$ on $\canF^o_\xi$ and $\canY_{\cdot\wedge\xi}$ is an $\canbfF$-$\Q$-local martingale.
Then,
$$
\Q=\tP_{x_0}\text{ on }\canF^o_\xi,
$$
where $\tP_{x_0}$ is the law of the diffusion
on natural scale
started in $x_0$
with the interior of the state space $J^\circ$,
the speed measure $\s'_+\,d\m$ on $J^\circ$
and the boundary points $l$ and $r$ being absorbing whenever they are accessible.
\end{lemma}

To see that the \(\sigma\)-field \(\canF^o_\xi\) in the statement of Lemma~\ref{lem:050324a1} is well-defined, recall from \cite[III.2.36]{JS} that strictly positive \(\canbfF\)-predictable times are \(\canbfF^o\)-stopping times. Further, recall that $\s'_+\,d\m$ is a valid speed measure on $J^\circ$ by Lemma~\ref{lem:050324a2}.

\begin{proof}
Take a strictly decreasing sequence $(y_n)_{n\in\mathbb N}\subset J^\circ$ and a strictly increasing sequence $(z_n)_{n\in\mathbb N}\subset J^\circ$ with $y_1<x_0<z_1$ and $y_n\searrow l$, $z_n\nearrow r$.
Fix some $n\in\mathbb N$ and define the $\canbfF^o$-stopping time \(\eta_n \triangleq \xi \wedge T_{y_n} \wedge T_{z_n}\).
Set $\tP^n_{x_0}\triangleq\tP_{x_0}\circ\canY_{\cdot\wedge T_{y_n}\wedge T_{z_n}}^{-1}$.
In other words, $\tP^n_{x_0}$ is the law of the diffusion
on natural scale
started in $x_0$
with the state space $[y_n,z_n]$, speed measure $\s'_+\,d\m$ on $(y_n,z_n)$
and absorbing boundaries $y_n$ and $z_n$.
Take a function \(f \in C ([y_n, z_n]; \bR)\) such that \(f\) is the difference of two convex functions on \((y_n, z_n)\) and \(d f'_+ = 2 g \s'_+\,d \m\) on \((y_n, z_n)\) for some \(g \in C([y_n, z_n]; \bR)\) with \(g (y_n) = g (z_n) = 0\). Here, \(d f'_+\) denotes the signed measure induced by the right-hand derivative \(f'_+\) of \(f\) via the formula
\[
d f'_+ ( (a, b] ) = f'_+ (b) - f'_+ (a), \quad a < b \text{ in } (y_n, z_n).
\]
Let \(\{L^x_t (\canY) \colon (t,x) \in [0, T] \times \bR\}\) be the (continuous in time and right-continuous in space)
semimartingale local time of the coordinate process \(\canY\) under \(\P_{x_0}\).
We compute that \(\P_{x_0}\)-a.s., for all \(t \in [0, \eta_n)\),
\begin{equation} \label{eq: iden identification}
\begin{split}
f (\canY_t) &= f(x_0)+\int_0^{t} f'_- (\canY_s) \,d \canY_s + \frac{1}{2} \int_{y_n}^{z_n} L^x_t (\canY) 2g (x) \s'_+ (x)\,\m (dx)
\\
&= 
f(x_0)+\int_0^{t} f'_- (\canY_s) \,d \canY_s + \int_{y_n}^{z_n} L^{\s(x)}_t (\s (\canY)) g (x)\,\m (dx)
\\
&=
f(x_0)+\int_0^{t} f'_- (\canY_s) \,d \canY_s + \int_{\s(y_n)}^{\s (z_n)} L^{x}_t (\s (\canY)) g (\s^{-1} (x))\,\m \circ \s^{-1} (dx)
\\
&=
f(x_0)+\int_0^{t} f'_- (\canY_s) \,d \canY_s + \int_0^{t} g (\s^{-1} (\s(\canY_s))) \,ds, 
\\
&=
f(x_0)+\int_0^{t} f'_- (\canY_s) \,d \canY_s + \int_0^{t} g (\canY_s)\,ds, 
\end{split}
\end{equation}
where we use the generalized It\^o formula (\cite[Theorem~IV.45.1]{RW2} or \cite[Lemma B.23]{criensurusov22}) in the first and second line and \cite[Exercise~VII.3.18]{RY} (\cite[Lemma B.3]{criensurusov22}) and \cite[Theorem~V.49.1]{RW2} (\cite[Lemma~B.14]{criensurusov22}) in the fourth line. 
As \(\mathsf{Q} \ll \P_{x_0}\) on $\canF^o_\xi$, by hypothesis, and $\eta_n\le\xi$, the identities in~\eqref{eq: iden identification} also holds up to a \(\mathsf{Q}\)-null set.
Thanks to~\eqref{eq: iden identification}, which holds on the predictable stochastic interval $[0, \eta_n)$,
and the fact that $\canY$ is a continuous \(\canbfF\)-\(\mathsf{Q}\)-local martingale on $[0,\eta_n)$, we conclude that the process
\begin{equation}\label{eq:150324a1}
f (\canY_{\cdot\wedge\eta_n}) - f (x_0) - \int_0^{\cdot\wedge\eta_n} g (\canY_s) \,ds
\end{equation}
is also a continuous \(\canbfF\)-\(\mathsf{Q}\)-local martingale on $[0,\eta_n)$.
As the process in~\eqref{eq:150324a1} is bounded on any finite time interval ($f$ and $g$ are bounded), \cite[Proposition 5.9]{Jacod} yields that it is an \(\canbfF\)-\(\mathsf{Q}\)-martingale for the whole time index set $\bR_+$.
By \cite[Lemmata~A.4,~A.5]{criensurusov22}, we get
$\mathsf{Q} = \tP^n_{x_0} \text{ on } \canF^o_{\eta_n}$.
Galmarino's test in the form of \cite[Exercise I.4.21]{RY} show that
\[\canY_{\cdot\wedge T_{y_n}\wedge T_{z_n}}^{-1}(\canF^o_{\eta_n})=\canF^o_{\eta_n}.\]
It follows that
\begin{equation}\label{eq:150324a2}
\Q=\tP_{x_0}\text{ on }\canF^o_{\eta_n}.
\end{equation}
As $\canY_{\cdot\wedge\xi}$ is an $\canbfF$-$\Q$-local martingale, it follows that, if $\Q(T_l<\xi)>0$ (resp., $\Q(T_r<\xi)>0$),
then, on the event $\{T_l<\xi\}$ (resp., $\{T_r<\xi\}$),
$\canY$ stays in $l$ (resp., in $r$) on $[T_l,\xi]$ (resp., $[T_r,\xi]$).
Together with the fact that $l$ and $r$ are inaccessible or absorbing for $\canY$ under $\tP_{x_0}$ (by the definition of $\tP_{x_0}$), this yields that
\begin{equation}\label{eq:150324a3}
(\Q+\tP_{x_0})\text{-a.e. }
\bigvee_{n\in\mathbb N}\canF^o_{\eta_n}=\canF^o_\xi.
\end{equation}
The statement now follows from \eqref{eq:150324a2} and~\eqref{eq:150324a3}.
The proof is complete.
\end{proof}

Lastly, we provide a step towards Condition~\ref{cond:050324a1}.
It is instructive to compare the next result to Lemma~\ref{lem:050324a3}.

\begin{lemma}[Step towards Condition~\ref{cond:050324a1}]\label{lem:060324a3}
Take a finite time horizon $T<\infty$.

\begin{enumerate}
	\item[\textup{(i)}] 
For some $x\in J^\circ$, consider the market $(\mathbb C_x,\canY)$ and assume that there exists a nonnegative c\`adl\`ag $\canbfF$-adapted process $\canZ=(\canZ_t)_{t\in[0,T]}$ with $\canZ_0=1$ such that $\canZ$ and $\canZ\,\canY$ are $\canbfF$-$\P_x$-local martingales.
Then, Condition~\ref{cond:050324a1} holds in a sufficiently small open neighborhood of the point~$x$.
\item[\textup{(ii)}]If the assumption in part~(i) holds for every $x\in J^\circ$, then Condition~\ref{cond:050324a1} is satisfied.
\end{enumerate}
\end{lemma}

\begin{proof}
(i)
Let \((S_n)_{n = 1}^\infty\) be an \(\canbfF\)-\(\P_{x}\)-localizing sequence for the local martingale \(\canZ\).
As \(\P_{x}\)-a.s. \(S_n \to \infty\), there exists an \(N \in \mathbb{N}\) such that \(\P_{x}(S_N > 0) > 0\) and, by Blumenthal's zero-one law (\cite[Lemma~4 on p. 106]{freedman} or \cite[Lemma~B.1]{criensurusov22}), it even holds that \(\P_{x} (S_N > 0) = 1\).
Furthermore, by Meyer's theorem on predictability (see \cite[Proposition~4]{chungwalsh} and \cite[Lemma~I.2.17]{JS}, or \cite[Lemma~B.20]{criensurusov22}), there exists an \(\canbfF\)-predictable time \(\xi\) such that \(\P_{x}\)-a.s. \(S_N = \xi\).
We may take \(\xi > 0\) identically. 
Define the probability measure \(\mathsf{Q}\) via the Radon--Nikodym density \(d \Q / d \P_{x} \triangleq \canZ_{T \wedge \xi}\).
As \(\canZ\,\canY\) is an \(\canbfF\)-\(\P_{x}\)-local martingale, \cite[Proposition III.3.8]{JS} yields that \(\canV \triangleq \canY_{\cdot \wedge \xi}\) is an \(\canbfF\)-\(\mathsf{Q}\)-local martingale.
We define the time-change
\[
L (t) \triangleq \inf \{s \in [0, T]\colon \langle \canV, \canV \rangle_s > t\} \wedge T, \quad t \in \bR_+,
\]
where \(\langle \canV, \canV \rangle\) denotes the \(\canbfF\)-\(\mathsf{Q}\)-quadratic variation process of \(\canV\).
Without loss of generality we assume that, under $\P_x$, the process $\canV$ is absorbed at the boundaries $l$ and $r$ whenever it hits a boundary point (if it is not the case, replace $\xi$ with $\xi\wedge\zeta$).
Then, the space-transformed process \(\s(\canV)\) is an \(\canbfF\)-\(\P_{x}\)-local martingale (\cite[Corollary~V.46.15]{RW2} or \cite[Lemma B.2]{criensurusov22}).
Hence, by Girsanov's theorem, the process \(\s (\canV)\) is an \(\canbfF\)-\(\mathsf{Q}\)-semimartingale.
Using \cite[Corollary~10.12, Theorem~10.16]{Jacod} (or \cite[Lemma B.26]{criensurusov22}), we observe that the time-changed process
\(
\s (\canV_{L})
\)
is an \((\canF_{L (t)})_{t \geq 0}\)-\(\mathsf{Q}\)-semimartingale.
By the Doeblin, Dambis, Dubins--Schwarz theorem (\cite[Theorem~V.1.7]{RY}), \(\canV_L\) is a standard \((\canF_{L(t)})_{t \geq 0}\)-\(\mathsf{Q}\)-Brownian motion stopped at \(\langle \canV, \canV\rangle_{T}\).
It follows from the fact that continuous local martingales and their quadratic variation processes have the same intervals of constancy (\cite[Proposition~IV.1.13]{RY}) and \cite[Corollary~5, Fact 6 on p.~107, Lemma~12 on p.~109]{freedman} (or \cite[Lemma~B.5]{criensurusov22}) that \(\mathsf{Q}\)-a.s. \(\langle \canV, \canV\rangle_T = \langle \canY, \canY\rangle_{\xi \wedge T} > 0\) and \(L (0) = 0\).
Now, we deduce from \cite[Theorem~B.28]{criensurusov22} that, in a sufficiently small open neighborhood $(\alpha,\beta)$ of $x$, the scale function \(\s\) is the difference of two convex functions.
That is, we can apply Lemma~\ref{lem:050324a1} for the diffusion $\canY_{\cdot\wedge T_{\alpha}\wedge T_{\beta}}$ and the strictly positive $\canbfF$-predictable time $\eta\triangleq T\wedge\xi\wedge T_{\alpha}\wedge T_{\beta}$ concluding that
$\Q=\tP_x$ on $\canF^o_\eta$,
where $\tP_{x}$ is the law of the diffusion
on natural scale
started in $x$
with \((\alpha, \beta)\) as interior of the state space, speed measure $\s'_+\,d\m$ on $(\alpha,\beta)$
and the boundary points $\alpha$ and $\beta$ being absorbing whenever they are accessible.
As $\eta>0$, we have $\canF_0\subset\canF_{\eta-}\subset\canF^o_\eta$ (see \cite[III.2.36]{JS} for the second inclusion).
Thus,
$$
\tP_x\ll\P_x\quad\text{on}\quad\canF_0.
$$
Now, by \cite[Corollary 2.17]{criensurusov22}, there exists a sufficiently small open neighborhood \(U (x)\subset(\alpha,\beta)\) of \(x\) and a Borel function \(\beta \colon U (x) \to \bR\) such that \(\beta \in L^2(U(x))\),
the scale function $\s$ is continuously differentiable in $U(x)$ with an absolutely continuous derivative $\s'$ and
\[
d \s' (a) = \s' (a) \beta (a)\,da \quad\text{on}\quad U(x).
\]
This means that Condition~\ref{cond:050324a1} is satisfied in this neighborhood $U(x)$.

(ii) The second part follows from the first one and the fact that Condition~\ref{cond:050324a1} is a local property.
\end{proof}

\subsection{Proof of Lemma~\protect\ref{lem:170623a1}} \label{sec: pf lem:170623a1}

For the reader's convenience, we repeat the formulation.

\medskip\noindent
\textbf{Lemma~\ref{lem:170623a1}.}
{\em Assume that Condition~\ref{cond:050324a1} holds and let \(b \in \{l, r\} \cap \bR\).
\begin{enumerate}
\item[\textup{(i)}]
If $|\s(b)|=\infty$, then \eqref{eq: b good} for $b$ is violated.

\item[\textup{(ii)}]
If \eqref{eq: b good} holds for $b$, then $|\s(b)|<\infty$.

\item[\textup{(iii)}]
Suppose that one of the conditions \textup{(iii.a)--(iii.b)} below is satisfied:
\begin{enumerate}
\item[\textup{(iii.a)}]
the boundary point $b$ is accessible for $\Y$ and \eqref{eq: very good add cond} holds;
\item[\textup{(iii.b)}]
the boundary point $b$ is inaccessible for $\Y$ and \eqref{eq: very good add cond} is violated.
\end{enumerate}
Then, \eqref{eq: b good} for $b$ is violated.

\item[\textup{(iv)}]
If \eqref{eq: b good} holds for $b$, then one of the conditions \textup{(iv.a)--(iv.b)} below is satisfied:
\begin{enumerate}
\item[\textup{(iv.a)}]
the boundary point $b$ is inaccessible for $\Y$ and \eqref{eq: very good add cond} holds;
\item[\textup{(iv.b)}]
the boundary point $b$ is accessible for $\Y$ and \eqref{eq: very good add cond} is violated.
\end{enumerate}
\end{enumerate} }

\begin{proof}
We first observe that the statements in (i) and~(ii) are equivalent.
Also (iii) and~(iv) are equivalent.
Therefore, it is enough to prove (i) and~(iv), which is done below.

Essentially, the claims can be deduced from observations made in our previous paper~\cite{criensurusov22}.
We explain the details, where we use the terminology introduced in \cite{criensurusov22}.
To this end, we assume Condition~\ref{cond:050324a1} and consider the market $(\mathbb B,\Y)$ with the infinite time horizon.
In this proof, we do not work with $\Y$ but $\Y_{\cdot\wedge\zeta(\Y)}$,
which is the diffusion with the starting point $x_0\in J^\circ$, state space~$J$, scale function $\s$ and speed measure $\m$ on $J^\circ$.\footnote{If we compare the characteristics of the diffusions $\Y$ and $\Y_{\cdot\wedge\zeta(\Y)}$, only the values of the speed measures at the accessible boundaries can differ.}
In addition, we consider (possibly, on another probability space) a diffusion $\wtY$ on natural scale started in $x_0$ with the interior of the state space $J^\circ$, speed measure $\s'\,d\m$ on $J^\circ$ and boundary points $l$ and $r$ being absorbing whenever they are accessible.
(Notice that the derivative $\s'$ is well-defined by Condition~\ref{cond:050324a1} and that $\s'\,d\m$ is a valid speed measure on $J^\circ$ by Lemma~\ref{lem:050324a2}.)
We make the following observations.
\begin{itemize}
\item
Condition~\ref{cond:050324a1} means that all points from $J^\circ$ are \emph{good} for the diffusions $\Y_{\cdot\wedge\zeta(\Y)}$ and $\wtY$ in the sense of \cite[Definition 2.5]{criensurusov22}.

\item
Under Condition~\ref{cond:050324a1},
a finite boundary point $b\in\{l,r\}\cap\bR$ satisfies~\eqref{eq: b good}
if and only if
it is \emph{half-good} for $\Y_{\cdot\wedge\zeta(\Y)}$ and $\wtY$ in the sense of \cite[Definition 2.9]{criensurusov22}
(or, equivalently,
\(b\) is \emph{good} for $\Y_{\cdot\wedge\zeta(\Y)}$ and $\wtY$ in the sense of \cite[Definition 2.11]{criensurusov22}).
The equivalence in the brackets holds because both $\Y_{\cdot\wedge\zeta(\Y)}$ and $\wtY$ are absorbed at their accessible boundaries.

\item
A finite boundary point $b\in\{l,r\}\cap\bR$ satisfies~\eqref{eq: very good add cond}
if and only if
it is inaccessible for $\wtY$ (apply~\eqref{eq:170323a3} for the natural scale and the speed measure $\s'\,d\m$).

\item
Clearly, a boundary point $b\in\{l,r\}$ is inaccessible for $Y$
if and only if
it is inaccessible for $\Y_{\cdot\wedge\zeta(\Y)}$.
\end{itemize}
Now, the statement~(i) immediately follows from \cite[Discussion 2.14 (iii)]{criensurusov22}
and the statement~(iv) follows from \cite[Lemma 2.10]{criensurusov22}.
\end{proof}

\subsection{Proof of Theorem~\protect\ref{th:090224a2}} \label{sec: pf th:090224a2}

For the reader's convenience, we repeat the formulation.

\medskip\noindent
\textbf{Theorem~\ref{th:090224a2}.}
{\em 
Assume that $T<\infty$.
The financial market \((\mathbb{B}, \Y)\) satisfies NUPBR if and only if Condition~\ref{cond:170623a1} holds and, for every $b\in\{l,r\}\cap\bR$, at least one of conditions \textup{(a)--(b)} below is satisfied:
\begin{enumerate}
\item[\textup{(a)}]
condition~\eqref{eq: b good} holds;
\item[\textup{(b)}]
the boundary point $b$ is inaccessible for~$\Y$.
\end{enumerate}}

\begin{proof}[Proof of ``Deterministic conditions $\Rightarrow$ NUPBR'']
We work in the market $(\mathbb B,\Y)$ with a finite time horizon $T<\infty$ and use the notations from the proof of Lemma~\ref{lem:050324a3}.
By Lemma~\ref{lem:050324a3} and Remark~\ref{rem:050324a1}, it remains to show that the deterministic conditions in Theorem~\ref{th:090224a2} imply
\begin{align}\label{eq:110224a3}
\P \Big( \zeta(\Y) \leq T, \int_0^{\zeta(\Y)} \theta^2_s\,d \langle \U,\U\rangle_s = \infty \Big) = 0.
\end{align}
We distinguish three cases.

First, if both boundaries \(l\) and \(r\) are inaccessible for \(Y\), then \(\P(\zeta(\Y) \leq T) = 0\) and \eqref{eq:110224a3} holds.

Second, suppose that one of the boundaries \(l\) or \(r\) is inaccessible, while the other one, call it~\(b\), is a  finite accessible boundary such that \eqref{eq: b good} holds.
In particular, \(|\s (b)| < \infty\) by Feller's test for explosions (recall~\eqref{eq:170323a1}).
Then, \(\P\)-a.s.
\[
\zeta(\Y) = T_{\s(b)} (\U) = \inf \{ t \in [0, T] \colon\U_t = \s (b) \}.
\]
Recalling that \(\U = \s (\Y)\) is a regular diffusion on natural scale, it follows from \cite[Lemma B.18]{criensurusov22} that 
\begin{align*}
\int_{\s (B)}
|s(b)&-x| \big(\beta(\q(x)) \q'(x)\big)^2 dx<\infty
\\&\Longrightarrow\;
\int_0^{T_{\s(b)} (\U)} \theta^2_s \, d \langle\U,\U \rangle_s < \infty
\;\;\P\text{-a.s. on } \{T_{\s (b)}(\U) \leq T\}, 
\end{align*}
where \(B \subsetneq J^\circ\) is an open interval with \(b\) as endpoint.
Notice that
\begin{equation} \label{eq: short computation}
\int_{\s (B)}
|s(b)-x| \big(\beta(\q(x)) \q'(x)\big)^2\,dx
=
\int_B
\frac{| \s(b) - \s(y)|\big(\beta (y)\big)^2}{\s' (y)}\,dy.
\end{equation}
From the equivalences (24) \(\Leftrightarrow\) (25) and (26) \(\Leftrightarrow\) (27) in \cite{MU12} (apply \cite{MU12} with functions \(\mu \equiv 0\), \(\sigma \equiv 1\) and \(b = - \beta/2\)), we deduce that
\begin{equation} \label{eq: def good points}
\int_B
\frac{| \s(b) - \s(y)|\big(\beta (y)\big)^2}{\s' (y)}\,dy<\infty
\quad\Longleftrightarrow\quad
\int_B
| b - y | \big( \beta (y) \big)^2\,dy < \infty. 
\end{equation}
Since the right-hand side coincides with~\eqref{eq: b good}, we conclude that \eqref{eq:110224a3} holds.

Finally, the case where both \(l\) and \(r\) are finite accessible boundaries  such that \eqref{eq: b good} holds for both of them, follows along the same lines.
\end{proof}

\begin{proof}[Proof of ``NUPBR $\Rightarrow$ deterministic conditions'']
Next, we assume NUPBR in the market $(\mathbb B,\Y)$ with a finite time horizon $T<\infty$.
Lemma~\ref{lem:060324a1} implies NUPBR in the market $(\mathbb B^\Y,\Y)$.
Then, NUPBR holds in the canonical setup \((\mathbb{C}_{x_0}, \canY)\) with the same time horizon~$T$ (\cite[Proposition~10.38~(b)]{Jacod}).
Take any $T'\in(0,T)$.
By Lemma~\ref{prop: equivalence initial value}, NUPBR holds in all markets \((\mathbb{C}_{y_0}, \canY)\), $y_0\in J^\circ$, with the time horizon $T'$.
Theorem~\ref{theo: FTAP} yields that there exist SMDs in all these markets.
Therefore, we can apply Lemma~\ref{lem:060324a3}~(ii) and obtain that Condition~\ref{cond:050324a1} is satisfied.
Lemma~\ref{lem:060324a2} upgrades Condition~\ref{cond:050324a1} to Condition~\ref{cond:170623a1}.

It remains to establish the properties of the boundary points.
We return to the (non-canonical) market $(\mathbb B,\Y)$.
Let $\Y=x_0+M+A$ be the canonical decomposition (\cite[Definition I.4.22]{JS}) of the continuous $\mathbf F$-$\P$-semimartingale $\Y$,
i.e.,
$M$ is a continuous $\mathbf F$-$\P$-local martingale and $A$ is a continuous $\mathbf F$-adapted process of finite variation,
$M_0=A_0=0$.
As the market $(\mathbb B,\Y)$ satisfies NUPBR, the structure condition as recalled in Theorem~\ref{theo: SC} holds, i.e.,
there exists an $\mathbf F$-predictable process $\lambda=(\lambda_t)_{t\in[0,T]}$ such that
\begin{equation}\label{eq:060324a1}
A=\int_0^\cdot \lambda_s\,d\langle M,M\rangle_s
\quad\text{and}\quad
\int_0^T \lambda_s^2\,d\langle M,M\rangle_s<\infty
\quad\P\text{-a.s.}
\end{equation}
In the following we use the notations from the proof of Lemma~\ref{lem:050324a3}.
We can do this because we already established Condition~\ref{cond:170623a1}, which was assumed in Lemma~\ref{lem:050324a3}.
We want to make the above canonical decomposition of $\Y$ ($\equiv\q(\U)$) more explicit.
By the generalized It\^o formula (\cite[Theorem~IV.45.1]{RW2} or \cite[Lemma B.23]{criensurusov22}), $\P$-a.s. on the stochastic interval $[0,\zeta(\Y)\wedge T)$, it holds that
\begin{equation}\label{eq:060324a2}
d\q(\U_t)=\q'(\U_t)\,d\U_t+\frac12\q''(\U_t)\,d\langle\U,\U\rangle_t.
\end{equation}
We can state~\eqref{eq:060324a2} only on $[0,\zeta(\Y)\wedge T)$, because the function $\q$ is continuously differentiable with an absolutely continuous derivative on $\s(J^\circ)$ only.
But on the event $\{\zeta(\Y)<T\}$ nothing interesting happens after $\zeta(\Y)$, as $\Y$ is stopped at time $\zeta(\Y)$ (hence so are the processes $M$ and~$A$).
As $\U$ is a continuous $\mathbf F$-$\P$-local martingale, we obtain from \eqref{eq:060324a2} that 
\[
d A_t = \frac{1}{2} \q'' (U_t) \, d \langle U, U \rangle_t, \quad d M_t = \q' (U_t) \, d U_t, \quad d \langle M, M \rangle_t = ( \q' (U_t) )^2 \, d \langle U, U \rangle_t.
\]
Using the structure condition~\eqref{eq:060324a1} yields that 
\begin{align*}
\frac{1}{2} \q'' (U_t) \, d \langle U, U \rangle_t = \lambda_t ( \q' (U_t) )^2 d \langle U, U \rangle_t,
\end{align*}
and integrating \(1 / ( \q' (U) )^2\) against both sides shows that 
\[
\frac{1}{2} \frac{\q'' (U_t)}{(\q' (U_t))^2} \, d \langle U, U \rangle_t = \lambda_t\, d \langle U, U\rangle_t.
\]
Putting these pieces together, we obtain that \(\P\)-a.s.
\begin{align*}
    \infty > \int_0^{\zeta(\Y)\wedge T} \lambda^2_s \, d \langle M, M\rangle_s &= \frac{1}{2}\int_0^{\zeta(\Y)\wedge T} \lambda_s \, \q'' (U_s)  \, d \langle U, U\rangle_s 
    \\&= \frac{1}{4} \int_0^{\zeta(\Y)\wedge T} \Big( \frac{\q'' (U_s)}{\q' (U_s)} \Big)^2 d \langle U, U \rangle_s.
\end{align*}
By the occupation time formula (\cite[Theorem~IV.45.1]{RW2} or \cite[Lemma~B.23]{criensurusov22}) and~\eqref{eq:110224a2}, \(\P\)-a.s.
	\begin{align*}
\int_0^{\zeta(\Y)\wedge T} \Big( \frac{\q'' (U_s)}{\q' (U_s)} \Big)^2 d \langle U, U \rangle_s &= \int \Big( \frac{\q'' (x)}{\q' (x)} \Big)^2  L_{\zeta(\Y)\wedge T}^x (U) \, dx 
	\\&= \int \Big( \beta (\q (x)) \q' (x) \Big)^2 L^x_{\zeta(\Y)\wedge T} (\U)\, dx 
	\\&= \int_0^{\zeta(\Y)\wedge T} \Big(
	\beta\big(\q(\U_s)\big)\q'(\U_s)
	\Big)^2
	\,d\langle\U,\U\rangle_s.
	\end{align*} 
Hence, we also have
\begin{equation}\label{eq:060324a3}
\int_0^{\zeta(\Y)\wedge T}
\Big(
\beta\big(\q(\U_s)\big)\q'(\U_s)
\Big)^2
\,d\langle\U,\U\rangle_s<\infty\quad\P\text{-a.s.}
\end{equation}

We need to prove that, if \(b \in \{l, r\}\cap \bR\) is accessible, then it satisfies~\eqref{eq: b good}.
To this end, take an accessible \(b \in \{l, r\}\cap \bR\) and assume for contradiction that \eqref{eq: b good} does not hold for the boundary point $b$.
It follows from \eqref{eq: short computation} and~\eqref{eq: def good points} that
$$
\int_{\s (B)}
|s(b)-x| \big(\beta(\q(x)) \q'(x)\big)^2\,dx=\infty
$$
for any open interval $B\subsetneq J^\circ$ with $b$ as endpoint.
\cite[Lemma B.18]{criensurusov22} implies
\begin{equation}\label{eq:060324a4}
\int_0^{T_{\s(b)}(\U)}
\Big(
\beta\big(\q(\U_s)\big)\q'(\U_s)
\Big)^2
\,d\langle\U,\U\rangle_s=\infty\quad\P\text{-a.s. on }\{T_{\s(b)}(\U)\le T\}.
\end{equation}
As $T_{\s(b)}(\U)=T_b(\Y)$ and $b$ is accessible for $\Y$, we get $\P(T_{\s(b)}(\U)\le T)>0$
from Lemma~\ref{lem:BR}.
Together with the fact that on $\{T_{\s(b)}(\U)\le T\}$ we have $\zeta(\Y)=T_{\s(b)}(\U)$, we get that \eqref{eq:060324a3} and~\eqref{eq:060324a4} contradict each other.
This concludes the proof.
\end{proof}

\subsection{Proof of Theorem~\protect\ref{th:100224a2}} \label{sec: pf th:100224a2}

For the reader's convenience, we repeat the formulation.

\medskip\noindent
\textbf{Theorem~\ref{th:100224a2}.}
{\em 
Suppose that $T=\infty$ and that $Y$ is \emph{not} on natural scale.
The financial market \((\mathbb{B}, \Y)\) satisfies NUPBR if and only if Condition~\ref{cond:170623a1} holds and, for every $b\in\{l,r\}$,
one of the following conditions \textup{(A)--(B)} is satisfied:
\begin{enumerate}
\item[\textup{(A)}]
$|b|<\infty$ and \eqref{eq: b good} holds;
\item[\textup{(B)}]
$|\s(b)|=\infty$ and the other boundary point $b^*$ satisfies~\textup{(A)}.
\end{enumerate}}

\begin{proof}[Proof of ``Deterministic conditions $\Rightarrow$ NUPBR'']
We work in the market $(\mathbb B,\Y)$ with the infinite time horizon $T=\infty$ and use the notations from the proof of Lemma~\ref{lem:050324a3}.
First, assume that one of the boundaries, call it $b$, satisfies~(B).
Then, the other boundary point $b^*$ satisfies~(A), i.e., $|b^*|<\infty$ and \eqref{eq: b good} holds for $b^*$.
Lemma~\ref{lem:170623a1} yields $|\s(b^*)|<\infty$.
As $|\s(b)|=\infty$ and $|\s(b^*)|<\infty$, we obtain, by \cite[Proposition~5.5.22]{KaraShre} (or \cite[Lemma B.7]{criensurusov22}), that $\P(\lim_{t\nearrow\zeta(\Y)}\Y_t=b^*)=1$ (regardless of whether $\zeta(\Y)$ is finite or infinite).
Now, the same calculations as in the proof of Theorem~\ref{th:090224a2},
which involve \cite[Lemma B.18]{criensurusov22},
\eqref{eq: short computation}
and~\eqref{eq: def good points},
yield that
\begin{equation}\label{eq:210324a1}
\int_0^{\zeta(\Y)}
\theta_s^2\,d\langle\U,\U\rangle_s
<\infty
\quad\P\text{-a.s.}
\end{equation}
Similarly, we obtain~\eqref{eq:210324a1} in the remaining case where both boundary points satisfy~(A).
NUPBR for $T=\infty$ now follows from \eqref{eq:210324a1},
Lemma~\ref{lem:050324a3},
Remark~\ref{rem:050324a1}
and
Theorem~\ref{theo: FTAP}.
\end{proof}

\begin{proof}[Proof of ``NUPBR $\Rightarrow$ deterministic conditions'']
As NUPBR for the infinite time horizon $T=\infty$ implies NUPBR for every finite time horizon, we obtain Condition~\ref{cond:170623a1} by Theorem~\ref{th:090224a2}.
It remains to establish (A) or~(B) for every $b\in\{l,r\}$.
As in the proof of Theorem~\ref{theo: FTAP}, it follows from
\cite[Theorem 4.12]{karatzaskardaras07} and \cite[Theorem~7]{HS10}
that the \emph{minimal martingale density}
$\widehat Z$ (see p.~42 in \cite{HS10}) is an SMD such that $\P$-a.s. $\widehat Z_\infty>0$.
In particular, 
$$
\widehat{Z} = \mathcal{E} \Big( - \int_0^{\cdot}
\frac{\q'' (\U_s)}{2 \q' (\U_s)} \, dU_s \Big)
\quad\text{on }[0,\zeta(\Y)), 
$$
see the proof of Theorem~\ref{th:090224a2} for details.
As nonnegative supermartingales cannot resurrect from zero (\cite[Lemma~III.3.6]{JS}),
$\widehat Z_\infty>0$ $\P$-a.s. implies $\widehat Z_{\zeta(\Y)}>0$ $\P$-a.s.,
hence
$$
\int_0^{\zeta(\Y)}\Big(\frac{\q''(\U_s)}{\q'(\U_s)}\Big)^2\,d\langle\U,\U\rangle_s<\infty\quad\P\text{-a.s.}
$$
or, equivalently,
\begin{equation}\label{eq:210324a2}
\int_0^{\zeta(\Y)}
\Big(
\beta\big(\q(\U_s)\big)\q'(\U_s)
\Big)^2
\,d\langle\U,\U\rangle_s<\infty\quad\P\text{-a.s.}
\end{equation}
(see~\eqref{eq:060324a3}).
Notice that our assumption that $\Y$ is \emph{not} on natural scale entails that \[\mu_L(\beta (\q) \q' \ne 0) > 0.\]
If $\U$ were recurrent, then \cite[Lemma B.19]{criensurusov22} would imply \(\P\)-a.s. $\zeta(\Y)=\infty$ and
\[\int_0^\infty
(\beta(\q(\U_s))\q'(\U_s))^2\,d\langle\U,\U\rangle_s=\infty,
\]
which contradicts~\eqref{eq:210324a2}.
Therefore, $\U$ is not recurrent.
Then, by \cite[Lemmata B.7, B.9]{criensurusov22},
$$
A_l\sqcup A_r
\triangleq
\big\{
\lim\nolimits_{t\nearrow\zeta(\Y)}\U_t=\s(l)
\big\}
\sqcup
\big\{
\lim\nolimits_{t\nearrow\zeta(\Y)}\U_t=\s(r)
\big\}
=
\Omega
\quad\P\text{-a.s.}
$$
and, for $b\in\{l,r\}$,
$\P(A_b)=0$ whenever $|\s(b)|=\infty$,
while $\P(A_b)>0$ whenever $|\s(b)|<\infty$.
It follows that, if there is a boundary point $b$ with $|\s(b)|=\infty$,
then the other boundary point $b^*$ satisfies $|\s(b^*)|<\infty$.

It remains to prove that, if $b\in\{l,r\}$ satisfies $|\s(b)|<\infty$, then $|b|<\infty$ and \eqref{eq: b good} holds for $b$.
To this end, let $b\in\{l,r\}$ be such that $|\s(b)|<\infty$.
By \cite[Lemma B.18]{criensurusov22}, \eqref{eq:210324a2} implies that
$$
\int_{\s (B)}
|s(b)-x| \big(\beta(\q(x)) \q'(x)\big)^2\,dx<\infty,
$$
where $B\subsetneq J^\circ$ is an open interval with $b$ as endpoint.
It follows from~\eqref{eq: short computation} that
\begin{equation}\label{eq:210324a3}
\int_B
\frac{| \s(b) - \s(y)|\big(\beta (y)\big)^2}{\s' (y)}\,dy<\infty.
\end{equation}
At this place it is tempting to apply~\eqref{eq: def good points},
but, contrary to the situation in the proof of Theorem~\ref{th:090224a2}, we do not yet know that $|b|<\infty$.
So, we again need to consult the equivalencies
(24) \(\Leftrightarrow\) (25) and (26) \(\Leftrightarrow\) (27) in \cite{MU12} (apply \cite{MU12} with functions \(\mu \equiv 0\), \(\sigma \equiv 1\) and \(b = - \beta/2\)) and, in our present situation, we deduce that
both $|\s(b)|<\infty$ and~\eqref{eq:210324a3}
are equivalent to
both $|b|<\infty$ and~\eqref{eq: b good}.
This concludes the proof.
\end{proof}

\subsection{Further lemmata} \label{sec: more lemma}

In order to prove the remaining Theorems \ref{th:090224a1} and~\ref{th:100224a1}, we prepare the following results.

As parts of our work are done in the canonical setting,
we sometimes arrive at (local) martingales on the canonical space,
which are naturally translated to $\mathbf F^\Y$-(local) martingales on the space $\mathbb B$.
The next lemma provides an instrument how to lift the $\mathbf F^\Y$-(local) martingale property to the $\mathbf F$-(local) martingale property.
This is a standalone result in the sense that it does not rely on anything else from Section~\ref{sec:proofs}.

For this result, we need only the Markov property w.r.t. the bigger filtration (not necessarily the strong one).
In order to emphasize this difference from our setting ($\Y$ is strong Markov w.r.t.~$\mathbf F$),
we use a different notation for the filtration and for the main process in Lemma~\ref{lem: markov martingale}.

\begin{lemma}[Lifting the martingale property to a bigger filtration]\label{lem: markov martingale}
Let \(M = (M_t)_{t \in\bR_+}\) be a real-valued continuous (not necessarily strong) Markov process w.r.t. a right-continuous filtration
\(\mathbf{G} = (\mathcal{G}_t)_{t \in\bR_+}\)
on some probability space $(\Omega,\cF,\P)$.
By $\mathbf F^M=(\cF^M_t)_{t\in\bR_+}$ we denote the right-continuous filtration generated by \(M\).
Further, let \(\canU = (\canU_t)_{t \in \bR_+}\) be a c\`adl\`ag \(\canbfF\)-adapted process on the canonical space \((\canOmega, \canF)\) (defined with \(J = \bR\)).
\begin{enumerate}
\item[\textup{(i)}]
If $\canU(M)$ is an $\mathbf F^M$-martingale, then it is a $\mathbf G$-martingale.
\item[\textup{(ii)}] If $\canU(M)$ is an $\mathbf F^M$-local martingale, then \(\canU(M)\) is a $\mathbf G$-local martingale.
\end{enumerate}
\end{lemma}

\begin{proof}
(i)
First, recall the classical fact that, as \(M\) is a Markov process w.r.t. \(\mathbf{G}\), it is also a Markov process w.r.t. \(\mathbf{F}^M\).
	Take times \(s_1 < \dots < s_n \leq s \leq t_1 < \dots < t_m \leq t\) in $\bR_+$ and bounded Borel functions \(f \colon \bR^n \to \bR\) and \(g \colon \bR^m \to \bR\). Then, by the Markov property of \(M\), we get a.s.
	\begin{align*}
		\E\big[ f (M_{s_1}, \dots, M_{s_n}) g (M_{t_1}, \dots, M_{t_m}) | \mathcal{G}_s \big] 
		&= f (M_{s_1}, \dots, M_{s_n}) \E\big[ g (M_{t_1}, \dots, M_{t_m}) | \mathcal{G}_s \big] 
		\\&= f (M_{s_1}, \dots, M_{s_n}) \E\big[ g (M_{t_1}, \dots, M_{t_m}) | M_s \big]
		\\&= f (M_{s_1}, \dots, M_{s_n}) \E\big[ g (M_{t_1}, \dots, M_{t_m}) | \mathcal{F}^M_s \big]
		\\&=  \E\big[ f (M_{s_1}, \dots, M_{s_n}) g (M_{t_1}, \dots, M_{t_m}) | \mathcal{F}^M_s \big].
	\end{align*}
By a monotone class argument, it follows that a.s. 
\[
\E\big[ \canV (M) | \mathcal{G}_s \big] = \E\big[ \canV (M) | \mathcal{F}^M_s \big]
\]
for all \(\canF\)-measurable random variables $\canV$ such that \(\canV (M)\) is integrable. 
This equality immediately implies the claim in~(i).

\smallskip
(ii)
Let \((\rho_n)_{n = 1}^\infty\) be an \(\mathbf{F}^M\)-localizing sequence for the \(\mathbf{F}^M\)-local martingale property of \(\canU(M)\). By virtue of \cite[Proposition~10.35~(a)]{Jacod}, there exists a sequence \((\tau_n)_{n = 1}^\infty\) of \(\canbfF\)-stopping times on the canonical space such that \(\rho_n = \tau_n (M)\) for all \(n \in \mathbb{N}\). By definition of a localizing sequence, for all \(n \in \mathbb{N}\), the process 
$\canU_{\cdot\wedge\tau_n(M)}(M)-\canU_0(M)$ is an $\mathbf F^M$-martingale.
Applying part (i) to the c\`adl\`ag $\canbfF$-adapted processes
$\canU_{\cdot\wedge\tau_n}-\canU_0$
on the canonical space, we obtain that $\canU_{\cdot\wedge\tau_n(M)}(M)-\canU_0(M)$ is a $\mathbf G$-martingale,
which means that $\canU(M)$ is a $\mathbf G$-local martingale.
\end{proof}

In the next lemma (Lemma~\ref{lem:130224a1}) we need the following conditions.

\smallskip\noindent
\textbf{Condition~A.}
For every finite boundary point $b\in\{l,r\}\cap\bR$, at least one of conditions \eqref{eq: b good}--\eqref{eq: very good add cond} is satisfied.

\smallskip\noindent
\textbf{Condition~B.}
For every boundary point $b\in\{l,r\}$,
one of the following conditions \textup{(a)--(b)} is satisfied:
\begin{enumerate}
\item[\textup{(a)}]
$|b|<\infty$ and \eqref{eq: b good} holds;
\item[\textup{(b)}]
$|b|=\infty$ and the other boundary point $b^*$ satisfies~\textup{(a)}.
\end{enumerate}

\smallskip\noindent
Notice that the Conditions A and~B make sense under Condition~\ref{cond:170623a1}.
The latter is explicitly assumed in the next lemma, which provides an important step towards constructing a (density of an) ACLMM.

\begin{lemma}[Step towards ACLMM]\label{lem:130224a1}
Assume Condition~\ref{cond:170623a1} and consider the market $(\mathbb B,\Y)$ with the infinite time horizon $T=\infty$.
Further assume the previously introduced Condition~A (resp., Condition~B).
Then, there exists a nonnegative c\`adl\`ag
$$
\mathbf F\text{-}\P\text
{-martingale (resp., uniformly integrable }
\mathbf F\text{-}\P\text{-martingale)}
$$
$Z=(Z_t)_{t\in\bR_+}$ with $Z_0=1$, strictly positive on the stochastic interval $[0,\zeta(\Y))$, such that $ZY$ is an $\mathbf F$-$\P$-local martingale.
\end{lemma}

\begin{remark}\label{rem:180324a1}
Below we provide two proofs of Lemma~\ref{lem:130224a1}.
The second one is more involved but in return it establishes more:
namely, that the process $\Z$ constructed in the proof of Lemma~\ref{lem:050324a3} (see~\eqref{eq: main Z}) has the properties claimed in Lemma~\ref{lem:130224a1}.
The process $\Z$ as defined in~\eqref{eq: main Z} is closely related to the \emph{minimal martingale density} that was already used in the proof of Theorem~\ref{theo: FTAP}.\footnote{Notice, however, the difference: the process $\Z$ given in~\eqref{eq: main Z} is not necessarily strictly positive.
The reason is that it is constructed under Condition~\ref{cond:170623a1}, which is too weak to guarantee strict positivity.}
Thus, the message is that, for our general diffusion markets,
under the assumptions of Lemma~\ref{lem:130224a1},
the ``generalized minimal martingale density'' $\Z$ of~\eqref{eq: main Z} is always a martingale (resp., a uniformly integrable martingale).
It is instructive to compare this comment with the counterexample in \cite{DS1998counter} or \cite[Chapter 10]{DS2006}.
\end{remark}

\begin{proof}[First proof of Lemma~\ref{lem:130224a1}]
For a while we work in the canonical setting.
In addition to the measure $\P_{x_0}$ on $(\canOmega,\canF)$,
we consider the law $\tP_{x_0}$ of the diffusion on natural scale
started in $x_0$ with the interior of the state space $J^\circ$, speed measure $\s'\,d\m$ on $J^\circ$
and boundary points $l$ and $r$ being absorbing whenever they are accessible.
Let $S$ be the
\emph{separating time for $\P_{x_0}$ and $\tP_{x_0}$}
(see \cite[Section~2]{chernyurusov04},
\cite[Section~2]{chernyurusov06} or
\cite[Section 2.1]{criensurusov22}).
In our present context, Condition~\ref{cond:050324a1} means that all points from $J^\circ$ are \emph{good} in the sense of \cite[Definition 2.5]{criensurusov22}.
By \cite[Theorem 2.13]{criensurusov22}, it follows that
$$
S\ge\zeta\;\;\P_{x_0},\tP_{x_0}\text{-a.s.},
$$
which implies, for every $t\in\bR_+$,
\begin{equation}\label{eq:190324a1}
\tP_{x_0}\sim\P_{x_0}\text{ on }\canF_t\cap\{t<\zeta\}.
\end{equation}
Moreover, by \cite[Corollary 2.15]{criensurusov22}
(resp., \cite[Corollary 2.16]{criensurusov22}),
we obtain that, under Condition~A (resp., Condition~B),
\begin{equation}\label{eq:190324a2}
\tP_{x_0}\ll_{\mathrm{loc}}\P_{x_0}
\quad
\text{(resp., }
\tP_{x_0}\ll\P_{x_0}
\text{)}.
\end{equation}
Let $\wtZ=(\wtZ_t)_{t\in\bR_+}$ be the density process of $\tP_{x_0}$ w.r.t. $\P_{x_0}$ (note that we are assuming at least Condition~A in this lemma).
It follows from \eqref{eq:190324a1} and~\eqref{eq:190324a2} that $\wtZ$ is a nonnegative c\`adl\`ag $\canbfF$-$\P_{x_0}$-martingale
(uniformly integrable under Condition~B),
strictly positive on $[0,\zeta)$.
Furthermore, Blumenthal's zero-one law (\cite[Lemma~4, p.~106]{freedman} or \cite[Lemma B.1]{criensurusov22}) yields that $\P_{x_0}$-a.s. $\wtZ_0=1$.

Next, as $\canY$ is continuous, the sequence
$$
\sigma_n\triangleq\inf\{t\geq 0 \colon |\canY_t|\ge n\},
\quad
n\in\mathbb N,\;
n>|x_0|,
$$
localizes its \(\canbfF\)-\(\tilde{\P}_{x_0}\)-martingale property. 
As identically \(\lim_{n \to \infty}\sigma_n = \infty\), part (c) of \cite[Proposition III.3.8]{JS} yields that $\wtZ\,\canY$
is an $\canbfF$-$\P_{x_0}$-local martingale.
By \cite[Theorem~10.37]{Jacod},
$\Z \triangleq\wtZ(\Y)$ is an $\mathbf F^\Y$-$\P$-martingale
(uniformly integrable under Condition~B)
and $ZY$ is an $\mathbf F^\Y$-$\P$-local martingale.
Finally, applying Lemma~\ref{lem: markov martingale} twice, first with
$M\triangleq\Y$ and
$\canU\triangleq\wtZ$,
then with $M\triangleq\Y$ and
$\canU\triangleq\wtZ\,\canY$,
completes the proof.
\end{proof}

\begin{proof}[Second proof of Lemma~\ref{lem:130224a1}]
In the following, we will prove more than what is claimed in Lemma~\ref{lem:130224a1}.
Namely, under the assumptions from Lemma~\ref{lem:130224a1}, we show that the process $\Z$ as constructed in the proof of Lemma~\ref{lem:050324a3} (see~\eqref{eq: main Z}) has the properties described in Lemma~\ref{lem:130224a1}, see
Remark~\ref{rem:180324a1} for an interpretation of this fact in the general context of mathematical finance.

Thus, we assume Condition~\ref{cond:170623a1} and use the notation from the proof of Lemma~\ref{lem:050324a3}.
In particular, we recall that the process $\Z=(\Z_t)_{t\in\bR_+}$ constructed there is a nonnegative $\mathbf F$-$\P$-local martingale with $\Z_0=1$,
$\Z$ is strictly positive on $[0,\zeta(\Y))$
and $\Z\Y$ is an $\mathbf F$-$\P$-local martingale.
As a nonnegative $\mathbf F$-$\P$-local martingale, $\Z$ is an $\mathbf F$-$\P$-supermartingale and $\P$-a.s. has a limit $\Z_\infty\triangleq\lim_{t\to\infty}\Z_t$.
Therefore, it suffices to prove that, under Condition~A,
$\E^\P[\Z_t]=1$ for all $t\in\bR_+$,
while, under Condition~B,
$\E^\P[\Z_\infty]=1$.

To this end, we turn to the canonical setting and consider the market $(\mathbb C_{x_0},\canY)$ with the infinite time horizon $T=\infty$.
We define the processes $\canU$, $\cantheta$ and $\canZ$ through $\canY$ in the same way, as $\U$, $\theta$ and $\Z$ are defined through $\Y$ in the proof of Lemma~\ref{lem:050324a3}.
Notice that the integrals in~\eqref{eq: main Z} could also be understood w.r.t. the filtration $\mathbf F^\Y$ (on $\mathbb B$) and would define the same process $\Z$.
Therefore, by \cite[Proposition 10.38~(b)]{Jacod}, the law of the process $\Z$ under $\P$ coincides with the law of $\canZ$ under $\P_{x_0}$.
In summary, it suffices to prove that, under Condition~A,
$\E^{\P_{x_0}}[\canZ_t]=1$ for all $t\in\bR_+$,
while, under Condition~B,
$\E^{\P_{x_0}}[\canZ_\infty]=1$.

For every $n\in\mathbb N$, we set
$$
\tau_n'\triangleq\inf
\Big\{
t\in[0,\zeta) \colon
\int_0^t
\cantheta_s^2\,d\langle\canU,\canU\rangle_s\ge n
\Big\},
$$
which is an $\canbfF^{\P_{x_0}}$-stopping time,
where $\canbfF^{\P_{x_0}}$ is the filtration $\canbfF$ completed with all $\P_{x_0}$-null sets.
Applying the occupation time formula as in~\eqref{eq: identity occ time formula} yields that
\begin{equation}\label{eq:230324a-1}
\int_0^t
\cantheta_s^2\,d\langle\canU,\canU\rangle_s<\infty
\quad\P_{x_0}\text{-a.s. for }t\in[0,\zeta)
\end{equation}
(alternatively, we can infer this from~\eqref{eq: identity occ time formula} via \cite[Proposition 10.38~(b)]{Jacod}).
In particular, $\P_{x_0}$-a.s. $\tau_n'>0$.
By Meyer's theorem on predictability (see \cite[Proposition~4]{chungwalsh} and \cite[Lemma~I.2.17]{JS}, or \cite[Lemma~B.20]{criensurusov22}), we conclude that, for every $n\in\mathbb N$, there exists an $\canbfF$-predictable time $\tau_n$ with $\tau_n(\omega)>0$ for all $\omega\in\canOmega$ such that
\begin{equation}\label{eq:230324a0}
\tau_n=\inf
\Big\{
t\in[0,\zeta) \colon
\int_0^t
\cantheta_s^2\,d\langle\canU,\canU\rangle_s\ge n
\Big\}
\quad\P_{x_0}\text{-a.s.}
\end{equation}
Thanks to Novikov's condition (\cite[Proposition~VIII.1.15]{RY}),
\begin{equation}\label{eq:230324a1}
\canZ_{\cdot\wedge\tau_n}
\text{ is a nonnegative uniformly integrable }\canbfF\text{-}\P_{x_0}\text{-martingale with }
\canZ_0=1.
\end{equation}
This enables us to define the probability measure $\Q_n$ on $(\canOmega,\canF)$ via the Radon-Nikodym density $d\Q_n/d\P_{x_0}\triangleq\canZ_{\tau_n}$.
As $\Z\Y$ is an $\mathbf F^\Y$-adapted continuous $\mathbf F$-$\P$-local martingale (on $\mathbb B$), it is also an $\mathbf F^\Y$-$\P$-local martingale and, by \cite[Theorem 10.37]{Jacod},
$\canZ\,\canY$ is an $\canbfF$-$\P_{x_0}$-local martingale.
Hence, also
$\canZ_{\cdot\wedge\tau_n}\canY_{\cdot\wedge\tau_n}$
is an $\canbfF$-$\P_{x_0}$-local martingale.
By \cite[Proposition III.3.8]{JS},
$\canY_{\cdot\wedge\tau_n}$
is an $\canbfF$-$\Q_n$-local martingale.
Now, Lemma~\ref{lem:050324a1} yields $\Q_n=\tP_{x_0}$ on $\canF^o_{\tau_n}$,
where $\tP_{x_0}$ is the law of the diffusion on natural scale
started in $x_0$ with the interior of the state space $J^\circ$,
speed measure $\s'\,d\m$ on $J^\circ$
and boundaries $l$ and $r$ being absorbing whenever they are accessible.
It is worth noting that $\tP_{x_0}$ is as in the first proof of Lemma~\ref{lem:130224a1}.

Our next aim is to show that
\begin{equation}\label{eq:230324a4}
\text{Condition A}
\;\;\Longrightarrow\;\;
\tau_n\nearrow\infty\quad\tP_{x_0}\text{-a.s.},
\end{equation}
while
\begin{equation}\label{eq:230324a5}
\text{Condition B}
\;\;\Longrightarrow\;\;
\{\tau_n=\infty\}\nearrow\canOmega\quad\tP_{x_0}\text{-a.s.}
\end{equation}
To this end, we need to understand 
$\int_0^\cdot \cantheta_s^2\,d\langle\canU,\canU\rangle_s$
on $[0,\zeta)$ under $\tP_{x_0}$.
First, we observe that this process is $\tP_{x_0}$-a.s. finite on $[0,\zeta)$
(cf.~\eqref{eq:230324a-1}, where we have the other measure $\P_{x_0}$, but notice that the argument via the occupation time formula needs \eqref{eq:110224a1} only).
As, under $\tP_{x_0}$, $\canY$ is on natural scale, it is worth expressing everything through $\canY$.
Under $\tP_{x_0}$, we have
$$
\cantheta_t=\frac12 \beta\big(\q(\canU_t)\big)\q'(\canU_t)
=\frac12
\frac{\beta (\canY_t)}{\s' (\canY_t)},
\quad t\in[0,\zeta),
$$
and
$$
d\langle\canU,\canU\rangle_t
=
d\langle\s(\canY),\s(\canY)\rangle_t
=
\big(\s'(\canY_t)\big)^2\,d\langle\canY,\canY\rangle_t
\quad\text{on }[0,\zeta).
$$
Thus, we have
\begin{equation}\label{eq:230324a6}
\int_0^t
\cantheta_s^2\,d\langle\canU,\canU\rangle_s
=\frac14
\int_0^t
\big(\beta(\canY_s)\big)^2\,d\langle\canY,\canY\rangle_s,
\quad t\in[0,\zeta).
\end{equation}
Further, by \cite[Lemma B.18]{criensurusov22}, if a boundary point $b\in\{l,r\}$ satisfies $|b|<\infty$ and~\eqref{eq: b good}, then
\begin{equation}\label{eq:230324a7}
\int_0^\zeta
\big(\beta(\canY_s)\big)^2\,d\langle\canY,\canY\rangle_s
<\infty
\quad\tP_{x_0}\text{-a.s. on }
\Big\{\lim_{t\nearrow\zeta}\canY_t=b \Big\}.
\end{equation}
Next, we observe that a boundary point $b\in\{l,r\}\cap\bR$ satisfies~\eqref{eq: very good add cond} if and only if $b$ is inaccessible for $\canY$ under $\tP_{x_0}$
(apply~\eqref{eq:170323a3} for the natural scale and the speed measure $\s'\,d\m$).
This and~\eqref{eq:230324a7} yield that, under Condition~A,
$$
\int_0^{t\wedge\zeta}
\big(\beta(\canY_s)\big)^2\,d\langle\canY,\canY\rangle_s
<\infty
\quad\tP_{x_0}\text{-a.s.,}
\quad\forall \, t\in\bR_+.
$$
Recalling \eqref{eq:230324a0} and~\eqref{eq:230324a6}, we get~\eqref{eq:230324a4}.
Furthermore, under Condition~B, we have that $\tP_{x_0}$-a.s.,
as $t\nearrow\zeta$,
$\canY_t$ tends to (the) boundary point(s) satisfying $|b|<\infty$ and~\eqref{eq: b good}, see \cite[Proposition~5.5.22]{KaraShre} or \cite[Lemma~B.7]{criensurusov22}.
Together with~\eqref{eq:230324a7} this implies
$$
\int_0^\zeta
\big(\beta(\canY_s)\big)^2\,d\langle\canY,\canY\rangle_s
<\infty
\quad\tP_{x_0}\text{-a.s.}
$$
under Condition~B, which yields~\eqref{eq:230324a5}.

Finally, under Condition~A (resp., Condition~B), we take any $t\in\bR_+$ (resp., $t=\infty$) and observe that
$$
1
\ge
\E^{\P_{x_0}}[\canZ_t]
\ge
\E^{\P_{x_0}}[\canZ_t \1_{\{\tau_n\ge t\}}]
=
\E^{\P_{x_0}}[\canZ_{\tau_n} \1_{\{\tau_n\ge t\}}]
=
\Q_n(\tau_n\ge t)
=
\tP_{x_0}(\tau_n\ge t)\to1
$$
as $n\to\infty$,
where we use \eqref{eq:230324a1} in the first equality,
then the definition of $\Q_n$,
then the facts that $\Q_n=\tP_{x_0}$ on $\canF^o_{\tau_n}$ and $\{\tau_n\le t\}\in\canF^o_{\tau_n}$
and, ultimately,
\eqref{eq:230324a4} (resp., \eqref{eq:230324a5}) in the statement about the limit.
All in all, we obtain $\E^{\P_{x_0}}[\canZ_t]=1$ for all $t\in\bR_+$ under Condition~A and even for $t=\infty$ under Condition~B.
The proof is complete.
\end{proof}

\begin{discussion}
It is worth noting that, under Condition~\ref{cond:170623a1} and Condition~A, the density process $\wtZ$ of $\tP_{x_0}$ w.r.t. $\P_{x_0}$ from the first proof of Lemma~\ref{lem:130224a1} coincides with the process $\canZ$ from the second one (defined as the canonical analog of $Z$ from~\eqref{eq: main Z}).
Indeed, as $\canZ$ is a nonnegative $\canbfF$-$\P_{x_0}$-martingale with $\canZ_0=1$,
a standard extension theorem (e.g., \cite[Lemma 19.19]{kallenberg}) yields the existence of a probability measure \(\tQ\) on \(\canF\) such that $\tQ\ll_{\textup{loc}}\P_{x_0}$ and \(d \tQ = \canZ_t d \P_{x_0}\) on \(\canF_t\) for all \(t \in \bR_+\). As the process $\canY$ is an $\canbfF$-$\tQ$-local martingale by \cite[Proposition III.3.8]{JS}, using Lemma~\ref{lem:050324a1} with $\xi=t$ for every $t\in\bR_+$, we get that $\tQ=\tP_{x_0}$ and hence \(\P_{x_0}\)-a.s. $\wtZ=\canZ$.
\end{discussion}

\subsection{Proof of Theorems \protect\ref{th:090224a1} and~\protect\ref{th:100224a1}} \label{sec: pf th:090224a1, th:100224a1}

For the reader's convenience, we repeat the formulations.

\medskip\noindent
\textbf{Theorem~\ref{th:090224a1}.}
{\em 
Assume that $T<\infty$.
The financial market \((\mathbb{B}, \Y)\) satisfies NA if and only if Condition~\ref{cond:170623a1} holds and, for every $b\in\{l,r\}\cap\bR$, at least one of conditions \eqref{eq: b good}--\eqref{eq: very good add cond} is satisfied.}

\medskip\noindent
\textbf{Theorem~\ref{th:100224a1}.}
{\em 
Suppose that $T=\infty$ and that $Y$ is \emph{not} on natural scale.
The financial market \((\mathbb{B}, \Y)\) satisfies NA if and only if Condition~\ref{cond:170623a1} holds and, for every $b\in\{l,r\}$,
one of the following conditions \textup{(a)--(b)} is satisfied:
\begin{enumerate}
\item[\textup{(a)}]
$|b|<\infty$ and \eqref{eq: b good} holds;
\item[\textup{(b)}]
$|b|=\infty$ and the other boundary point $b^*$ satisfies~\textup{(a)}.
\end{enumerate}}

\medskip
To prove the sufficiency of the deterministic conditions for NA, we discuss two different methods.
The first proof is shorter, but it works only in the case where $T<\infty$, because it is based on Theorem~\ref{th:210324a1}.
The second one is a direct verification of NA, which works both for $T<\infty$ and $T=\infty$.

\begin{proof}[First proof of the implication ``Deterministic conditions $\Rightarrow$ NA'' (Theorem~\ref{th:090224a1} only)]
Here, we consider only the case of a finite time horizon $T<\infty$.
Our tactic is to use the characterization of NA from Theorem~\ref{th:210324a1}. Take an \(\mathbf{F}\)-stopping time \(\sigma \leq T\). Further, set \(\zeta_0 \triangleq \zeta (\Y) \wedge T\). Notice that the deterministic conditions in Theorem~\ref{th:090224a1}
are precisely
Condition~\ref{cond:170623a1} together with Condition~A
from Lemma~\ref{lem:130224a1}.
Hence, by Lemma~\ref{lem:130224a1}, there exists a nonnegative c\`adl\`ag $\mathbf F$-$\P$-martingale $\Z=(\Z_t)_{t\in[0,T]}$ with $\Z_0=1$, strictly positive on the stochastic interval $[0,\zeta_0)$, such that $\Z\Y$ is an $\mathbf F$-$\P$-local martingale on $[0,T]$. Using the continuous time-change \(s \mapsto \tau_s \triangleq (s + \sigma) \wedge T\), it follows from \cite[Theorem~10.16]{Jacod}
that \((\Z_{\tau_s} \Y_{\tau_s})_{s \in [0, T]}\) is a \({^\sigma}\mathbf{F}\)-\(\P\)-local martingale for the time-changed filtration \({^\sigma}\mathbf{F} = (\cF_{\tau_s})_{s \in [0, T]}\).
Furthermore, as \(\tau_s \leq T\) for all \(s \in [0, T]\), the process \((\Z_{\tau_s})_{s \in [0, T]}\) is a true \({^\sigma}\mathbf{F}\)-\(\P\)-martingale by the optional sampling theorem.
Thanks to Lemma~\ref{lem:BR}, we have 
\[
\P ( \sigma < \zeta (Y)) \geq \P (T < \zeta(\Y)) > 0.
\]
Hence, also
\[
\E^\P\big[ Z_\sigma \1_{\{\sigma < \zeta (Y)\}} \big] > 0
\]
because \(\P\)-a.s. \(Z_\sigma> 0\) on \(\{\sigma < \zeta (\Y)\}\).
We define a probability measure \(\Q\) by
\begin{align*}
\Q(G) \equiv {^\sigma}\Q(G) \triangleq \frac{\P (\sigma < \zeta (\Y))}{\E^\P \big[ Z_\sigma \1_{\{\sigma < \zeta (\Y)\}}\big]} \E^{\P} \big[ Z_T \1_{G \, \cap \, \{\sigma < \zeta (\Y)\}} \big] + \P ( G \cap \{\sigma \geq \zeta(\Y) \}), \quad G \in \cF.
\end{align*}
By virtue of \cite[Proposition~III.3.8]{JS}, and recalling that the boundaries of \(\Y\) are absorbing or inaccessible under \(\P\), we get that \((\Y_{\tau_s})_{s \in [0, T]}\) is an \({^\sigma}\mathbf{F}\)-\(\Q\)-local martingale.
By definition of~\(\Q\), it is evident that \(\Q\ll \P\).
Finally, if \(G \in \cF_\sigma\) is such that \(\Q (G) = 0\), we immediately get \(\P (G \cap \{\sigma \geq \zeta (Y)\}) = 0\), and \(\P(G \cap \{\sigma < \zeta (\Y)\}) = 0\) follows from 
\[
\E^\P \big[ \Z_{\sigma} \1_{G \, \cap \, \{\sigma < \zeta (\Y)\}} \big] = 0
\]
because \(\P\)-a.s. \(Z_\sigma > 0\) on \(\{\sigma < \zeta (\Y)\}\).
Hence, \(\P\ll \Q\) on \(\cF_{\sigma}\), which entails that \(\P \sim \Q\) on \(\cF_\sigma\).
Summing up, Theorem~\ref{th:210324a1} yields that NA holds. The proof is complete.
\end{proof}

\begin{proof}[Second proof of the implication ``Deterministic conditions $\Rightarrow$ NA'' (Theorems \ref{th:090224a1} and~\ref{th:100224a1})]
We work in the market $(\mathbb B,\Y)$ with finite or infinite time horizon $T\in(0,\infty]$. At this point, recall Agreement~\ref{agr:050224a1}.
Suppose that there is an admissible strategy
$H\in L(\mathbb B,\Y)$
that realizes arbitrage.
By appropriate scaling of $H$ we can and will assume that $H$ is $1$-admissible.
Set $\zeta_0\triangleq\zeta(\Y)\wedge T$ and recall that $Y$ is stopped at time $\zeta_0$.
Therefore, the continuous process
$$
X_t\triangleq\int_0^t H_s\,d\Y_s,\quad t\in[0,T],
$$
is also stopped at $\zeta_0$, which entails that
$$
\P(X_{\zeta_0}\ge0)=1
\quad\text{and}\quad
\P(X_{\zeta_0}>0)>0.
$$
As $\{X_{\zeta_0}>0\}=\bigcup_{n\in\mathbb N}\{X_{\zeta_0}>1/n\}$,
there exists $\gamma>0$ such that $\P(X_{\zeta_0}>\gamma)>0$
and hence, $\P(\tau_\gamma<\zeta_0)>0$, where
$$
\tau_\gamma
\triangleq
\big(
\inf\{t\in[0,T] \colon X_t\ge\gamma\}
\big)\wedge\zeta_0.
$$
It follows that
\begin{equation}\label{eq:130224a1}
\P(X_{\tau_\gamma}\ge0)=1
\quad\text{and}\quad
\P(X_{\tau_\gamma}>0,\tau_\gamma<\zeta_0)>0.
\end{equation}
Notice that the deterministic conditions in Theorem~\ref{th:090224a1}
(resp., Theorem~\ref{th:100224a1})
are nothing else but
Condition~\ref{cond:170623a1} together with Condition~A
(resp., Condition~B)
from Lemma~\ref{lem:130224a1}.
By Lemma~\ref{lem:130224a1}, there exists a nonnegative c\`adl\`ag uniformly integrable $\mathbf F$-$\P$-martingale $\Z=(\Z_t)_{t\in[0,T]}$ with $\Z_0=1$, strictly positive on the stochastic interval $[0,\zeta_0)$, such that $\Z\Y$ is an $\mathbf F$-$\P$-local martingale on $[0,T]$.
Integration by parts as in \cite[I.4.45]{JS} yields that
$$
d\Z_t \Y_t
=
\Y_t\,d\Z_t + \Z_{t-}\,d\Y_t + d[\Z,\Y]_t
$$
(as $\Y$ is continuous, we do not write $\Y_{t-}$ in the first integrand).
Although, in general, a stochastic integral w.r.t. a c\`adl\`ag local martingale can fail to be a local martingale, it is always a local martingale provided the integrand is locally bounded
(\cite[I.4.31 and I.4.34~(b)]{JS}).
As $\Y$ is an $\mathbf F$-adapted continuous process, it is locally bounded.
Therefore, $\int_0^\cdot \Y_s\,d\Z_s$ is an $\mathbf F$-$\P$-local martingale.
Then, the process
$$
M\triangleq\int_0^\cdot \Z_{s-}\,d\Y_s+[\Z,\Y]
$$
is a continuous $\mathbf F$-$\P$-local martingale
(notice that the quadratic covariation $[\Z,\Y]$ is continuous because $\Y$ is continuous, see \cite[Theorem I.4.47~(c)]{JS}).
Using integration by parts once again, we compute
\begin{align*}
d\Z_t X_t
&=
X_t\,d\Z_t + \Z_{t-}\,dX_t + d[\Z,X]_t
\\
&=
X_t\,d\Z_t + H_t\Z_{t-}\,d\Y_t + H_t\,d[\Z,\Y]_t
\\
&=
X_t\,d\Z_t + H_t\,dM_t.
\end{align*}
This shows that $\Z X$ is a (c\`adl\`ag) $\mathbf F$-$\P$-local martingale on $[0,T]$.
As $H$ is $1$-admissible, we have $X\ge-1$, so $\Z(X+1)$ is a nonnegative $\mathbf F$-$\P$-local martingale, hence a nonnegative $\mathbf F$-$\P$-supermartingale.
Now, the fact that $\Z$ is a uniformly integrable $\mathbf F$-$\P$-martingale implies that
$$
\Z X = \Z(X+1) - \Z
$$
is an $\mathbf F$-$\P$-supermartingale which is closable at time~$T$ (the closability at $T$ is, of course, a nonvoid requirement only in the case $T=\infty$).
In particular, we can apply Doob's optional stopping theorem to the (in the case $T=\infty$, possibly, unbounded) stopping time $\tau_\gamma$ and obtain
$\E^\P[\Z_{\tau_\gamma}X_{\tau_\gamma}]\le\E^\P[\Z_0X_0]=0$.
On the other hand, by~\eqref{eq:130224a1} and the facts that $\Z$ is nonnegative on $[0,T]$ and strictly positive on $[0,\zeta_0)$, we get
$\E^\P[\Z_{\tau_\gamma}X_{\tau_\gamma}]>0$.
This contradiction proves the sufficiency of the deterministic conditions for NA.
\end{proof}

\begin{proof}[Proof of the implication ``NA $\Rightarrow$ deterministic conditions'' (Theorems \ref{th:090224a1} and~\ref{th:100224a1})]
In the following, we consider a finite or infinite time horizon \(T \in (0, \infty]\) and assume that the market $(\mathbb B,\Y)$ satisfies NA.
Set $T_0\triangleq T\wedge1$. 
Then, it is clear that NA holds in the market $(\mathbb B,\Y)$ also for the time horizon $T_0$.
Using Lemma~\ref{lem:210324a1}, it follows that NA holds in the market $(\mathbb B^\Y,\Y)$ with the time horizon $T_0$, and, 
by \cite[Proposition 10.38~(b)]{Jacod}, NA also holds in the canonical setup $(\mathbb C_{x_0},\canY)$ with the time horizon $T_0$.
Take any $T_1\in(0,T_0)$ and \(y_0 \in J^\circ\).
By Lemma~\ref{lem:060324a5}, NA holds for $(\mathbb C_{y_0},\canY)$ with the time horizon $T_1$.
Applying Theorem~\ref{theo: FTAP NA} for the market $(\mathbb C_{y_0},\canY)$ with the time horizon $T_1$ yields the existence of an ACLMM for this market.
The density process of this ACLMM w.r.t. $\P_{y_0}$ satisfies the assumption on the process $\canZ$ in Lemma~\ref{lem:060324a3}~(i) (with the starting point $x=y_0$ and the time horizon $T_1$).
As $y_0\in J^\circ$ is arbitrary, Lemma~\ref{lem:060324a3}~(ii) shows that Condition~\ref{cond:050324a1} is satisfied.
Lemma~\ref{lem:060324a4} upgrades Condition~\ref{cond:050324a1} to Condition~\ref{cond:170623a1}.

It remains to establish the properties of the boundary points.
As the market $(\mathbb B,\Y)$ with the (finite or infinite) time horizon $T$ satisfies NA, by Theorem~\ref{theo: FTAP NA}, there exists an ACLMM $\Q$ for the market $(\mathbb B,\Y)$.
As $\Y$ is a continuous $\mathbf F$-$\Q$-local martingale,
it is also an $\mathbf F^\Y$-$\Q$-local martingale.
By \cite[Theorem 10.37]{Jacod}, $\canY$ is a $\canbfF$-$\Q_{x_0}$-local martingale, where $\Q_{x_0}\triangleq\Q\circ\Y^{-1}$,
which is a probability measure on $(\canOmega,\canF)$
(the canonical setting with the time horizon $T$)
such that $\Q_{x_0}\ll\P_{x_0}$.
As Condition~\ref{cond:050324a1} is already established, we can apply
Lemma~\ref{lem:050324a1} with $\xi=T$ to get that $\Q_{x_0}=\tP_{x_0}$,
where $\tP_{x_0}$ is the law of the diffusion
on natural scale
started in $x_0$
with interior state space $J^\circ$,
speed measure $\s'\,d\m$ on $J^\circ$
and boundary points $l$ and $r$ being absorbing whenever accessible.
(Notice that $\s'$ exists everywhere on $J^\circ$, as we already proved that Condition~\ref{cond:050324a1} is satisfied.)
All in all, we obtain
\begin{equation}\label{eq:250324a1}
\tP_{x_0}\ll\P_{x_0}.
\end{equation}
In the case $T<\infty$, \eqref{eq:250324a1} and \cite[Corollary 2.15]{criensurusov22} yield the properties of the boundary points needed in Theorem~\ref{th:090224a1}.
In the case $T=\infty$, \eqref{eq:250324a1} and \cite[Corollary 2.16]{criensurusov22} yield the properties of the boundary points needed in Theorem~\ref{th:100224a1}.
This concludes the proof of the necessity of the deterministic conditions for NA.
\end{proof}

\subsection{Alternative proof of the characterizations of NFLVR} \label{sec: alternative NFKVR} \label{sec: pf alternativ cor:090224a1, cor:100224a1}

The deterministic characterizations of NFLVR, Corollaries \ref{cor:090224a1} and~\ref{cor:100224a1}, are, of course, immediate consequences of the deterministic characterizations of NUPBR and NA, which are proved above.
But there is a much more direct way of proving the deterministic characterizations of NFLVR.
Before we explain this in detail, let us recall the statements of the corollaries for the reader's convenience. 

\medskip\noindent
\textbf{Corollary~\ref{cor:090224a1}.}
{\em
	Assume that $T<\infty$.
	The financial market \((\mathbb{B}, \Y)\) satisfies NFLVR if and only if Condition~\ref{cond:170623a1} holds and, for every $b\in\{l,r\}\cap\bR$, at least one of conditions \textup{(A)--(B)} below is satisfied:
	\begin{enumerate}
		\item[\textup{(A)}]
		condition~\eqref{eq: b good} holds;
		\item[\textup{(B)}]
		the boundary point $b$ is inaccessible for $\Y$ and \eqref{eq: very good add cond} holds.
	\end{enumerate}}

\medskip\noindent
\textbf{Corollary~\ref{cor:100224a1}.}
{\em
	Suppose that $T=\infty$ and that $Y$ is \emph{not} on natural scale.
	The financial market \((\mathbb{B}, \Y)\) satisfies NFLVR if and only if Condition~\ref{cond:170623a1} holds and, for every $b\in\{l,r\}$,
	one of the following conditions \textup{(I)--(II)} is satisfied:
	\begin{enumerate}
		\item[\textup{(I)}]
		$|b|<\infty$ and \eqref{eq: b good} holds;
		\item[\textup{(II)}]
		$|b|=\infty$, $|\s(b)|=\infty$ and the other boundary point $b^*$ satisfies~\textup{(I)}.
	\end{enumerate}}

\smallskip 
In our previous paper~\cite{criensurusov22}, we studied the questions of (local) absolute continuity of laws of two general diffusions on the canonical space.
To illustrate the results in \cite{criensurusov22} we, essentially, prove the Corollaries~\ref{cor:090224a1} and~\ref{cor:100224a1} in the canonical setting (see \cite[Theorems~3.5 and 3.8 and Remark~3.10]{criensurusov22}).
Those proofs are much more direct than the proofs above because the existence of a ELMM is a global property directly related to the questions studied in \cite{criensurusov22} (which is in contrast to the existence of an SMD).

Thus, from \cite[Section~3]{criensurusov22} we know that the claims of Corollaries \ref{cor:090224a1} and~\ref{cor:100224a1} hold in the canonical setting $(\mathbb C_{x_0},\canY)$.
It is a natural and interesting question whether these results can be transferred directly from the canonical setting to our more general market $(\mathbb B,\Y)$.
The answer to this question is affirmative and it leads to standalone proofs of the characterizations of NFLVR in the sense that they require only Lemma~\ref{lem: markov martingale} and nothing else from Section~\ref{sec:proofs}.

\begin{proof}[Alternative proof of Corollaries \ref{cor:090224a1} and~\ref{cor:100224a1}]
The argumentation below applies both to a finite and to the infinite time horizon.
That is, we consider $T\in(0,\infty]$ and recall Agreement~\ref{agr:050224a1}.

\smallskip
\emph{``Deterministic conditions $\Rightarrow$ NFLVR'':}
The aim is to prove NFLVR in the (non-canonical) market $(\mathbb B,\Y)$, where $\mathbb{B} = (\Omega, \cF, \mathbf{F} = (\cF_t)_{t \in [0, T]}, \P)$ (just recalling the notation).
First, by \cite[Section~3]{criensurusov22},
the deterministic conditions in Corollary~\ref{cor:090224a1}
(resp., Corollary~\ref{cor:100224a1})
give us NFLVR in the canonical setting $(\mathbb C_{x_0},\canY)$
with $T<\infty$ (resp., $T=\infty$).
By Theorem~\ref{theo: FTAP NFLVR}, there exists an ELMM $\Q_{x_0}$ in this canonical setting.
Let $\canZ=(\canZ_t)_{t\in[0,T]}$ be the density process of $\Q_{x_0}$ w.r.t. $\P_{x_0}$, which is a strictly positive uniformly integrable c\`adl\`ag $\canbfF$-$\P_{x_0}$-martingale with $\E^{\P_{x_0}}[\canZ_t]\equiv1$, also with $\canZ_\infty>0$ $\P_{x_0}$-a.s. in the case $T=\infty$.
As $\canY$ is continuous, part (c) from \cite[Proposition~III.3.8]{JS} yields that
$\canZ\,\canY$ is an $\canbfF$-$\P_{x_0}$-local martingale.
By \cite[Theorem~10.37]{Jacod},
$\Z\triangleq\canZ(\Y)$ is an $\mathbf F^\Y$-$\P$-martingale and
$\Z\Y$ is an $\mathbf F^\Y$-$\P$-local martingale.
Further, Lemma~\ref{lem: markov martingale} yields that $\Z$ is an $\mathbf F$-$\P$-martingale and $\Z\Y$ is an $\mathbf F$-$\P$-local martingale.
From above, recall the other good properties of $\Z$, i.e., strict positivity, $\E^{\P}[\Z_t]\equiv1$, uniform integrability under~$\P$.
These enable us to define a probability measure $\Q\sim\P$ on $(\Omega,\cF)$ via the Radon-Nikodym density $d\Q/d\P\triangleq\Z_T$.
In other words, $\Z$ is the density process of $\Q$ w.r.t.~$\P$.
Applying \cite[Proposition~III.3.8]{JS} again, we obtain that $Y$ is an $\mathbf F$-$\Q$-local martingale, that is, $\Q$ is an ELMM.
Theorem~\ref{theo: FTAP NFLVR} yields NFLVR in the market $(\mathbb B,\Y)$.

\smallskip
\emph{``NFLVR $\Rightarrow$ deterministic conditions'':}
We assume that the market $(\mathbb B,\Y)$ with finite or infinite time horizon $T\in(0,\infty]$ satisfies NFLVR.
By Theorem~\ref{theo: FTAP NFLVR}, there exists an ELMM $\Q$ for this market $(\mathbb B,\Y)$.
As the $\mathbf F$-$\Q$-local martingale $\Y$ is continuous,
it is also an $\mathbf F^\Y$-$\Q$-local martingale.
Hence, by \cite[Theorem 10.37]{Jacod}, the measure $\Q\circ\Y^{-1}$ is an ELMM in the canonical market $(\mathbb C_{x_0},\canY)$ with the same time horizon $T$.
Again by Theorem~\ref{theo: FTAP NFLVR},
we have NFLVR in the canonical market $(\mathbb C_{x_0},\canY)$ with the same time horizon $T$.
Now, the deterministic conditions in Corollary~\ref{cor:090224a1}
(resp., Corollary~\ref{cor:100224a1})
follow from \cite[Theorem 3.5 and Remark 3.10]{criensurusov22}
(resp., \cite[Theorem 3.8 and Remark 3.10]{criensurusov22}).
This concludes the proof.
\end{proof}

\subsection{Proofs of Theorems~\ref{theo: NGA} and \ref{theo: NGA2}} \label{sec: theo: NGA, theo: NGA2}

For the reader's convenience, we repeat the formulations.

\medskip\noindent
\textbf{Theorem~\ref{theo: NGA}.}
{\em	Let $T<\infty$.
	There exists an EMM for the market \((\mathbb{B}, Y)\), i.e., a probability measure \(\Q \sim \P\) such that \(Y\) is an \(\mathbf{F}\)-\(\Q\)-martingale, if and only if \((\mathbb{B}, Y)\) satisfies NFLVR and,
	for every infinite boundary point \(b \in \{l,r\}\setminus\bR\),
	\[
	\int_B | x | \s' (x) \m (dx) = \infty
	\]
	holds for some (equivalently, for every) open interval \(B \subsetneq J^\circ\) with \(b\) as endpoint.}

\begin{proof}
If an EMM \(\Q\) exists, then NFLVR holds
for the market $(\mathbb B,\Y)$
by Theorem~\ref{theo: FTAP NFLVR}.
Thanks to the tower rule for conditional expectations,
the $\mathbf F$-$\Q$-martingale $\Y$ is also an $\mathbf F^\Y$-$\Q$-martingale.
Therefore,
the push-forward \(\Q \circ \Y^{-1}\) is an EMM for the canonical market \((\mathbb{C}_{x_0}, \canY)\).
By Lemma~\ref{lem:050324a1}
(which can be applied, as Condition~\ref{cond:050324a1} holds by Corollary~\ref{cor:090224a1}),
the only ELMM for this canonical market is the law \(\tP_{x_0}\) of a diffusion on natural scale with
the interior of the state space $J^\circ$,
starting value \(x_0\) and speed measure \(\s' d \m\) whose accessible boundaries are absorbing.
Thus, \(\tP_{x_0}\) is an EMM for the canonical market.
In particular,
the coordinate process \(\canY\) is a true \(\canbfF\)-\(\tP_{x_0}\)-martingale. Now, \cite[Theorem~1]{kotani} shows that every infinite boundary point \(b \in \{l, r\} \setminus \bR\) has to satisfy \(\int_B |x|\s' (x) \m (dx) =\infty\) for some (equivalently, for every) open interval \(B \subsetneq J^\circ\) with \(b\) as endpoint.

Conversely, assume that NFLVR holds for $(\mathbb B,\Y)$
and that every infinite boundary point \(b \in \{l, r\} \setminus \bR\) satisfies \(\int_B |x|\s' (x) \m (dx) =\infty\) for some (equivalently, for every) open interval \(B \subsetneq J^\circ\) with~\(b\) as endpoint.
As in the proof of the direction ``NFLVR \(\Rightarrow\) deterministic conditions'' from Section~\ref{sec: alternative NFKVR}, NFLVR holds also for the canonical market $(\mathbb C_{x_0},\canY)$ with the same time horizon~$T$.
By Lemma~\ref{lem:050324a1} (which can be applied, as Condition~\ref{cond:050324a1} holds by Corollary~\ref{cor:090224a1}),
the only ELMM for this canonical market is the measure \(\tP_{x_0}\) introduced above.
But thanks to \cite[Theorem~1]{kotani}, $\tP_{x_0}$ is even an EMM for the canonical market \((\mathbb{C}_{x_0}, \canY)\).
As in the proof of the direction ``deterministic conditions \(\Rightarrow\) NFLVR'' from Section~\ref{sec: alternative NFKVR}, it follows from Lemma~\ref{lem: markov martingale} that an EMM for the market \((\mathbb{B}, Y)\) exists. The proof is complete. 
\end{proof}

\medskip\noindent
\textbf{Theorem~\ref{theo: NGA2}.}
{\em	Let $T=\infty$.
	There exists an EMM for the market \((\mathbb{B}, Y)\), i.e., a probability measure \(\Q \sim \P\) such that \(Y\) is a uniformly integrable \(\mathbf{F}\)-\(\Q\)-martingale, if and only if NFLVR holds and both boundaries \(l\) and \(r\) are finite, i.e., \(l, r \in \bR\).}
	
\begin{proof}
First, suppose that NFLVR holds and that \(J\) is bounded. Then, by Theorem~\ref{theo: FTAP NFLVR}, there exists an ELMM and, as \(J\) is bounded, any ELMM is already an EMM.

Conversely, assume that an EMM \(\Q\) exists. Then, NFLVR holds by Theorem~\ref{theo: FTAP NFLVR}. It remains to prove that \(J\) needs to be bounded. 
As in the proof of Theorem~\ref{theo: NGA} above, it follows from Lemma~\ref{lem:050324a1} that the push-forward
\(\Q\circ \Y^{-1}\)
is the law of a diffusion on natural scale with the interior of the state space $J^\circ$
whose accessible boundaries are absorbing.
As \(\Y\) is a UI \(\Q\)-martingale, the UI martingale convergence theorem implies that \(\Y_t\) converges \(\Q\)-a.s. and in \(L^1(\Q)\) as \(t \to \infty\). If \(l = - \infty\) and \(r = \infty\), it follows from \cite[Proposition~5.5.22]{KaraShre} (or \cite[Lemma~B.7]{criensurusov22}) that \(\Y_t\) does not converge \(\Q\)-a.s. as \(t \to \infty\). Therefore, at least one of the boundaries has to be finite. In case one of the boundaries \(l\) and \(r\) is infinite and the other one, call it \(b\), is finite, then \cite[Proposition~5.5.22]{KaraShre} (or \cite[Lemma~B.7]{criensurusov22}) implies that \(\Q\)-a.s. \(\Y_\infty \triangleq \lim_{t \to \infty} \Y_t = b\). However, this is a contradiction to the convergence in \(L^1(\Q)\),
as \(\E^\Q[\Y_\infty] = b \ne x_0 = \lim_{t \to \infty} \E^{\Q}[\Y_t]\). We conclude that \(J\) has to be bounded. 
\end{proof}

\bibliographystyle{plain}

\end{document}